\documentclass[sigconf,nonacm]{acmart}
\usepackage[vlined,boxed]{algorithm2e}
\usepackage{amsmath}
\usepackage{amsfonts}
\usepackage{epsfig}
\usepackage{graphics}
\usepackage{color}
\usepackage{comment}
\usepackage{paralist}
\usepackage{mathtools}
\usepackage{multirow}
\usepackage{subfigure}
\usepackage{eepic}
\usepackage{stmaryrd}
\usepackage{textcomp}
\usepackage{balance}
\usepackage[inline]{enumitem}
\usepackage{ulem}
\usepackage{arydshln}
\usepackage{bigdelim}
\usepackage{makecell}
\usepackage{todonotes}
\usepackage{amsmath}
\usepackage{tikz}\usetikzlibrary{arrows,fit,shapes.symbols,patterns}

\begin{document}

\newcommand{\OMIT}[1]{}

\newcommand{\naf}{\neg}
\newcommand{\sms}[1]{{\sf sms}(#1)}
\newcommand{\flip}{{\sf Flip}}
\newcommand{\fin}{\mathsf{fin}}
\newcommand{\SM}{\mathsf{SM}}
\newcommand{\repl}{\mathsf{repl}}
\newcommand{\params}[1]{\langle #1 \rangle}
\newcommand{\terminals}[1]{\mathsf{terminals}(#1)}
\newcommand{\comp}{\hookrightarrow}
\newcommand{\paths}[1]{\mathsf{paths}(#1)}
\newcommand{\dg}[1]{\mathsf{dg}(#1)}
\newcommand{\scc}[1]{\mathsf{scc}(#1)}

\newcommand{\edb}[1]{\mathsf{edb}(#1)}
\newcommand{\idb}[1]{\mathsf{idb}(#1)}

\newcommand{\full}[1]{\mathsf{full}(#1)}
\newcommand{\heads}[1]{\mathsf{heads}(#1)}
\newcommand{\pos}[1]{\mathsf{pos}(#1)}
\newcommand{\ground}[1]{\mathsf{ground}(#1)}

\newcommand{\op}{\mathit{op}}
\newcommand{\PS}{\mathcal{P}}
\newcommand{\viol}[2]{\mathsf{V}(#1,#2)}
\newcommand{\rs}[2]{\mathsf{RS}(#1,#2)}
\newcommand{\crs}[2]{\mathsf{CRS}(#1,#2)}
\newcommand{\crss}[3]{\mathsf{CRS}_{#3}(#1,#2)}
\newcommand{\cancrs}[2]{\mathsf{CanCRS}(#1,#2)}
\newcommand{\cancrss}[3]{\mathsf{CanCRS}_{#3}(#1,#2)}
\newcommand{\opr}[2]{\mathsf{ORep}(#1,#2)}
\newcommand{\ops}[3]{\mathsf{Ops}(#1,#2,#3)}
\newcommand{\abs}[1]{\mathsf{abs}_{>0}(#1)}
\renewcommand{\abs}[1]{\mathsf{RAS}(#1)}
\newcommand{\insP}{\ins{P}}
\newcommand\sem[1]{{[\![ #1 ]\!]}}
\newcommand{\probrep}[2]{\mathsf{P}_{#1}(#2)}
\newcommand{\oca}[2]{\mathsf{OCA}_{#2}(#1)}
\newcommand{\ocqa}[1]{\mathsf{OCQA}(#1)}
\newcommand{\ur}{\mathsf{ur}}
\newcommand{\us}{\mathsf{us}}
\newcommand{\uo}{\mathsf{uo}}

\newcommand{\mi}[1]{\mathit{#1}}
\newcommand{\ins}[1]{\mathbf{#1}}
\newcommand{\adom}[1]{\mathsf{dom}(#1)}
\renewcommand{\paragraph}[1]{\textbf{#1}}
\newcommand{\ra}{\rightarrow}
\newcommand{\fr}[1]{\mathsf{fr}(#1)}
\newcommand{\dep}{\Sigma}
\newcommand{\sch}[1]{\mathsf{sch}(#1)}
\newcommand{\ar}[1]{\mathsf{ar}(#1)}
\newcommand{\body}[1]{\mathsf{body}(#1)}
\newcommand{\head}[1]{\mathsf{head}(#1)}
\newcommand{\guard}[1]{\mathsf{guard}(#1)}
\newcommand{\class}[1]{\mathbb{#1}}
\newcommand{\app}[2]{\langle #1,#2 \rangle}
\newcommand{\crel}[1]{\prec_{#1}}

\newcommand{\ccrel}[1]{\prec_{#1}^+}

\newcommand{\tcrel}[1]{\prec_{#1}^{\star}}
\newcommand{\rctaa}{\class{CT}_{\forall \forall}^{\mathsf{res}}}
\newcommand{\rctaapr}{\mathsf{CT}_{\forall \forall}^{\mathsf{res}}}
\newcommand{\rctae}{\class{CT}_{\forall \exists}^{\mathsf{res}}}
\newcommand{\rctaepr}{\mathsf{CT}_{\forall \exists}^{\mathsf{res}}}
\newcommand{\base}[1]{\mathsf{base}(#1)}
\newcommand{\eqt}[1]{\mathsf{eqtype}(#1)}
\newcommand{\result}[1]{\mathsf{result}(#1)}
\newcommand{\chase}[2]{\mathsf{ochase}(#1,#2)}
\newcommand{\pred}[1]{\mathsf{pr}(#1)}
\newcommand{\origin}[1]{\mathsf{org}(#1)}
\newcommand{\eq}[1]{\mathsf{eq}(#1)}
\newcommand{\dept}[1]{\mathsf{depth}(#1)}

\newcommand{\rep}[2]{\mathsf{rep}_{#2}(#1)}
\newcommand{\repp}[2]{\mathsf{rep}_{#2}\left(#1\right)}
\newcommand{\rfreq}[2]{\mathsf{rfreq}_{#2}(#1)}
\newcommand{\homs}[3]{\mathsf{hom}_{#2,#3}(#1)}
\newcommand{\prob}[1]{\mathsf{#1}}
\newcommand{\key}[1]{\mathsf{key}(#1)}
\newcommand{\keyval}[2]{\mathsf{key}_{#1}(#2)}
\newcommand{\block}[2]{\mathsf{block}_{#2}(#1)}
\newcommand{\sblock}[2]{\mathsf{sblock}_{#2}(#1)}

\newcommand{\rt}[1]{\mathsf{root}(#1)}
\newcommand{\child}[1]{\mathsf{child}(#1)}

\newcommand{\var}[1]{\mathsf{var}(#1)}
\newcommand{\const}[1]{\mathsf{const}(#1)}
\newcommand{\pvar}[2]{\mathsf{pvar}_{#2}(#1)}

\newcommand{\att}[1]{\mathsf{att}(#1)}
\newcommand{\card}[1]{\sharp #1}

\newcommand{\pr}{\mathsf{Pr}}
\newcommand{\prsp}{\mathsf{PS}}

\newcommand{\sign}[1]{\mathsf{sign}(#1)}
\newcommand{\litval}[2]{\mathsf{lval}_{#2}(#1)}
\newcommand{\angletup}[1]{\langle #1 \rangle}
%%%%%%%%%%%%%%%%%%%%%%%%% Environment delimiters

\def\qed{\hfill{\qedboxempty}      % qed with empty box
  \ifdim\lastskip<\medskipamount \removelastskip\penalty55\medskip\fi}

\def\qedboxempty{\vbox{\hrule\hbox{\vrule\kern3pt
                 \vbox{\kern3pt\kern3pt}\kern3pt\vrule}\hrule}}

\def\qedfull{\hfill{\qedboxfull}   % qed with full box
  \ifdim\lastskip<\medskipamount \removelastskip\penalty55\medskip\fi}

\def\qedboxfull{\vrule height 4pt width 4pt depth 0pt}

\newcommand{\markfull}{\qedboxfull}
\newcommand{\markempty}{\qed}

\newtheorem{claim}[theorem]{Claim}
\newtheorem{fact}[theorem]{Fact}
\newtheorem{observation}{Observation}
\newtheorem{remark}{Remark}
\newtheorem{apptheorem}{Theorem}[section]
\newtheorem{appcorollary}[apptheorem]{Corollary}
\newtheorem{appproposition}[apptheorem]{Proposition}%[section]
\newtheorem{applemma}[apptheorem]{Lemma}%[section]
\newtheorem{appclaim}[apptheorem]{Claim}%[section]
\newtheorem{appfact}[apptheorem]{Fact}%[section]

\fancyhead{}

\title{Generative Datalog with Stable Negation}

\author{Mario Alviano}
\affiliation{%
	\institution{University of Calabria}
	\country{}
}
\email{alviano@mat.unical.it }

\author{Matthias Lanzinger}
\affiliation{%
	\institution{University of Oxford}
		\country{}
}
\email{matthias.lanzinger@cs.ox.ac.uk}

\author{Michael Morak}
\affiliation{%
	\institution{University of Klagenfurt}
	\country{}
}
\email{michael.morak@aau.at}

\author{Andreas Pieris}
\affiliation{%
	\city{University of Edinburgh \& University of Cyprus}
		\country{}
}
\email{apieris@inf.ed.ac.uk}

\begin{abstract} 
	Extending programming languages with stochastic behaviour such as
	probabilistic choices or random sampling has a long tradition in
	computer science. A recent development in this direction is a
	declarative probabilistic programming language, proposed by B\'ar\'any
	et al. in 2017, which operates on standard relational databases. In
	particular, B\'ar\'any et al. proposed generative Datalog, a
	probabilistic extension of Datalog that allows sampling from discrete
	probability distributions. Intuitively, the output of a generative
	Datalog program $\Pi$ on an input database $D$ is a probability space
	over the minimal models of $D$ and $\Pi$, the so-called possible
	outcomes. This is a natural generalization of the (deterministic)
	semantics of Datalog, where the output of a program on a database is
	their unique minimal model. A natural question to ask is how generative
	Datalog can be enriched with the useful feature of negation, which in
	turn leads to a strictly more expressive declarative probabilistic
	programming language. In particular, the challenging question is how the
	probabilistic semantics of generative Datalog with negation can be
	robustly defined. Our goal is to provide an answer to this question by
	interpreting negation according to the stable model semantics.
	%coming from logic programming.
	%
	%The goal of this work is to provide an answer to this question by interpreting negation according to the standard stable model semantics coming from logic programming.
\end{abstract}

\maketitle

\section{Introduction}\label{sec:introduction}

Extending programming languages with stochastic behaviour, such as probabilistic
choices or random sampling, has a long tradition in computer science that goes
back in the 70s and 80s~\cite{Kozen81,Saheb-Djahromi78}.
During the last decade or so, there was a significant effort in developing dedicated probabilistic programming languages (for example, Pyro~\cite{BCJOPKSSH19}, Stan~\cite{JSSv076i01}, Church~\cite{GMRBT08}, R2~\cite{NHRS14}, Figaro~\cite{Pfeffer}, and Anglican~\cite{ToMW15}, to name a few) that allow the specification and ``execution'', via probabilistic inference, of sophisticated probabilistic models.

A recent development in this direction is a purely declarative probabilistic programming language based on Datalog (PPDL), proposed by B\'ar\'any et al.~\cite{BCKOV17}, which operates on standard relational databases. 
%PPDL is indeed relational and declarative in nature, and at the same time supports the main features of probabilistic programming languages such as random sampling and conditioning on observations.
%
A PPDL program consists of the generative and the constraint component.
The generative component is essentially a rule-based program written in a probabilistic extension of Datalog, called generative Datalog, that allows sampling from discrete probability distributions. Intuitively, the output of a generative Datalog program $\Pi$ on an input database $D$ is a probability space over the minimal models of $D$ and $\Pi$,
%(in fact, over the minimal models of $D$ and a first-order sentence $\dep_\Pi$ that somehow captures the intended meaning of $\Pi$), 
the so-called possible outcomes. Typically, the probability space over the possible outcomes obtained via such a generative process is called prior distribution. The probabilistic semantics of generative Datalog can be seen as a natural generalization of the (deterministic) semantics of Datalog, where the output of a program on a database is their unique minimal model. 
Now, the constraint component allows
one to pose logical constraints that the relevant possible outcomes should satisfy and, semantically,
transforms the prior distribution into 
%the so-called posterior distribution, that is, 
the subspace obtained by conditioning on the constraints.

Generative Datalog is indeed a powerful probabilistic formalism that allows for defining a rich family of probabilistic models, while retaining the key property of declarativity, that is, the independence of the order in which the rules are executed. On the other hand, as discussed below, there are simple probabilistic models that cannot be defined using generative Datalog due to its monotonic nature. 
%Let us illustrate this via the following simple example:
%, which will serve as our running example throughout the paper.

%it suffers from an obvious weakness, that is, it does not support the useful feature of negation. 

\begin{example}[\textbf{Network Resilience}]\label{exa:netwrok-resilience}
	We consider a network of routers, and assume that some of them have been
	infected by a malware that interrupts the routing, and it also attempts
	to infect neighbouring routers with a success rate of 10\%. We say that
	such a network is dominated by the malware if all routers are infected
	or isolated in the network (i.e., connected only to infected routers).
	In such a scenario, we would like to be able to predict how likely it is
	for the network to be dominated.
	%
	%One may think that this can be easily achieved using generative Datalog. Indeed,
	The propagation of the malware can be encoded using the generative Datalog rule
	\[
	\text{\rm Infected}(x,\mi{true}), \text{\rm Connected}(x,y) \ra \text{\rm Infected}(y,\flip\params{0.1})
	\]
	meaning that if the router $x$ is infected and $x$ is connected with
	$y$, then $y$ is also infected with probability $0.1$; that is, the
	probabilistic term $\flip\params{0.1}$ becomes $\mi{true}$ with
	probability $0.1$, or $\mi{false}$ with probability $0.9$.
	%(that is,
	%$\text{\rm Infected}(y,\mi{true})$ is derived with probability $0.1$,
	%and $\text{\rm Infected}(y,\mi{false})$ is derived with probability
	%$0.9$).
	However, the domination of the network is a non-monotonic
	property in the sense that expanding a dominated network may result in a
	non-dominated one, whereas generative Datalog is inherently monotonic.
	Hence, although the propagation of the malware can be encoded via
	generative Datalog, the likelihood of domination cannot be predicted via
	this formalism. \hfill\markfull
\end{example}

With the aim of overcoming the above weakness of generative Datalog, we enrich
it with the useful feature of \emph{negation as failure}, interpreted according to the
standard stable model semantics stemming from logic programming~\cite{GeLi88}.
Having such a formalism in place, we could predict the likelihood of a network
being dominated using the probabilistic rule from
Example~\ref{exa:netwrok-resilience}, and the rules
\begin{align*}
&\text{\rm Router}(x), \naf \text{\rm Infected}(x,\mi{true}) \ra \textrm{Uninfected}(x)\\
&\text{\rm Uninfected}(x), \text{\rm Uninfected}(y), \text{\rm Connected}(x,y) \ra \text{\rm Fail},
\end{align*}
which essentially state that a router which is {\em not} infected is uninfected, and whenever two uninfected routers are connected in the network, then the malware fails to dominate the network.

\medskip
\noindent \paragraph{Our Contributions.}
Adding stable negation to generative Datalog comes with several technical
complications, discussed in Section~\ref{sec:gen-datalog}, that need to be
studied and understood in order to define a robust semantics. This is precisely
the goal of the present work. Our main contributions can be summarized as
follows.
%\begin{itemize}
	%\item 

In Section~\ref{sec:gen-datalog}, we introduce a formal semantics for generative Datalog with stable negation. 
%In contrast to what one may think, 
This
%is not done by merely
cannot simply be done by following the development underlying the semantics of
positive generative Datalog from~\cite{BCKOV17} and considering stable models
instead of minimal models.
%It is conceptually more accurate to think of
%possible outcomes as {\em sets} of stable models; in fact, as {\em ground
%programs} that symbolically represent sets of stable models, and they also
%explicitly encode the underlying probabilistic choices.
In fact, with negation, it is not enough to consider single models; we need to
thing of possible outcomes as \emph{sets} of stable models. These sets are
represented symbolically via \emph{ground programs} that also explicitly encode
the underlying probabilistic choices.
This last property is crucial
since, due to non-monotonicity,
%nature of generative Datalog with negation, 
the probabilistic choices are not necessarily reflected in the resulting set of
stable models.
%This, in turn, is problematic since
Hence, a set of stable models may not
carry enough information to allow us to calculate its probability.
	
%\item 
In Section~\ref{sec:chase}, we show that the probabilistic semantics from
Section~\ref{sec:gen-datalog} can be equivalently defined via a fixpoint
procedure. We introduce a new chase procedure that operates on ground programs
(rather than on databases, as in the standard chase procedure), which we then use to
define a chase-based probability space that leads to the desired fixpoint
semantics. Interestingly, the semantics is independent of the
%choice of the
order in which the rules are executed.
	
%\item 
Finally, in Section~\ref{sec:stratified-negation}, we concentrate on the central class of generative Datalog programs with {\em stratified negation}, and we provide further evidence of the robustness of our semantics.
%\end{itemize}

A discussion on related work can be found in the appendix.

\section{Preliminaries}\label{sec:preliminaries}

We recall the basics on relational databases, TGDs with stable negation, and probability spaces.
We assume the disjoint countably infinite sets $\ins{C}$ and $\ins{V}$ of {\em constants} and {\em variables}, respectively; we refer to constants and variables as {\em terms}. For $n > 0$, let $[n] = \{1,\ldots,n\}$.

\medskip

\noindent\paragraph{Relational Databases.}
A {\em (relational) schema} $\ins{S}$ is a finite set of relation names (or predicates) with associated arity; we write $\ar{R}$ for the arity of a predicate $R$.
A {\em (relational) atom} $\alpha$ over $\ins{S}$ is an expression of the form $R(t_1,\ldots,t_n)$, where $R \in \ins{S}$, $\ar{R} = n$, and $t_i \in \ins{C} \cup \ins{V}$ for each $i \in [n]$.
A {\em literal} over $\ins{S}$ is either an atom over $\ins{S}$ (i.e., a positive literal), or an atom over $\ins{S}$ preceded by the negation symbol $\naf$ (i.e., a negative literal).
An {\em instance} of $\ins{S}$ is a (possibly infinite) set of atoms over $\ins{S}$ that mention only constants, while a {\em database} of $\ins{S}$ is a {\em finite} instance of $\ins{S}$.
We write $\adom{I}$ for the set of constants occurring in an instance $I$.
%
%For a set $S$ of predicates, we write $I_{\mid S}$ for the instance obtained by restricting the instance $I$ to the predicates in $S$, i.e., the instance $\{R(\bar t) \in I \mid R \in S\}$.
%
%For a set $\mathcal{I}$ of instances, we write $\mathcal{I}^\fin$ for the set of databases (or finite instances) in $\mathcal{I}$, and $\mathcal{I}^\infty$ for the set of instances $\mathcal{I} \setminus \mathcal{I}^\fin$.
%
For simplicity, in the rest of the paper, we assume that the constants of $\ins{C}$ are translatable into real numbers. Therefore, whenever we refer to a constant of $\ins{C}$, we actually refer to its translation into a real number. 
%This essentially means that we are dealing with instances over the real numbers; in other words, for an instance $I$, it is always the case that $\adom{I} \subseteq \mathbb{R}$.

\OMIT{
\medskip
\noindent\paragraph{Datalog with Stable Negation.}
%
%A {\em literal} is either an atom (i.e., a positive literal), or an atom preceded by the negation symbol $\naf$ (i.e., a negative literal).
%
A {\em Datalog rule with negation} (or simply {\em Datalog$^\naf$ rule}) $\rho$ over a schema $\ins{S}$ is an expression
\[
R_1(\bar u_1),\ldots,R_n(\bar u_n), \naf P_1(\bar v_1),\ldots,\naf P_m(\bar v_m) \ra R_0(\bar w)
\]
for $n,m \geq 0$, where $R_0,R_1,\ldots,R_n,P_1,\ldots,P_m$ are predicates of $\ins{S}$, $\bar u_1,\ldots,\bar u_n$ are tuples over $\ins{C}  \cup \ins{V}$, and $\bar v_1,\ldots,\bar v_m,\bar w$ are tuples over $\ins{C} \cup \ins{V}$ such that each of their variables is mentioned in $\bar u_j$ for some $j \in [n]$; the latter condition is known as {\em safety}. 
The atom that appears on the right of the $\ra$ symbol is the {\em head} of $\rho$, denoted $H(\rho)$, while the expression on the left of the $\ra$ symbol is the {\em body} of $\rho$, denoted $B(\rho)$. 
We further write $B^+(\rho)$ for the set of positive literals in $B(\rho)$, and $B^-(\rho)$ for the set of atoms appearing in the negative literals of $B(\rho)$.
A {\em Datalog$^\naf$ program} (resp., {\em infinitary Datalog$^\naf$ program}) $\Pi$ over $\ins{S}$ is a finite (resp., infinite) set of Datalog$^\naf$ rules over $\ins{S}$.\footnote{Of course, in real-life scenarios, we are only interested in finite programs. Having said that, infinitary programs are a very useful technical tool.} 
A predicate $R \in \ins{S}$ occurring in $\Pi$ is called {\em extensional} if there is no rule of the form ${\rm body} \ra R(\bar w)$ in $\Pi$, i.e., $R$ occurs only in rule-bodies, and {\em intensional} if there exists at least one rule of the form ${\rm body} \ra R(\bar w)$ in $\Pi$, that is, $R$ appears in the head of at least one rule of $\Pi$. The {\em extensional (database) schema} of $\Pi$, denoted $\edb{\Pi}$, consists of the extensional predicates in $\Pi$, while the {\em intensional schema} of $\Pi$, denoted $\idb{\Pi}$, is the set of all intensional predicates in $\Pi$. The {\em schema} of $\Pi$, denoted $\sch{\Pi}$, is the set of predicates $\edb{\Pi} \cup \idb{\Pi}$. 
%
%We write $\adom{\Pi}$ for the set of constants of $\ins{C}$ occurring in $\Pi$.
%
%Recall our simplifying assumption that whenever we refer to a constant we actually refer to its translation into a real number. Thus, Datalog$^\naf$ rules actually mention variables and real numbers, which means that $\adom{\Pi} \subseteq \mathbb{R}$.
%
We write $\pos{\Pi}$ for the program obtained from $\Pi$ by removing all
%the
negative literals, i.e., $\pos{\Pi} = \{B^+(\rho) \ra H(\rho) \mid \rho \in
\Pi\}$. Let $\heads{\Pi}$ be the set of head atoms of $\Pi$.

In this work, we interpret negation according to the stable model semantics~\cite{GeLi88}. To formally define the adopted semantics, we need to recall when an instance is a model of an (infinitary) Datalog$^\naf$ program $\Pi$.
%and the notion of grounding for programs.
%
Given a Datalog$^\naf$ rule $\rho$ of the form
\[
R_1(\bar u_1),\ldots,R_n(\bar u_n), \naf P_1(\bar v_1),\ldots,\naf P_m(\bar v_m) \ra R_0(\bar w)
\]
we write $\dep_\rho$ for the first-order sentence
\begin{multline*}
\forall x_1 \cdots \forall x_n (R_1(\bar u_1) \wedge \ldots \wedge R_n(\bar u_n) \wedge\\ 
\neg P_1(\bar v_1)\wedge \ldots \wedge \neg P_m(\bar v_m) \ra R_0(\bar w)),
\end{multline*}
where $x_1,\ldots,x_n$ are all the variable occurring in $\rho$. In case $B(\rho)$ is empty, which means that $\bar w$ consists only of constants due to the safety condition, $\dep_\rho$ is the sentence $\text{\rm True} \ra R(\bar w)$.
For an (infinitary) Datalog$^\naf$ program $\Pi$ over a schema $\ins{S}$, let $\dep_\Pi$ be the first-order sentence $\bigwedge_{\rho \in \Pi} \dep_\rho$. We say that an instance $I$ of $\ins{S}$ is a {\em model} of $\Pi$ if it is a (classical) model of the sentence $\dep_\Pi$; note that here we adopt the {\em unique-name assumption}, that is, every constant of $\ins{C}$ refers to a distinct object.

We further need to recall the notion of grounding for Datalog$^\naf$ programs.
A {\em ground instance} of a Datalog$^\naf$ rule $\rho$ is obtained by replacing each variable in $\rho$ with a constant of $\ins{C}$. 
%of $\adom{\Pi}$. 
Given a Datalog$^\naf$ program $\Pi$, the {\em grounding of $\Pi$}, denoted $\ground{\Pi}$, is the infinitary program consisting of all the ground instances of the rules of $\Pi$.

Consider now an (infinitary) Datalog$^\naf$ program $\Pi$. Given an instance $I$ of $\sch{\Pi}$, the {\em reduct of $\Pi$ w.r.t. $I$}, denoted $\Pi^I$, is obtained from $\ground{\Pi}$ by first removing all the rules with a false negative literal, and then eliminating the remaining negative literals, that is,
\[
\Pi^I\ =\ \pos{\{\rho \in \ground{\Pi} \mid B^-(\rho) \cap I = \emptyset\}}.
\]
We say that $I$ is a {\em stable model of $\Pi$} if $I$ is a model of $\Pi$, and there is no instance $J \subsetneq I$ such that $J$ is model of $\Pi^I$.
Intuitively, $I$ is a stable model of $\Pi$ if it can be obtained by ``executing'' $\Pi$ using $I$ as an oracle for the negative literals. We write $\sms{\Pi}$ for the set of stable models of $\Pi$.
Finally, given a database $D$ of $\edb{\Pi}$, an instance $I$ of $\sch{\Pi}$ is a {\em stable model of $D$ and $\Pi$} if $I \in \sms{\Pi[D]}$ with $\Pi[D] = \{ \ra \alpha \mid \alpha \in D\} \cup \Pi$. We write $\sms{D,\Pi}$ for the set of stable models of $D$ and $\dep$, i.e., $\sms{D,\dep} = \sms{\Pi[D]}$.
}

%Note that all the above notations and definitions, although are given for Datalog$^\naf$ programs defined as finite sets, they can be also applied to {\em infinitary programs} consisting of infinitely many rules.

\medskip
\noindent\paragraph{TGDs with Stable Negation.} %Towards a probabilistic extension of Datalog with stable negation, 
%In our development, we will exploit tuple-generating dependencies with stable negation.
%Therefore, we proceed to recall the syntax and semantics of TGDs with stable negation according to~\cite{AlMP17}.
%
A {\em tuple-generating dependency with negation} (or simply TGD$^\naf$) $\sigma$ over a schema $\ins{S}$ is an expression
\[
\forall \bar x \forall \bar y \, (\phi(\bar x, \bar y)\ \ra\ \exists \bar z \, \psi(\bar x, \bar z)),
\]
where $\phi(\bar x,\bar y)$ (resp., $\psi(\bar x,\bar z)$) is a quantifier-free conjunction of literals (resp., atoms) over $\ins{S}$ with variables from $\bar x \cup \bar y$ (resp., $\bar x \cup \bar z$), and every variable in a negative literal occurs also in a positive literal; the latter condition is known as {\em safety}.
For brevity, we will write $\sigma$ as $\phi(\bar x, \bar y) \ra \exists \bar z \, \psi(\bar x, \bar z)$, and use comma instead of $\wedge$ for joining atoms.
The formula $\psi(\bar x,\bar z)$ is the {\em head} of $\sigma$, denoted $H(\sigma)$, whereas the formula $\phi(\bar x,\bar y)$ is the {\em body} of $\sigma$, denoted $B(\sigma)$. We further write $B^+(\sigma)$ for the set of positive literals in $B(\sigma)$, and $B^-(\sigma)$ for the set of atoms appearing in the negative literals of $B(\sigma)$. 
A {\em TGD$^\naf$ program} (resp., {\em infinitary TGD$^\naf$ program}) $\dep$ over $\ins{S}$ is a finite (resp., infinite) set of TGD$^\naf$ over $\ins{S}$. The {\em schema} of $\dep$, denoted $\sch{\dep}$, is the set of predicates in $\dep$, which is always finite, even if $\dep$ is
%an infinitary program.
infinitary.
%
%We write
Let $\heads{\dep}$ be the set of head atoms of $\dep$.

For our purposes, although syntactically a TGD$^\naf$ is a first-order sentence, semantically is not since the negation should be interpreted as stable negation rather than classical negation. We proceed to recall the formal definition as given in~\cite{AlMP17}. The high-level idea is, given an (infinitary) TGD$^\naf$ program $\dep$, to construct a second-order sentence $\SM[\dep]$ that encodes the intention underlying the stable model semantics, and then define the stable models of $\dep$ as the (classical) models of $\SM[\dep]$.
\OMIT{
At this point, one may think that we could simply eliminate the existentially quantified variables in $\dep$ via Skolemization, which in turn leads to a logic program with function symbols, and then adopt the standard stable model semantics for logic programs 
%(a generalization of what we have seen above for Datalog$^\naf$)
as done in~\cite{GHKL21}. However, despite the usefulness of this approach, as thoroughly discussed in~\cite{AlMP17}, it fails to fully capture the intended meaning of TGDs with stable negation.
}
Consider an (infinitary) TGD$^\naf$ program $\dep$. Let $\ins{R} = (R_1,\ldots,R_n)$ be the list of predicates of $\sch{\dep}$ in some fixed order (e.g., lexicographic order), and $\ins{X}_\dep = (X_{\dep}^{1},\ldots,X_{\dep}^{n})$ a list of $n$ distinct predicate variables. For a literal $\ell$ occurring in $\dep$, let
\begin{align*}
\repl_{\dep}(\ell)\ =\
\begin{cases}
X_{\dep}^{i}(\bar{t}), & \mbox{if } \ell = R_i(\bar{t}) \\
%& \\
\neg R_i(\bar{t}), & \mbox{if } \ell = \neg R_i(\bar{t}).
\end{cases}
\end{align*}
We define $\repl_{\dep}(\dep)$ as the (infinitary) TGD$^\naf$ program obtained by applying $\repl_\dep(\cdot)$ to every literal in $\dep$. The sentence $\SM[\dep]$ is
\[
 \bigwedge_{\sigma \in \dep} \sigma\ \wedge\ \neg \exists \ins{X}_{\dep} \left( (\ins{X}_\dep < \ins{R}) \wedge \repl_{\dep}(\dep)\right),
\]
%\[
%\bigwedge_{\alpha \in D} \alpha\ \wedge\ \bigwedge_{\sigma \in \dep} \sigma\ \wedge\ \neg \exists \ins{X}_{\dep} \left( (\ins{X}_\dep < \ins{R}) \wedge \repl_{\dep}(D) \wedge \repl_{\dep}(\dep)\right),
%\]
where $\ins{X}_\dep < \ins{R}$ denotes the sentence
\begin{multline*}
\bigwedge_{i=1}^{n} \forall x_1 \cdots \forall x_{\ar{R_i}} \left(X_{\dep}^{i}(x_1,\ldots,x_{\ar{R_i}}) \ra R_i(x_1,\ldots,x_{\ar{R_i}})\right) \wedge\\
\neg\left(\bigwedge_{i=1}^{n} \forall x_1 \cdots \forall x_{\ar{R_i}} \left(R_i(x_1,\ldots,x_{\ar{R_i}}) \ra X_{\dep}^{i}(x_1,\ldots,x_{\ar{R_i}})\right)\right).
\end{multline*}
An instance $I$ of $\sch{\dep}$ is a {\em stable model of $\dep$} if it is a (classical) model of $\SM[\dep]$; note that here we adopt the {\em unique-name assumption}, that is, every constant of $\ins{C}$ refers to a distinct object. We write $\sms{\dep}$ for the set of stable models of $\dep$.
Finally, given a database $D$ of $\sch{\dep}$, an instance $I$ of $\sch{\dep}$ is a {\em stable model of $D$ and $\dep$} if $I \in \sms{\dep[D]}$ with $\dep[D] = \{\mathsf{True} \ra \alpha \mid \alpha \in D\} \cup \dep$. We write $\sms{D,\dep}$ for the set of stable models of $D$ and $\dep$, i.e., $\sms{D,\dep} = \sms{\dep[D]}$.

\OMIT{
\medskip

\noindent {\it Remark.} Recall that in the definition of the stable model semantics for Datalog$^\naf$, a rule $\rho$ is converted into a logical sentence $\dep_\rho$. Observe that $\dep_\rho$ is a TGD$^\naf$ without existentially quantified variables. Hence, one could define the stable model semantics for Datalog$^\naf$ as for TGD$^\naf$ without relying on grounding.
Indeed, for an (infinitary) Datalog$^\naf$ program $\Pi$, $\sms{\Pi} = \sms{\{\dep_\rho \mid \rho \in \Pi\}}$~\cite{AlMP17}.

\medskip

\noindent {\it Remark.} Henceforth, by abuse of terminology, we may silently treat an (infinitary) TGD$^\naf$ program {\em without} existentially quantified variables as an (infinitary) Datalog$^\naf$ program, and vice versa.
}

\medskip
\noindent\paragraph{Probability Spaces.} 
A {\em discrete probability space} is a pair $(\Omega,P)$, where $\Omega$ is a finite or countably infinite set, called the {\em sample space}, and $P$ is a function $\Omega \ra [0,1]$ such that $\sum_{o \in \Omega} P(o) = 1$, called {\em discrete probability distribution over $\Omega$}.
Let $\mathcal{P}_{\Omega}$ be the set that collects all the functions $P : \Omega \ra [0,1]$ such that $(\Omega,P)$ is a discrete probability space, that is, all the  discrete probability distributions over $\Omega$. A {\em parameterized probability distribution over $\Omega$} is a function $\delta : \mathbb{R}^k \rightarrow \mathcal{P}_\Omega$ with $k >0$ being the parameter dimension, i.e., $\delta(\bar p)$ is a discrete probability distribution over $\Omega$ for every parameter instantiation $\bar{p} \in \mathbb{R}^k$. For the sake of presentation, we write $\delta\params{\bar p}$ instead of $\delta(\bar p)$ to avoid the overuse of parentheses, i.e., expressions of the  form $\delta(\cdot)(\cdot)$.
Due to our assumption that constants are essentially real numbers, it suffices to consider discrete probability distributions that are {\em numerical} , that is, over a sample space $\Omega \subseteq \mathbb{R}$.
%
%A simple example that illustrates the notion of parameterized probability distribution can be found in the appendix.
A simple illustrative example can be found in the appendix.

\OMIT{
\begin{example}
	Throwing a die is a classical example of a random process. The throwing of an {\em unbiased} die can be modelled via the discrete probability space $(\Omega,P)$, where $\Omega = \{1,\ldots,6\}$, the six faces of a die, and $P$ is such that $P(i) = \frac{1}{6}$ for each $i \in \Omega$.
	On the other hand, we can model the throwing of a {\em biased} die by using a parameterized probability distribution. In particular, this can be done via $\mathsf{Die} : \mathbb{R}^6 \ra \mathcal{P}_{\Omega}$ with $\Omega = \{0,1,\ldots,6\}$ defined as follows: for every $\bar p = (p_1,\ldots,p_6) \in \mathbb{R}^6$, $\sum_{i=1}^{6} p_i = 1$ implies $\mathsf{Die}\params{\bar p}(0) = 0$ and $\mathsf{Die}\params{\bar p}(i) = p_i$ for each $i \in [6]$; otherwise, $\mathsf{Die}\params{\bar p}(0) = 1$ and $\mathsf{Die}\params{\bar p}(i) = 0$ for each $i \in [6]$. Note that the outcome $0$ is associated with incorrect instantiations of the parameters. \hfill\markfull
\end{example}
}

%Although the notion of discrete probability space is very useful for our work, it is does not suffice since we are going to encounter uncountable sample spaces. Therefore,
We further need to recall the notion of probability space that goes beyond countable sample spaces.
Let $\Omega$ be a (possibly uncountable) set. A {\em $\sigma$-algebra} over $\Omega$ is a subset $\mathcal{F}$ of $2^\Omega$ (the powerset of $\Omega$), i.e., a collection of subsets of $\Omega$, that (i) contains $\Omega$, (ii) is closed under complement, i.e., $F \in \mathcal{F}$ implies $\Omega \setminus F \in \mathcal{F}$, and (iii) is closed under countable unions, i.e., for a countable set $\mathcal{E}$ of elements of $\mathcal{F}$, $\bigcup_{E \in \mathcal{E}} E \in \mathcal{F}$. As a consequence, a $\sigma$-algebra contains the empty set, and is closed under countable intersections.
A {\em probability space} is defined as a triple $(\Omega, \mathcal{F}, P)$, where
\begin{itemize}
	\item $\Omega$ is a (possibly uncountable) set, called the {\em sample space},
	\item $\mathcal{F}$ is a $\sigma$-algebra over $\Omega$, called the {\em event space}, and
	\item $P : \mathcal{F} \ra [0,1]$, called a {\em probability measure}, is such that $P(\Omega) = 1$, and, for every countable set $\mathcal{E}$ of pairwise disjoint elements of $\mathcal{F}$, $P(\bigcup_{E \in \mathcal{E}} E) = \sum_{E \in \mathcal{E}} P(E)$.
\end{itemize}
For a non-empty $F \subseteq 2^\Omega$, the closure of $F$ under complement and countable unions is a $\sigma$-algebra, and it is said to be {\em generated} by $F$.

\section{Generative Datalog with Negation}\label{sec:gen-datalog}

A Datalog$^\naf$ program (which is essentially a TGD$^\naf$ program without existentially quantified variables) specifies how to obtain a set of stable models from an input database. In this section, we present generative Datalog$^\naf$ programs that essentially specify how to infer a distribution over sets of stable models from an input database. We first present the syntax of generative Datalog$^\naf$ programs. We then give an informal description of their semantics, highlighting the challenges in defining the probabilistic semantics in question, and then proceed with the formal definition. 
%In particular, we define the notion of possible outcomes 
%(which, as we shall see, are ground (infinitary) TGD$^\naf$ programs) 
%that form the samples in the (possibly uncountable) sample space, and then define a probability space over that sample space.
%
Henceforth, by $\Delta$ we refer to a finite set of parameterized numerical probability distributions.

\medskip
\noindent \paragraph{Syntax.}
%
%Let $\Delta$ be a finite set of parameterized numerical probability distributions. 
A {\em $\Delta$-term} is an expression of the form $\delta\params{\bar p}[\bar q]$, where $\bar p$ is a non-empty tuple of terms, $\bar q$ is a (possibly empty) tuple of terms, and $\delta$ belongs to $\Delta$ with parameter dimension $|\bar p|$. We call $\bar p$ the {\em distribution parameters}, and $\bar q$ the (optional) {\em event signature}. When the event signature is missing we simply write $\delta\params{\bar p}$.
Intuitively, $\delta\params{\bar p}[\bar q]$ denotes a sample from the probability distribution $\delta\params{\bar p}$, where different samples are drawn for different event signatures.
A {\em $\Delta$-atom} over a schema $\ins{S}$ is an expression of the form $R(t_1,\ldots,t_n)$, where $R \in \ins{S}$, $\ar{R} = n$, and $t_i$ is either an ordinary term (constant or variable), or a $\Delta$-term for each $i \in [n]$. In other words, a $\Delta$-atom is an atom that can mention, apart from terms, also $\Delta$-terms.
%
%In the rest of the paper, $\Delta$ always refers to a finite set of parameterized numerical probability distributions.

A {\em generative Datalog rule with negation w.r.t.~$\Delta$} (or simply {\em GDatalog$^{\naf}[\Delta]$ rule}) $\rho$ over a schema $\ins{S}$ is an expression of the form
\[
R_1(\bar u_1),\ldots,R_n(\bar u_n), \naf P_1(\bar v_1),\ldots,\naf P_m(\bar v_m) \ra R_0(\bar w)
\]
for $n,m \geq 0$, where $R_0,R_1,\ldots,R_n,P_1,\ldots,P_m$ are predicates of $\ins{S}$, $\bar u_1,\ldots,\bar u_n$ are tuples over $\ins{C}  \cup \ins{V}$, $\bar v_1,\ldots,\bar v_m$ are tuples over $\ins{C} \cup \ins{V}$ such that each of their variables is mentioned in $\bar u_j$ for some $j \in [n]$, and $\bar w$ is a tuple consisting of terms from $\ins{C} \cup \ins{V}$ and $\Delta$-terms such that each of its variables, possibly as part of the distribution parameters or the event signature of a $\Delta$-term, is mentioned in $\bar u_j$ for some $j \in [n]$. 
%
%The difference between ordinary  Datalog$^\naf$ rules and GDatalog$^\naf[\Delta]$ rules is that the conclusion can be a $\Delta$-atom.
%
The ($\Delta$-)atom that appears on the right of the $\ra$ symbol is the {\em head} of $\rho$, denoted $H(\rho)$, while the expression on the left of the $\ra$ symbol is the {\em body} of $\rho$, denoted $B(\rho)$. 
We further write $B^+(\rho)$ for the set of positive literals in $B(\rho)$, and $B^-(\rho)$ for the set of atoms appearing in the negative literals of $B(\rho)$.

A {\em GDatalog$^\naf[\Delta]$ program} $\Pi$ over $\ins{S}$ is a finite set of Datalog$^\naf$ rules over $\ins{S}$.
A predicate $R \in \ins{S}$ occurring in $\Pi$ is called {\em extensional} if there is no rule of the form ${\rm body} \ra R(\bar w)$ in $\Pi$, i.e., $R$ occurs only in rule-bodies; otherwise, it is called {\em intensional}.
%if there exists at least one rule of the form ${\rm body} \ra R(\bar w)$ in $\Pi$, that is, $R$ appears in the head of at least one rule of $\Pi$. 
The {\em extensional (database) schema} of $\Pi$, denoted $\edb{\Pi}$, consists of the extensional predicates in $\Pi$, while the {\em intensional schema} of $\Pi$, denoted $\idb{\Pi}$, is the set of all intensional predicates in $\Pi$. The {\em schema} of $\Pi$, denoted $\sch{\Pi}$, is the set $\edb{\Pi} \cup \idb{\Pi}$. 
%
%As for TGD$^\naf$ programs, 
Given a GDatalog$^\naf[\Delta]$ program $\Pi$, and a database $D$ of $\edb{\Pi}$, $\Pi[D]$ is the program $\{ \ra \alpha \mid \alpha \in D\} \cup \Pi$.

\OMIT{
defined in the same way as a Datalog$^\naf$ rule over $\ins{S}$ with the key difference that the head of $\rho$ can be either an (ordinary) atom or a $\Delta$-atom over $\ins{S}$. Let us stress that the safety condition must be satisfied even if the head is a $\Delta$-atom. More precisely, every variable that appears in a negative literal of the body of $\rho$, or in the head of $\rho$, possibly as part of the distribution parameters or the event signature of a $\Delta$-term, must also appear in a positive literal of the body of $\rho$. The notations $H(\rho)$, $B(\rho)$, $B^+(\rho)$ and $B^-(\rho)$ can be defined as for Datalog$^\naf$ rules.
A {\em GDatalog$^{\naf}[\Delta]$ program} $\Pi$ over $\ins{S}$ is a finite set of GDatalog$^{\naf}[\Delta]$ rules over $\ins{S}$; as we shall see, in our later development we do not need to deal with infinitary GDatalog$^{\naf}[\Delta]$ programs.
The notions of extensional and intensional predicate, as well as $\edb{\Pi}$, $\idb{\Pi}$ and $\sch{\Pi}$ are defined as for Datalog$^\naf$ programs. 
%We also write $\adom{\Pi}$ for the set of constants of $\ins{C}$ occurring in $\Pi$ (including those that are part of the distribution parameters or the event signature of a $\Delta$-term).
}

\begin{example}\label{exa:gdatalog-syntax}
	Consider the ``network resilience'' example already discussed in the Introduction. Recall that the network is dominated by the malware if all routers are infected or isolated (i.e., connected only to infected routers).
	Assume we have the following schema: 
	\begin{eqnarray*}
		\text{\rm Router}(\mi{router\_id}) && \text{\rm Infected}(\mi{router\_id},\mi{boolean})\\
		\text{\rm Connected}(\mi{router\_id},\mi{router\_id})&& \text{\rm Uninfected}(\mi{router\_id}).
	\end{eqnarray*}
	Consider the parameterized distribution $\flip : \mathbb{R} \ra \mathcal{P}_{\{0,1\}}$ over $\{0,1\}$ such that $\flip\params{p}(1) = p$ and $\flip \params{p}(0) = 1-p$, and let $\Delta = \{\flip\}$.
	The GDatalog$^\naf[\Delta]$ program $\Pi$ over the above schema that encodes the malware domination cases is as follows:
	\begin{align*}
	&\text{\rm Infected}(x,1), \text{\rm Connected}(x,y) \ra \text{\rm Infected}(y,\flip\params{0.1}[x,y])\\
	&\text{\rm Router}(x), \naf \text{\rm Infected}(x,1) \ra \textrm{Uninfected}(x)\\
	&\text{\rm Uninfected}(x), \text{\rm Uninfected}(y), \text{\rm Connected}(x,y) \ra \bot,
	\end{align*}
	where $\bot$ denotes $\mathsf{False}$. The symbol $\bot$ is syntactic sugar as $\mathsf{False}$ can be always simulated using stable negation. In particular, we can replace $\bot$ with a $0$-ary predicate $\text{\rm Fail}$ that is forced to be false in every stable model via the rule $\text{\rm Fail}, \neg \text{\rm Aux} \ra \text{\rm Aux}$, where $\text{\rm Aux}$ is a new predicate not mentioned in the program in question. \hfill\markfull
\end{example}

\noindent \paragraph{An Informal Semantics.}
Consider a positive (i.e., without negation) GDatalog$[\Delta]$ program $\Pi$. Every configuration (i.e., a set) $C$ of probabilistic choices made via the $\Delta$-terms during the execution of $\Pi$ on an input database $D$ leads to a {\em different} minimal model $M_C$ of $D$ and $\Pi$ (in fact, of $D$ and a set of TGDs $\dep_\Pi$ that somehow captures the intended meaning of $\Pi$). Furthermore, the probabilistic choices of $C$ are explicitly represented in $M_C$, which means that $M_C$ carries enough information that allows us to calculate its probability. 
According to the probabilistic semantics from~\cite{BCKOV17}, the output of $\Pi$ on $D$ is a probability space over the sample space consisting of all the minimal models of $D$ and $\dep_\Pi$ with non-zero probability, the so-called {\em possible outcomes}. This is a natural generalization of the (deterministic) semantics of Datalog, where the output of a positive program on an input database is their unique minimal model.

Once we add negation to generative Datalog, one may be tempted to think that we can follow a similar approach by considering stable instead of minimal models.
However, in the presence of stable negation, the situation changes significantly. A configuration of probabilistic choices no longer leads to a single minimal model, but rather a (possibly empty) set of stable models. Consider the GDatalog$^\naf[\{\flip\}]$ program $\Pi_{\mi{coin}}$ consisting of
\begin{eqnarray*}
\ra \text{\rm Coin}(\flip\params{0.5}) && \text{\rm Coin}(1), \neg \text{\rm Aux}_1 \ra \text{\rm Aux}_2\\
\text{\rm Coin}(0) \ra \bot && \text{\rm Coin}(1), \neg \text{\rm Aux}_2 \ra \text{\rm Aux}_1
\end{eqnarray*}
which encodes the flipping of a ``fair'' coin.
%; the definition of $\flip$ and the meaning of $\bot$ have been discussed in Example~\ref{exa:gdatalog-syntax}.
%
Intuitively, if the flipping results to $0$ (which corresponds to {\em Heads}), 
%i.e., $\flip\params{0.5}$ in the program is replaced by $0$, 
then there is no stable model, but if the result is $1$ (which corresponds to {\em Tails}), then $\{\text{\rm Aux}_1, \text{\rm Coin}(1)\}$ and $\{\text{\rm Aux}_2, \text{\rm Coin}(1)\}$ are the stable models of $\Pi_{\mi{coin}}$ (in fact, similarly to the positive case, of a TGD$^\naf$ program $\dep_{\Pi_{\mi{coin}}}$ that somehow captures the intended meaning of $\Pi_{\mi{coin}}$; here the database is empty).
It may also happen that different configurations of probabilistic choices lead to the same set of stable models (for example, if we add to $\Pi_{\mi{coin}}$ the rule $\text{\rm Coin(1)} \ra\bot$).
%
%Things are actually even more complex as different configurations of probabilistic choices may lead to the same (possibly empty) set of stable models. The latter is not shown by the ``coin'' example, but it is easy to devise an example that illustrates this fact.

The above informal discussion essentially tells us that, in general, there is no correspondence between stable models and configurations of probabilistic choices. 
%with negation in place, individual stable models are only loosely connected to the probabilistic choices made via the $\Delta$-terms occurring in the given program, and 
Thus, it will be conceptually problematic to directly consider individual stable models as possible outcomes, which in turn form the underlying sample space.
This leads to the obvious idea of considering sets of stable models (including the empty set) with non-zero probability, obtained via certain configurations of probabilistic choices, as possible outcomes.
Coming back to the ``coin'' example, this would result in the possible outcome $\{\{\text{\rm Aux}_1, \text{\rm Coin}(1)\}, \{\text{\rm Aux}_2, \text{\rm Coin}(1)\}\}$ for the
configuration where the result of the flip is $1$, and the possible outcome $\emptyset$ for the configuration where the result of the flip is $0$, and the two respective events should have probability $0.5$ each.

Although the above idea is a plausible one, which we have thoroughly explored, it brings us to the other complication that is coming with negation.
Unlike positive generative Datalog, the probabilistic choices are not necessarily reflected in the resulting set of stable models. In other words, there are cases where it is not possible to extract the configuration of probabilistic choices that led to a set of stable models $M$ by simply inspecting $M$. This can be seen in the ``coin'' example, where the fact that the result of the flip is $0$ is not explicitly encoded in the respective stable model, that is, the empty model. 
This is problematic as a set of stable models does not carry enough information that allows us to calculate its probability.
To overcome this complication, the key idea is to symbolically represent a set of stable models $M$ by an (infinitary) ground program such that its set of stable models is precisely $M$, and it also reflects the probabilistic choices that led to $M$. In other words, the possible outcomes will not be sets of stable models, but rather (infinitary) ground programs that lead to sets of stable models.
There is, however, a crucial choice to be made here: infinite possible outcomes (i.e., infinitary ground programs) are considered as valid or invalid. The probabilistic semantics by B\'{a}r\'{a}ny et al. in~\cite{BCKOV17} consider infinite possible outcomes as valid ones. On the other hand, as argued by Grohe et al. in~\cite{GKKL20}, where generative Datalog with continuous distributions is studied, it is both technically and conceptually more meaningful to consider infinite possible outcomes as invalid ones, which are collected in the so-called error event. In our development, we adopt the view by Grohe et al.~\cite{GKKL20}.

This brings us to the importance of the adopted grounding strategy since, depending on how sophisticated this strategy is, a certain configuration of probabilistic choices may have a finite or an infinite representative grounding, which in turn affects the resulting probability space. Thus, the definition of our semantics is parameterized by a grounder, and we will discuss a mechanism for performing a comparison among semantics induced by different grounders.

%{\color{red} The goal is to intuitively explain the semantics for generative Datalog with stable negation, and highlight the technical challenges.
%
%Sets of models are possible outcomes. Problem: Grouping. Example (see file). Ground programs to represent different groups, each representing a set of models. Ground programs are finite (= useful). Grounding is important, naive grounding may lead to many infinite programs, and is counter-intuitive: show this w.r.t. TODS semantics, and also w.r.t. an example of a stratified program. In the probability space section, refer back to this choice.}

\OMIT{
\medskip
\noindent \paragraph{Grounding Datalog$^\naf$ Programs.} As discussed above, the semantics of generative Datalog with stable negation heavily relies on a notion of grounder. We provide an abstract definition of what we call grounder for (infinitary) Datalog$^\naf$ programs, and then extend it to Datalog$^\naf[\Delta]$ programs. 
Consider an (infinitary) Datalog$^\naf$ program $\Pi$, and a monotonic function $G : 2^{\ground{\Pi}} \ra 2^{\ground{\Pi}}$. For an (infinitary) program $\Pi' \in 2^{\ground{\Pi}}$, we define the sets
\[
G^{0}(\Pi') = \Pi' \quad \text{and} \quad G^{i+1}(\Pi') = G(G^{i}(\Pi')) \text{ for } i > 0,
\]
and we let
\[
G^{\infty}(\Pi')\ =\ \bigcup_{i \geq 0} G^{i}(\Pi').
\]
Due to the monotonicity of $G$, it can be shown that $G^{\infty}(\Pi')$ is the least fixpoint of $G$ that contains $\Pi' \in 2^{\ground{\Pi}}$. 
%
%Note also that $G^{\infty}(\Pi')$, although one may think that it can be infinite, it is actually finite since it is a subset of the finite set $\ground{\Pi}$.
We can now introduce the notion of grounder for Datalog$^\naf$.

\begin{definition}[\textbf{Datalog$^\naf$ Grounder}]\label{def:grounder-datalog}
	Consider an (infinitary) Datalog$^\naf$ program $\Pi$. A {\em grounder of $\Pi$} is a monotonic function
	\[
	G\ :\ 2^{\ground{\Pi}}\ \ra\ 2^{\ground{\Pi}}
	\]
	such that $\sms{G^{\infty}(\emptyset)}= \sms{\Pi}$. $G^{\infty}(\emptyset)$ the {\em $G$-grounding of $\Pi$}. \hfill\markfull
\end{definition}

In simple words, a grounder of an (infinitary) Datalog$^\naf$ program $\Pi$ is a monotonic operator $G$ that after repeatedly applying it, starting from the empty program, until there is nothing else to produce, it builds a ground (infinitary) program that has exactly the same stable models as $\Pi$.
Alternatively, a grounder of $\Pi$ can be seen as a function of the form $\{\emptyset\} \ra 2^{\ground{\Pi}}$ that assigns to the empty program a subset $\Pi'$ of $\ground{\Pi}$ such that $\sms{\Pi'} = \sms{\Pi}$. This alternative form of the notion of grounder for Datalog$^\naf$ will be conceptually useful when we deal below with GDatalog$^\naf[\Delta]$.

The most naive grounder of an (infinitary) Datalog$^\naf$ program $\Pi$ is the one that returns in a single step the whole set $\ground{\Pi}$. It is usually called Herbrand grounder, denoted $\mi{Herbrand}_\Pi$, with $\mi{Herbrand}_\Pi(\Pi') = \ground{\Pi}$ for every $\Pi' \in 2^{\ground{\Pi}}$.
A slightly more clever grounder for $\Pi$ than the Herbrand one, which avoids the construction of some superfluous rules, is the simple grounder $\mi{Simple}_\Pi$ that extends a set of rules $\Pi' \in 2^{\ground{\Pi}}$ by matching the positive body literals of $\Pi$ with ground head atoms of $\Pi'$. The notion of matching can be formalized via homomorphisms with which we assume the reader is familiar. We write $h(A) \subseteq B$ to indicate that $h$ is a homomorphism from a set of atoms $A$ to a set of atoms $B$.
Furthermore, given a rule $\rho$ and a homomorphism $h$ from $B^+(\rho)$ to some other set of atoms, we write $h(\rho)$ for the rule obtained by applying $h$ to the body and head of $\rho$. We then define the simple grounder as follows: for every $\Pi' \in 2^{\ground{\Pi}}$,
\[
\mathit{Simple}_\Pi(\Pi')\ =\ \Pi' \cup \left\{
h(\rho)\ \mid\ 
\rho \in \Pi\ \text{and}\ h(B^+(\rho)) \subseteq \heads{\Pi'}
\right\}.
\]
Despite its simplicity (hence the name ``simple''), $\mathit{Simple}_\Pi$ is a relevant grounder and its usefulness will be revealed later in the paper.
But, of course, one can devise more sophisticated grounders that can recognise that the truth value of some literals is fixed in all stable models, and avoid the generation of trivially satisfied rules.
}

\medskip
\noindent \paragraph{From GDatalog$^\naf[\Delta]$ to TGD$^\naf$.} 
%Our next step is to extend the notion of grounder to generative Datalog with stable negation. To this end, we first need to translate a GDatalog$^\naf[\Delta]$ program $\Pi$ into a TGD$^\naf$ program $\Sigma_\Pi$.
%As discussed above, a generative Datalog program with negation is converted into a TGD$^\naf$ program that somehow captures its intended meaning. 
We proceed to formalize the translation of a generative Datalog program with negation into a TGD$^\naf$ program.
%
%the semantics of generative Datalog with stable negation heavily relies on a notion of grounder. We are going to provide an abstract definition of what we call grounder for generative Datalog$^\naf$ programs. To this end, we first need a way to translate a generative Datalog$^\naf$ program to a TGD$^\naf$ program.
%
Consider a GDatalog$^\naf[\Delta]$ rule $\rho$
\[
R_1(\bar u_1),\ldots,R_n(\bar u_n), \naf P_1(\bar v_1),\ldots,\naf P_m(\bar v_m) \ra R_0(\bar w),
\]
where $\bar w = (w_1,\ldots,w_{\ar{R_0}})$ with $w_{i_1} = \delta_1\params{\bar p_{1}}[\bar q_{1}],\ldots,w_{i_r} = \delta_r\params{\bar p_{r}}[\bar q_{r}]$, for $1 \leq i_1 < \cdots < i_r \leq \ar{R_0}$, be all the $\Delta$-terms in $\bar w$. 
If $r=0$ (i.e., there are no $\Delta$-terms in $\bar w$), then $\rho_{\exists}$ is defined as the singleton set $\{\dep_\rho\}$, where $\dep_\rho$ is the TGD$^\naf$ without existentially quantified variables obtained after converting $\rho$ into a first-order sentence in the usual way: commas are treated as conjunctions, and all the variables are universally quantified.
Otherwise, $\rho_{\exists}$ is defined as the set consisting of the following TGD$^\naf$:
\begin{eqnarray*}
R_1(\bar u_1),\ldots,R_n(\bar u_n),\naf P_1(\bar v_1),\ldots,\naf P_m(\bar v_m)\ \ra\ \text{\rm Active}_{|\bar q_j|}^{\delta_j}(\bar p_{j},\bar q_{j})\\
\text{\rm Active}_{|\bar q_j|}^{\delta_j}(\bar p_{j},\bar q_{j}) \ra \exists y_j\ \textrm{Result}_{|\bar q_j|}^{\delta_j}(\bar p_{j},\bar q_{j},y_j)
\end{eqnarray*}
for each $j \in [r]$, and
\begin{multline*}
\text{\rm Result}_{|\bar q_1|}^{\delta_1}(\bar p_{1},\bar q_{1},y_1),\ldots, \text{\rm Result}_{|\bar q_r|}^{\delta_r}(\bar p_{r},\bar q_{r},y_r),\\
R_1(\bar u_1),\ldots,R_n(\bar u_n),
\naf P_1(\bar v_1),\ldots,\naf P_m(\bar v_m) \ra\ 
R_0(\bar w'),
\end{multline*}
where $\textrm{Active}_{|\bar q_j|}^{\delta_j}(\bar p_{j},\bar q_{j})$ and $\textrm{Result}_{|\bar q_j|}^{\delta_j}$ are fresh $(|\bar p_j|+|\bar q_j|)$-ary and $(|\bar p_j|+|\bar q_j|+1)$-ary predicates, respectively, not occurring in $\sch{\Pi}$, $y_1,\ldots,y_r$ are distinct variables not occurring in $\rho$, and $\bar w' = (w_1,\ldots,w_{i_1-1},y_1,w_{i_1+1},\ldots,w_{i_r-1},y_r,w_{i_r+1},\ldots,w_{\ar{R_0}})$ is obtained from $\bar w$ by replacing $\Delta$-terms with variables $y_1,\ldots,y_r$.
We finally define $\dep_\Pi$ as the TGD$^\naf$ program $\bigcup_{\rho \in \Pi} \rho_\exists$.
Note that the TGDs of $\dep_{\Pi}$ with an existentially quantified variable in their head, which we call {\em active-to-result (AtR)} TGDs, are essentially encoding the probabilistic choices during an execution of the program. 
We write $\dep_{\Pi}^{\exists}$ for the program that collects all the active-to-result TGDs of $\dep_{\Pi}$, and we write $\dep_{\Pi}^{\not\exists}$ for the program $\dep_{\Pi} \setminus \dep_{\Pi}^{\exists}$.
%
%{\color{red} Here we need to a give a bit of intuition underlying this translation. In particular, we need to explain the role of the Active atoms and AtR rules, which is a key difference compared to what the other papers did.}

\begin{example}\label{exa:translation-to-tgds}
	Let $\Pi$ be the GDatalog$^\naf[\Delta]$ program given in Example~\ref{exa:gdatalog-syntax}. The TGD$^\naf$ program $\dep_\Pi$ consists of
		\begin{align*}
	&\text{\rm Infected}(x,1), \text{\rm Connected}(x,y) \ra \text{\rm Active}_{2}^{\flip}(0.1,x,y)\\
	&\text{\rm Active}_{2}^{\flip}(0.1,x,y) \ra \exists z \, \text{\rm Result}_{2}^{\flip}(0.1,x,y,z)\\
	&\text{\rm Result}_{2}^{\flip}(0.1,x,y,z),\text{\rm Infected}(x,1), \text{\rm Connected}(x,y) \ra\\ &\hspace{60mm}\text{\rm Infected}(y,z),
	\end{align*}
	obtained from the first rule of $\Pi$, and the last two rules of $\Pi$ (interpreted as TGD$^\naf$ without existentially quantified variables). \hfill\markfull
\end{example}

%\medskip
\noindent \paragraph{Grounding GDatalog$^\naf[\Delta]$ Programs.} Having the translation from GDatalog$^\naf[\Delta]$ to TGD$^\naf$, we can now introduce the notion of grounding for generative Datalog$^\naf$, which relies on a simple notion of grounding for TGD$^\naf$ programs. A {\em ground instance} of a TGD$^\naf$ $\sigma$ is obtained by replacing each (universally or existentially quantified) variable in $\sigma$ with a constant of $\ins{C}$. Given a TGD$^\naf$ program $\dep$, the {\em grounding of $\dep$}, denoted $\ground{\dep}$, is the infinitary program consisting of all the ground instances of the TGD$^\naf$ of $\dep$.
We can now proceed with the grounding for generative Datalog$^\naf$ programs.

Roughly speaking, a grounder $G$ of a generative Datalog$^\naf$ program $\Pi$ is a monotonic function whose domain is a collection of sets of ground AtR TGDs of $\dep_\Pi$, and it assigns to each such set $\dep'$ (which essentially encodes a configuration of probabilistic choices) a subset $\dep''$ of $\ground{\dep_{\Pi}^{\not\exists}}$ such that $\sms{\dep' \cup \dep''}$ is the set of stable models of $\dep_\Pi$ that are consistent with the  probabilistic choices encoded by the ground AtR TGDs in $\dep'$. Then, the grounding of $\Pi$ relative to $G$ will be a set of (infinitary) ground TGD$^\naf$ programs that are somehow induced by $G$. 
We proceed to formalize this.

\OMIT{
Recall the alternative form for a grounder of a Datalog$^\naf$ program $\Pi$ as a function $\{\emptyset\} \ra 2^{\ground{\Pi}}$ that assigns to the empty program a subset $\Pi'$ of $\ground{\Pi}$ such that $\sms{\Pi'} = \sms{\Pi}$. The fact that the domain of such a  function consists only of the empty program reflects the fact that $\Pi$ does not contain probabilistic rules (i.e., rules whose head is a $\Delta$-atom).
Now, when $\Pi$ is a GDatalog$^\naf[\Delta]$, the idea is to consider a function whose domain is a collection of sets of ground AtR TGDs of $\dep_\Pi$, and it assigns to each such set $\dep'$ (which essentially encodes certain probabilistic choices) a subset $\dep''$ of $\ground{\dep_{\Pi}^{\not\exists}}$ such that $\sms{\dep' \cup \dep''}$ is the set of stable models of $\dep_\Pi$ that are consistent with the  probabilistic choices encoded by the ground AtR TGDs in $\dep'$. Note that the notation $\ground{\cdot}$, although it has been only defined for Datalog$^\naf$ programs in Section~\ref{sec:preliminaries}, it can be defined in exactly the same way for TGD$^\naf$ programs. %without existentially quantified variables.
We proceed to formalize the above intuitive description.
}

Consider a GDatalog$^\naf[\Delta]$ program $\Pi$. 
%A {\em $\ins{C}$-ground instance} of an AtR TGD $\sigma \in \dep_{\Pi}^{\exists}$ is obtained by replacing each variable in $\sigma$ with a constant of $\ins{C}$. The {\em $\ins{C}$-grounding} of $\dep_{\Pi}^{\exists}$, denoted $\mathsf{ground}_{\ins{C}}(\dep_{\Pi}^{\exists})$, is the set of all $\ins{C}$-ground instances of the AtR TGDs of $\dep_{\Pi}^{\exists}$.
%
A subset $\dep$ of $\ground{\dep_{\Pi}^{\exists}}$ is called {\em (functionally) consistent} if, with $\alpha = \text{\rm Active}_{|\bar q|}^{\delta}(\bar p,\bar q)$, for every two ground AtR TGDs $\sigma_o$ and $\sigma_{o'}$ of the form
\[
\alpha\ \ra\ \text{\rm Result}_{|\bar q|}^{\delta}(\bar p,\bar q,o) \quad \text{and} \quad \alpha\ \ra\ \text{\rm Result}_{|\bar q|}^{\delta}(\bar p,\bar q,o')
\]
respectively, it holds that $o = o'$. Intuitively, this means that $\dep$ encodes valid probabilistic choices.
We write $[2^{\ground{\dep_{\Pi}^{\exists}}}]_{=}$ for all the consistent subsets of $\ground{\dep_{\Pi}^{\exists}}$.
Observe that a consistent subset $\dep$ of $\ground{\dep_{\Pi}^{\exists}}$ induces a partial function $\mi{AtR}_\dep : \mi{Act} \ra \mi{Res}$, where $\mi{Act}$ (resp., $\mi{Res}$) is the set of all atoms that can be formed using predicates of the form $\text{\rm Active}_{n}^{\delta}$ (resp., $\text{\rm Result}_{n}^{\delta}$) of $\sch{\dep}$, and constants from $\ins{C}$.
We say that $\mi{AtR}_{\dep}$ is {\em compatible} with a set $\dep' \subseteq \ground{\dep_{\Pi}^{\not\exists}}$, written $\mi{AtR}_\dep \comp \dep'$, if it is defined on every atom of the form $\text{\rm Active}_{n}^{\delta}(\bar p,\bar q)$ occurring in $\heads{\dep'}$, that is, the set of head atoms of $\dep'$.
A {\em totalizer} for the partial function $\mi{AtR}_\dep$ is a set $\dep' \in [2^{\ground{\dep_{\Pi}^{\exists}}}]_{=}$ such that $\dep \subseteq \dep'$ and $\mi{AtR}_{\dep'}$ is total.
We can now introduce the desired notion of grounding.

\begin{definition}[\textbf{Grounding Generative Datalog$^\naf$}]\label{def:grounder-gdatalog}
	Consider a GDatalog$^\naf[\Delta]$ program $\Pi$. A {\em grounder of $\Pi$} is a monotonic function
	\[
	G\ :\ [2^{\ground{\dep_{\Pi}^{\exists}}}]_{=}\ \ra\ 2^{\ground{\dep_{\Pi}^{\not\exists}}} 
	\]
	such that, for every $\dep \in [2^{\ground{\dep_{\Pi}^{\exists}}}]_{=}$, if $\mi{AtR}_\dep \comp G(\dep)$, then $\sms{G(\dep) \cup \dep} = \sms{\dep_{\Pi}^{\not\exists} \cup \dep'}$ for every totalizer $\dep'$ of $\mi{AtR}_{\dep}$.
	The {\em $G$-grounding of $\Pi$}, denoted $G$-$\ground{\Pi}$, is defined as the set
	\begin{multline*}
	\{\dep \cup G(\dep) \mid \dep \in \terminals{G} \text{ and}\\
	\text{there is no } \dep' \in \terminals{G} \text{ such that } \dep' \subsetneq \dep\}
	\end{multline*}
	with $\terminals{G} = \left\{\dep \in [2^{\ground{\dep_{\Pi}^{\exists}}}]_{=}\, \mid\, \mi{AtR}_\dep \comp G(\dep) \right\}$. \hfill\markfull
\end{definition}

The importance of the monotonicity of grounders for generative Datalog$^\naf$ will be revealed in Section~\ref{sec:chase}, where we present our chase-based semantics. Before proceeding with our probabilistic semantics, we present a concrete example of a grounder, which is conceptually very simple (hence the name ``simple grounder''), but at the same time a relevant one as we shall see later in the paper.

\medskip
\noindent \paragraph{Simple Grounder.}
Given a GDatalog$^\naf[\Delta]$ program $\Pi$, we define the so-called simple grounder $\mi{GSimple}_{\Pi}$, which relies on an operator that operates on (infinitary) existential-free TGD$^\naf$ programs. 
Consider an (infinitary) TGD$^\naf$ program $\dep$ without existentially quantified variables. We define the operator $\mi{Simple}_{\dep}$ that extends an (infinitary) TGD$^\naf$ program $\dep' \in 2^{\ground{\dep}}$ by matching the positive body literals of $\dep$ with ground head atoms of $\dep'$. The notion of matching can be formalized via homomorphisms with which we assume the reader is familiar. We write $h(A) \subseteq B$ to indicate that $h$ is a homomorphism from a set of atoms $A$ to a set of atoms $B$. Given a TGD$^\naf$ $\sigma$ and a homomorphism $h$ from $B^+(\sigma)$ to some other set of atoms, we write $h(\sigma)$ for the TGD$^\naf$ obtained by applying $h$ to the body and head of $\sigma$. 
We define $\mi{Simple}_{\dep}$ as the monotonic function from $2^{\ground{\dep}}$ to $2^{\ground{\dep}}$ such that, for every $\dep' \in 2^{\ground{\dep}}$,
\[
\mathit{Simple}_\dep(\dep')\ =\ \dep' \cup \left\{
h(\sigma) \mid 
\sigma \in \dep\ \text{ and }\ h(B^+(\sigma)) \subseteq \heads{\dep'}
\right\}.
\]
We further define the sets
\[
\mi{Simple}_{\dep}^{0}(\dep') = \dep' 
\]
\[
\mi{Simple}_{\dep}^{i+1}(\dep') = \mi{Simple}_{\dep}(\mi{Simple}_{\dep}^{i}(\dep')) \text{ for } i > 0,
\]
and we let
\[
\mi{Simple}_{\dep}^{\infty}(\dep')\ =\ \bigcup_{i \geq 0} \mi{Simple}_{\dep}^{i}(\dep').
\]
Due to the monotonicity of $\mi{Simple}_{\dep}$, it is clear that $\mi{Simple}_{\dep}^{\infty}(\dep')$ is the least fixpoint of $\mi{Simple}_{\dep}$ that contains $\dep' \in 2^{\ground{\dep}}$. We are now ready to define the function $\mi{GSimple}_{\Pi}$.

\begin{definition}[\textbf{Simple Grounder}]\label{def:simple-grounder}
	Consider a GDatalog$^\naf[\Delta]$ program $\Pi$. $\mi{GSimple}_{\Pi} : [2^{\ground{\dep_{\Pi}^{\exists}}}]_{=} \ra 2^{\ground{\dep_{\Pi}^{\not\exists}}}$ is such that
	\[
	\mi{GSimple}_{\Pi}(\dep)\ =\ \mi{Simple}_{\dep'}^{\infty}(\emptyset) \setminus \dep
	\]
	where $\dep' = \dep_{\Pi}^{\not\exists} \cup \dep$, for every $\dep \in [2^{\ground{\dep_{\Pi}^{\exists}}}]_{=}$. \hfill\markfull
\end{definition}

It is not difficult to show that $\mi{GSimple}_{\Pi}$ is indeed a grounder:

\begin{proposition}\label{pro:gsimple-grounder}
	Consider a GDatalog$^\naf[\Delta]$ program $\Pi$. It holds that $\mi{GSimple}_{\Pi}$ is a grounder of $\Pi$.
\end{proposition}

An example that illustrates the simple grounder follows:

\begin{example}\label{exa:gsimple}
	Let $\Pi$ be the GDatalog$^\naf[\Delta]$ program from Example~\ref{exa:gdatalog-syntax}. Its translation into the TGD$^\naf$ program $\dep_\Pi$ can be found in Example~\ref{exa:translation-to-tgds}.
	Consider the database $D$ defined as
	\begin{multline*}
	\{\text{\rm Router}(i) \mid i \in [3]\} \cup \{\text{\rm Infected}(1,1)\}\\
	\{\text{\rm Connected}(i,j) \mid i,j \in [3] \text{ and } i \neq j\}
	\end{multline*}
	that stores a fully connected network consisting of three routers, the first being initially infected.
	Hence, $\mi{GSimple}_{\Pi[D]}(\emptyset)$ contains
	\begin{align*}
		&\text{\rm Infected}(1,1), \text{\rm Connected}(1,2) \ra \text{\rm Active}_{2}^{\flip}(0.1,1,2)\\
		&\text{\rm Infected}(1,1), \text{\rm Connected}(1,3) \ra \text{\rm Active}_{2}^{\flip}(0.1,1,3)
	\end{align*}
	and the following TGD$^\naf$ for all $i,j \in [3]$ with $i \neq j$:
\begin{align*}
	&\text{\rm Router}(i), \naf \text{\rm Infected}(i,1) \ra \textrm{Uninfected}(i)\\
	&\text{\rm Uninfected}(i), \text{\rm Uninfected}(j), \text{\rm Connected}(i,j) \ra \bot.
\end{align*}
	Now, considering the program
	\[
	\dep\ =\ \left\{\text{\rm Active}_{2}^{\flip}(0.1,1,i) \ra \text{\rm Result}_{2}^{\flip}(0.1,1,i,0) \mid i \in \{2,3\}\right\},
	\]
	$\mi{GSimple}_{\Pi[D]}(\Sigma)$ extends $\mi{GSimple}_{\Pi[D]}(\emptyset)$ with
	\begin{align*}
		&\text{\rm Infected}(1,1), \text{\rm Connected}(1,i) \ra \text{\rm Active}_{2}^{\flip}(0.1,1,2)\\
		&\text{\rm Result}_{2}^{\flip}(0.1,1,i,0),\text{\rm Infected}(1,1), \text{\rm Connected}(1,i) \ra\\ &\hspace{60mm}\text{\rm Infected}(i,0)
	\end{align*}
	for each $i \in \{2,3\}$.
	%
	%It can be verified
	We have that $\Sigma \in \terminals{\mi{GSimple}_{\Pi[D]}}$, and $\Sigma \cup \mi{GSimple}_{\Pi[D]}(\Sigma) \in \mi{GSimple}_{\Pi[D]}$-$\ground{\Pi[D]}$.
	%Moreover,
	Also, $\sms{\Sigma \cup \mi{GSimple}_{\Pi[D]}(\Sigma)} = \emptyset$
	since routers 2 and 3 are uninfected.
\hfill\markfull
\end{example}

\OMIT{
A concrete grounder of a GDatalog$^\naf[\Delta]$ program $\Pi$ is the simple grounder $\mi{GSimple}_{\Pi}$ (we adopt the name $\mi{GSimple}_\Pi$ instead of $\mi{Simple}_\Pi$ in order to avoid confusion with the simple grounder for Datalog$^\naf$), which relies on the simple grounder for Datalog$^\naf$ defined above.
In particular, for every $\dep \in \left[2^{\ground{\dep_{\Pi}^{\exists}}}\right]_{=}$,
%with $\Pi'$ and $\Pi''$ being the (infinitary) Datalog$^\naf$ programs obtained from the (infinitary) TGD$^\naf$ programs $\dep_{\Pi}^{\exists} \cup \dep$ and $\dep$, respectively, in the obvious way (see an $\exists$-free TGD$^\naf$ as a Datalog$^\naf$ rule), we define
\[
\mi{GSimple}_{\Pi}(\dep)\ =\ \mi{Simple}_{\dep'}^{\infty}(\emptyset) \setminus \dep \quad \text{where} \quad \dep' = \dep_{\Pi}^{\not\exists} \cup \dep.
\]
It is not difficult to show that this is indeed a grounder:

\begin{proposition}\label{pro:gsimple-grounder}
	Consider a GDatalog$^\naf[\Delta]$ program $\Pi$. It holds that $\mi{GSimple}_{\Pi}$ is a grounder of $\Pi$.
\end{proposition}

The usefulness of the grounder $\mi{GSimple}_{\Pi}$ (and thus, of the simple grounder for Datalog$^\naf$) is revealed below during the comparison of the semantics for generative Datalog$^\naf$ that we are about to propose with the existing semantics for generative Datalog from~\cite{BCKOV17} when we focus our attention on positive Datalog programs.

At this point, one may ask why a grounder $G$ of a GDatalog$^\naf[\Delta]$ program $\Pi$ is not defined in a simpler way than what was done in Definition~\ref{def:grounder-gdatalog} by using the same principle as in the definition of $\mi{GSimple}_{\Pi}$, that is, for every $\dep \in \left[2^{\ground{\dep_{\Pi}^{\exists}}}\right]_{=}$, $G(\dep) = H^{\infty}(\emptyset) \setminus \dep$, where $H$ is a grounder of the (possibly infinitary) Datalog$^\naf$ program $\dep_{\Pi}^{\not\exists} \cup \dep$.  {\color{red} We should explain why this definition is not enough.}
}

\OMIT{
\begin{definition}\label{def:grounder-gdatalog}
	Consider a GDatalog$^\naf[\Delta]$ program $\Pi$. A {\em grounder of $\Pi$} is a monotonic function
	$T_\Pi$ from the consistent subsets of $\mathsf{ground}_{\ins{C}}(\dep_{\Pi}^{\exists})$ to $2^{\ground{\dep_{\Pi}^{\not\exists}}}$ such that, for every consistent subset $\dep$ of $\mathsf{ground}_{\ins{C}}(\dep_{\Pi}^{\exists})$, if $\mi{AtR}_\dep$ is compatible with $T_\Pi(\dep)$, then $\sms{T_{\Pi}(\dep) \cup \dep} = \sms{\dep_\Pi \cup \dep'}$ for each totalizer $\dep'$ of $\mi{AtR}_{\dep}$. \hfill\markfull
\end{definition}
}

\noindent \paragraph{Probabilistic Semantics.} We now have all the ingredients for defining the semantics of generative Datalog$^\naf$. We start by formalizing the notion of possible outcome, which in turn gives rise to the sample space over which we need to define a probability space.

\begin{definition}[\textbf{Possible Outcome}]\label{def:possible-outcome}
	Consider a GDatalog$^\naf[\Delta]$ program $\Pi$, and a database $D$ of $\edb{\Pi}$. Let $G$ be a grounder of $\Pi[D]$. A {\em possible outcome of $D$ w.r.t.~$\Pi$ relative to $G$} is a (possibly infinitary) program $\dep \in G\text{-}\ground{\Pi[D]}$ such that $\delta\params{\bar p}(o) > 0$ for every $\Delta$-atom $\textrm{Result}_{|\bar q|}^{\delta}(\bar p,\bar q,o) \in \heads{\dep}$. We write $\Omega_{\Pi,G}(D)$ for the set of all possible outcomes of $D$ w.r.t.~$\Pi$ relative to $G$. 
	%We further write $\Omega_{\Pi,T}^{\fin}(D)$ for the set that collects all the (finite) programs of $\Omega_{\Pi,T}(D)$, and $\Omega_{\Pi,T}^{\infty}(D)$ for the set $\left(\Omega_{\Pi,T}(D) \setminus \Omega_{\Pi,T}^{\fin}(D)\right)$.
	\hfill\markfull
\end{definition}

Observe that each $\dep \in \Omega_{\Pi,G}(D)$ induces a (possibly empty) subset of $\sms{\dep_{\Pi[D]}}$, that is, $\sms{\dep}$. Note that, for $\dep',\dep'' \in \Omega_{\Pi,G}(D)$, $\dep' \neq \dep''$ does {\em not} imply $\sms{\dep'} \neq \sms{\dep''}$, i.e., two different possible outcomes may induce the same subset of $\sms{\dep_{\Pi[D]}}$.

Given a GDatalog$^\naf[\Delta]$ program $\Pi$, and a database $D$ of $\edb{\Pi}$, the output of $\Pi$ on $D$, relative to some grounder $G$ of $\Pi[D]$, should be understood  as a probability space over $\Omega_{\Pi,G}(D)$. We proceed to explain how such a probability space can be defined. In essence, the events that we are interested in are subsets of $\sms{\dep_{\Pi[D]}}$ induced by proper (i.e., finite) programs of $\Omega_{\Pi,G}(D)$, or, in other words, by finite possible outcomes of $D$ w.r.t.~$\Pi$ relative to $G$; we write $\Omega_{\Pi,G}^{\fin}(D)$ for the set that collects all the (finite) programs of $\Omega_{\Pi,G}(D)$. Therefore, finite possible outcomes resulting in the same set of stable models of $\sms{\dep_{\Pi[D]}}$ must belong to same event.
Moreover, following
%the development by
Grohe et al. on generative Datalog with
continuous distributions~\cite{GKKL20}, we collect all the infinite possible
outcomes of $\Omega_{\Pi,G}(D)$ in the {\em infinity (or error) event}
$\Omega_{\Pi,G}^{\infty}(D)$.

Let $F_{\Pi,G}^{D}$ be the subset of $2^{\Omega_{\Pi,G}(D)}$ consisting of $\Omega_{\Pi,G}^{\infty}(D)$, and all the maximal subsets $E$ of $\Omega_{\Pi,G}^{\fin}(D)$ such that, for all sets $\dep',\dep'' \in E$, $\sms{\dep'} = \sms{\dep''}$. We define $\mathcal{F}_{\Pi,G}^{D}$ as the $\sigma$-algebra over $\Omega_{\Pi,G}(D)$ generated by $F_{\Pi,G}^{D}$.
%, that is, the closure of $F_{\Pi,G}^{D}$ under complement and countable unions.
%
We finally define the function $P_{\Pi,G}^{D} : \mathcal{F}_{\Pi,G}^{D} \ra [0,1]$ as follows. For every $\dep \in \Omega_{\Pi,G}^{\fin}(D)$, let
\[
\mi{Pr}(\dep)\ =\ \prod_{\textrm{Result}_{|\bar q|}^{\delta}(\bar p, \bar q, o) \in \heads{\dep}}{\delta\params{\bar p}(o)}.
\]
Then, for every (countable) set $E \in \mathcal{F}_{\Pi,G}^{D}$ with $E \in 2^{\Omega_{\Pi,G}^{\fin}(D)}$, let
\[
P_{\Pi,G}^{D}(E)\ =\ \sum_{\dep \in E}{\mi{Pr}(\dep)}.
\]
Observe that $\Omega_{\Pi,G}^{\fin}(D) \in \mathcal{F}_{\Pi,G}^{D}$, and thus, we can let
\[
P_{\Pi,G}^{D}\left(\Omega_{\Pi,G}^{\infty}(D)\right)\ =\ 1 - P_{\Pi,G}^{D}\left(\Omega_{\Pi,G}^{\fin}(D)\right).
\]
Finally, $P_{\Pi,G}^{D}$ naturally extends to countable unions via countable addition.
\OMIT{
Finally, for every (countable) set $E = \Omega_{\Pi,G}^{\infty}(D) \cup \bigcup_{i>0} E_i \in \mathcal{F}_{\Pi,G}^{D}$, 
\[
P_{\Pi,G}^{D}(E)\ =\ P_{\Pi,G}^{D}\left(\Omega_{\Pi,G}^{\infty}(D)\right) + \sum_{i>0} P_{\Pi,G}^{D}(E_i).
\]
}
We are now ready to define the output of a generative Datalog$^\naf$ program on an input database relative to a grounder.
%$\Pi$ on $D$ relative to $G$.

\begin{definition}[\textbf{Output of GDatalog$^\naf[\Delta]$ Programs}]\label{def:output-gdatalog-programs}
	Consider a GDatalog$^\naf[\Delta]$ program $\Pi$, and a database $D$ of $\edb{\Pi}$. Let $G$ be a grounder of $\Pi[D]$. The {\em output of $\Pi$ on $D$ relative to $G$}, denoted $\Pi_{G}(D)$, is defined as the triple $\left(\Omega_{\Pi,G}(D),\mathcal{F}_{\Pi,G}^{D},P_{\Pi,G}^{D}\right)$. \hfill\markfull
\end{definition}

It can be shown that the following holds, which provides the desired probabilistic semantics for generative Datalog$^\naf$.

\begin{theorem}\label{the:semantics}
	Consider a GDatalog$^\naf[\Delta]$ program $\Pi$,
	%a
	database $D$ of
	$\edb{\Pi}$, and
	%a
	grounder $G$ of $\Pi[D]$. $\Pi_{G}(D)$ is a
	probability space.
\end{theorem}

\OMIT{
\begin{theorem}\label{the:semantics}
	Consider a GDatalog$^\naf[\Delta]$ program $\Pi$, and a database $D$ of $\edb{\Pi}$. For every grounder $T$ of $\Pi[D]$, it holds that $\left(\Omega_{\Pi,T}(D),\mathcal{F}_{D,\Pi}^{T},P_{D,\Pi}^{T}\right)$ is a probability space.
	%, where $T$ is a grounder of $\Pi[D]$. 
\end{theorem}
}

We proceed to show the semantics for generative Datalog$^\naf$ in action by exploiting  our ``network resilience'' example.

\begin{example}
	Continuing Example~\ref{exa:gsimple}, $\dep \in \Omega_{\Pi,\mi{GSimple}_{\Pi[D]}}^{\fin}(D)$, and
	$\mi{Pr}(\dep) = \flip\params{0.1}(0)^2 = 0.9^2$.
	It can be verified that all other possible outcomes are finite and induce a non-empty subset of $\sms{\dep_{\Pi[D]}}$. This implies that $\{\dep\} \in F_{\Pi,\mi{GSimple}_{\Pi[D]}}^{D}$, and the probability of the event ``$\Pi[D]$ has some stable model'' is 
	\begin{eqnarray*}
	&& P_{\Pi,\mi{GSimple}_{\Pi[D]}}^{D}\left(\Omega_{\Pi,\mi{GSimple}_{\Pi[D]}}^{\fin}(D) \setminus \{\dep\}\right)\\
	&=& 1 - P_{\Pi,\mi{GSimple}_{\Pi[D]}}^{D}(\{\dep\})\ = 1 - 0.9^2\ =\ 0.19.
	\end{eqnarray*}
	Summing up, the network stored in the database $D$  (given in Example~\ref{exa:gsimple}) is dominated by the malware with probability $0.19$.
	\hfill\markfull
\end{example}

%\medskip
\noindent \paragraph{Positive Programs.} A natural question is whether the probabilistic semantics proposed above for generative Datalog$^\naf$ coincides with the existing semantics for generative Datalog introduced in~\cite{BCKOV17}, if we concentrate on positive programs that guarantee that all the possible outcomes are finite, and thus, the error event takes probability zero.
It should not be overlooked, however, that the probabilistic semantics for generative Datalog$^\naf$ is actually a family of semantics since, depending on the adopted grounder, we get different semantics; a mechanism for performing a qualitative comparison among the different semantics is discussed below.
Therefore, the key question is whether there exists a grounder for generative Datalog$^\naf$ that gives rise to semantics that coincides with that of~\cite{BCKOV17}.
Interestingly, we can show that such a grounder exists; in fact, we can show that this is the simple grounder introduced above (see Definition~\ref{def:simple-grounder}).
Due to space constraints, the details are deferred to the appendix.
%Due to space constraints, we omit the formal statement and proof of this result, which are deferred to the appendix.

\OMIT{
%; in fact, it is the simplest grounder that one can think of (hence it is dubbed ``simple grounder'').
%
Given a GDatalog$^\naf[\Delta]$ program $\Pi$, we proceed to define the so-called simple grounder $\mi{GSimple}_{\Pi}$, which relies on an operator that operates on (infinitary) existential-free TGD$^\naf$ programs. 
Consider an (infinitary) TGD$^\naf$ program $\dep$ without existentially quantified variables. We define the operator $\mi{Simple}_{\dep}$ that extends an (infinitary) TGD$^\naf$ program $\dep' \in 2^{\ground{\dep}}$ by matching the positive body literals of $\dep$ with ground head atoms of $\dep'$. The notion of matching can be formalized via homomorphisms with which we assume the reader is familiar. We write $h(A) \subseteq B$ to indicate that $h$ is a homomorphism from a set of atoms $A$ to a set of atoms $B$. Furthermore, given a TGD$^\naf$ $\sigma$ and a homomorphism $h$ from $B^+(\sigma)$ to some other set of atoms, we write $h(\sigma)$ for the TGD$^\naf$ obtained by applying $h$ to the body and head of $\sigma$. 
We define $\mi{Simple}_{\dep}$ as the monotonic function from $2^{\ground{\dep}}$ to $2^{\ground{\dep}}$ such that, for every $\dep' \in 2^{\ground{\dep}}$,
\[
\mathit{Simple}_\dep(\dep')\ =\ \dep' \cup \left\{
h(\sigma) \mid 
\sigma \in \dep\ \text{ and }\ h(B^+(\sigma)) \subseteq \heads{\dep'}
\right\}.
\]
We further define the sets
\[
\mi{Simple}_{\dep}^{0}(\dep') = \dep' 
\]
\[
\mi{Simple}_{\dep}^{i+1}(\dep') = \mi{Simple}_{\dep}(\mi{Simple}_{\dep}^{i}(\dep')) \text{ for } i > 0,
\]
and we let
\[
\mi{Simple}_{\dep}^{\infty}(\dep')\ =\ \bigcup_{i \geq 0} \mi{Simple}_{\dep}^{i}(\dep').
\]
Due to the monotonicity of $\mi{Simple}_{\dep}$, it is clear that $\mi{Simple}_{\dep}^{\infty}(\dep')$ is the least fixpoint of $\mi{Simple}_{\dep}$ that contains $\dep' \in 2^{\ground{\dep}}$. We are now ready to define the function $\mi{GSimple}_{\Pi}$.

\begin{definition}[\textbf{Simple Grounder}]\label{def:simple-grounder}
	Consider a GDatalog$^\naf[\Delta]$ program $\Pi$. $\mi{GSimple}_{\Pi} : [2^{\ground{\dep_{\Pi}^{\exists}}}]_{=} \ra 2^{\ground{\dep_{\Pi}^{\not\exists}}}$ is such that
\[
\mi{GSimple}_{\Pi}(\dep)\ =\ \mi{Simple}_{\dep'}^{\infty}(\emptyset) \setminus \dep
\]
where $\dep' = \dep_{\Pi}^{\not\exists} \cup \dep$, for every $\dep \in [2^{\ground{\dep_{\Pi}^{\exists}}}]_{=}$. \hfill\markfull
\end{definition}

It is not difficult to show that $\mi{GSimple}_{\Pi}$ is indeed a grounder:

\begin{proposition}\label{pro:gsimple-grounder}
	Consider a GDatalog$^\naf[\Delta]$ program $\Pi$. It holds that $\mi{GSimple}_{\Pi}$ is a grounder of $\Pi$.
\end{proposition}

%Another caveat here is that the probabilistic semantics for generative Datalog$^\naf$ considers infinite possible outcomes as invalid ones, and thus, collected in the error event (as done in~\cite{GKKL20} for generative Datalog with continuous distributions), whereas the probabilistic semantics of~\cite{BCKOV17} considers all the possible outcomes (finite or infinite) as valid ones. Hence, a proper comparison of the two semantics should focus on generative Datalog programs that are not only positive, but they also guarantee that all the possible outcomes are finite, and thus, the error event takes probability zero.

As said above, we can show that the simple grounder for generative Datalog$^\naf$ gives rise to semantics for generative Datalog programs with finite possible outcomes that coincides with the semantics of~\cite{BCKOV17}. 
Due to space constraints, we omit the formal statement and proof of this result, which are deferred to the appendix.
}

\OMIT{
Another interesting fact is that, in the case of positive programs, the simple grounder allows us to establish an invariance property under SMS-equivalence. Given two GDatalog$^\naf[\Delta]$ programs $\Pi$ and $\Pi'$, we say that they are {\em SMS-equivalent} if $\sms{\dep_\Pi} = \sms{\dep_{\Pi'}}$.
We further say that they are {\em simple-semantic equivalent} if, for every database $D$ of $\edb{\Pi} \cap \edb{\Pi'}$, with $G_\Pi = \mi{GSimple}_{\Pi[D]}$ and $G_{\Pi'} = \mi{GSimple}_{\Pi'[D]}$, the following holds: for every set $\mathcal{I}$ of instances,
\begin{multline*}
P_{\Pi,G_\Pi}^{D}\left(\{\Pi'' \in \Omega_{\Pi,G_\Pi}^{\fin}(D) \mid \sms{\Pi''} = \mathcal{I}\}\right) =\\
P_{\Pi',G_{\Pi'}}^{D}\left(\{\Pi'' \in \Omega_{\Pi',G_{\Pi'}}^{\fin}(D) \mid \sms{\Pi''} = \mathcal{I}\}\right).
\end{multline*}
We can then establish the following result; note that the converse of the following implication does not hold in general:

\begin{theorem}\label{the:stable-invariance-simple}
	If two GDatalog$[\Delta]$ programs are SMS-equivalent, then they are simple-semantic equivalent.
\end{theorem}

This actually tells us that rewriting a GDatalog$[\Delta]$ program in a way that preserves SMS-equivalence is
safe, as the probabilistic semantics induced by the simple grounder is in essence preserved.
}

\medskip
\noindent \paragraph{Qualitative Comparison of Probabilistic Semantics.} As already discussed, Theorem~\ref{the:semantics} gives rise to a family of semantics for generative Datalog$^\naf$ since different grounders lead to different semantics. It is thus natural to ask how we can qualitatively compare the various semantics.
%by somehow comparing the underlying  grounders.
%
Recall that infinite possible outcomes are considered as invalid ones that are collected in the error event, and, in general, the larger the error event, the less the probability assigned to valid (that is, finite) possible outcomes. This essentially tells us that the greater the probability assigned to finite outcomes by a probability measure, the better. Recall that the output of a program $\Pi$ on a database $D$ relative to a grounder $G$ of $\Pi[D]$ is the probability space $\left(\Omega_{\Pi,G}(D),\mathcal{F}_{\Pi,G}^{D},P_{\Pi,G}^{D}\right)$.

\begin{definition}[\textbf{Comparison of Semantics}]\label{def:compare-semantics}
	Consider a GDatalog$^\naf[\Delta]$ program $\Pi$, and a database $D$ of $\edb{\Pi}$. Let $G, G'$ be grounders of $\Pi[D]$. We say that {\em $\Pi_G(D)$ is as good as $\Pi_{G'}(D)$} if
	\begin{multline*}
	P_{\Pi,G}^{D}\left(\{\dep \in \Omega_{\Pi,G}^{\fin}(D) \mid \sms{\dep} = \mathcal{I}\}\right)\ \geq\\ 
	P_{\Pi,G'}^{D}\left(\{\dep \in \Omega_{\Pi,G'}^{\fin}(D) \mid \sms{\dep} = \mathcal{I}\}\right)
	\end{multline*}
	for every set of stable models $\mathcal{I} \subseteq \sms{\dep_{\Pi[D]}}$. \hfill\markfull
\end{definition}

The above notion equips us with a mechanism that can guide the choice of the probabilistic semantics based on the class of generative Datalog$^\naf$ that we are interested in, which in turn depends on the underlying application. For example, if we are interested only in positive generative Datalog programs, we can then show that there is no need to go beyond the probabilistic semantics induced by the simple grounder, i.e., the semantics induced by the simple grounder is as good as the semantics induced by any other grounder.

\begin{theorem}\label{the:gsimple-best}
	Consider a GDatalog$[\Delta]$ program $\Pi$, and a database $D$ of $\edb{\Pi}$. Let $G$ be an arbitrary grounder of $\Pi[D]$. It holds that $\Pi_{\mi{GSimple}_{\Pi[D]}}(D)$ is as good as $\Pi_{G}(D)$.
\end{theorem}

\OMIT{
\begin{theorem}\label{the:gsimple-best}
	Consider a GDatalog$[\Delta]$ program $\Pi$, and a database $D$ of $\edb{\Pi}$, and let $G = \mi{GSimple}_{\Pi[D]}$. $\left(\Omega_{\Pi,G}(D),\mathcal{F}_{\Pi,G}^{D},P_{\Pi,G}^{D}\right)$ is as good as $\left(\Omega_{\Pi,G'}(D),\mathcal{F}_{\Pi,G'}^{D},P_{\Pi,G'}^{D}\right)$ for every grounder $G'$ of $\Pi[D]$.
\end{theorem}
}

\OMIT{
In simple words, Theorem~\ref{the:gsimple-best} tells that for positive generative Datalog programs, the ultimate semantics is the one induced by the simple grounder. In Section~\ref{sec:stratified-negation}, where we focus on the central class of generative Datalog$^\naf$ programs with stratified negation, we show that the simple grounder does not suffice, and we provide a more sophisticated grounder that leads to the ultimate semantics.
}

\OMIT{
We conclude this section by stressing that one can come up with simple conditions at the level of grounders that are sufficient for a certain semantics to be as good as some other semantics. For example, one may correctly think that the less ground TGD$^\naf$ are produced by a grounder, the better since such a strategy will classify less possible outcomes as invalid ones. This is shown next.

\begin{proposition}\label{pro:compare-semantics}
	Consider a GDatalog$^\naf[\Delta]$ program $\Pi$, and a database $D$ of $\edb{\Pi}$. Let $G,G'$ be grounders of $\Pi[D]$ with $G(\dep) \subseteq G'(\dep)$ for every $\dep \in [2^{\ground{\dep_{\Pi[D]}^{\exists}}}]_{=}$. $\Pi_G(D)$ is as good as $\Pi_{G'}(D)$.
\end{proposition}
}

\OMIT{
It would be quite useful to be able to classify certain semantics as good as other ones based on a simple property of the underlying grounders. 
%In other words, we would like to have an ``as good as'' notion at the level of grounders that can be then transferred to the induced probabilistic semantics. 
Intuitively, the less ground TGD$^\naf$ are produced by a grounder, the better since such a strategy will classify less possible outcomes as invalid ones. This is formalized by the next definition.

\begin{definition}[\textbf{Compare Grounders}]\label{def:compare-grounders}
	Consider a GDatalog$^\naf[\Delta]$ program $\Pi$, and let $G$ and $G'$ be grounders of $\Pi$. We say that {\em $G$ is as good as $G'$} if $G(\dep) \subseteq G'(\dep)$ for every $\dep \in \left[2^{\ground{\dep_{\Pi}^{\exists}}}\right]_{=}$. \hfill\markfull
\end{definition}

We can now show that indeed the ``as good as'' notion at the level of grounders can be transferred to the induced semantics.

\begin{theorem}\label{the:compare-semantics}
	Consider a GDatalog$^\naf[\Delta]$ program $\Pi$, and a database $D$ of $\edb{\Pi}$. Let $G$ and $G'$ be grounders of $\Pi[D]$ such that $G$ is as good as $G'$. It holds that $\Pi_G(D)$ is as good as $\Pi_{G'}(D)$.
\end{theorem}
}

\section{Fixpoint Probabilistic Semantics}\label{sec:chase}

We now tackle the question whether the probabilistic semantics presented in the previous section can be equivalently defined via a fixpoint procedure. An affirmative answer to this question will provide a procedure that is amenable to practical implementations. To this end, we  introduce a novel chase procedure for generative Datalog$^\naf$, which we then use to define a chase-based probability space that leads to the desired fixpoint probabilistic semantics.

\medskip

\noindent \paragraph{Chasing Generative Datalog$^\naf$ Programs.} We start by introducing our novel chase procedure. Let us stress that our chase procedure deviates from the standard one, which typically operates over databases with the aim of completing an incomplete database as dictated by a given set of TGDs. Our chase procedure operates on ground AtR TGDs with the aim of completing such programs as dictated by a given ground TGD$^\naf$ program.
We start with the notion of trigger application, which corresponds to a chase step.

\begin{definition}[\textbf{Trigger Application}]\label{def:chase-step}
	Consider a GDatalog$^\naf[\Delta]$ program $\Pi$, and let $\dep \in 2^{\ground{\dep_{\Pi}^{\exists}}}$ and $\dep' \in 2^{\ground{\dep_{\Pi}^{\not\exists}}}$.
	%for some GDatalog$^\naf[\Delta]$ program $\Pi$.
	%
	A {\em trigger for $\dep'$ on $\dep$} is an atom $\alpha = \text{\rm Active}_{|\bar q|}^{\delta}(\bar p,\bar q) \in \heads{\dep'}$ such that there is no TGD in $\dep$ of the form $\alpha \ra \text{\rm Return}^{\delta}_{|\bar q|}(\bar p,\bar q,o)$ for some $o \in \ins{C}$.
	An {\em application of $\alpha$ to $\dep$} returns the set of (infinitary) TGD programs $\{\dep_1,\dep_2,\ldots\}$ such that the following hold:
	\begin{itemize}
		\item for each integer $i >0$, $\dep_i = \dep \cup \left\{\alpha \ra \text{\rm Return}^{\delta}_{|\bar q|}(\bar p,\bar q,o)\right\}$, where $o \in \ins{C}$ and $\delta\params{\bar p}(o) > 0$, and
		\item for each constant $o \in \ins{C}$ with $\delta\params{\bar p}(o) > 0$, there exists $i > 0$ such that $\dep_i = \dep \cup \left\{\alpha \ra \text{\rm Return}^{\delta}_{|\bar q|}(\bar p,\bar q,o)\right\}$.
	\end{itemize}
	Such a trigger application is denoted as $\dep \params{\alpha} \{\dep_1,\dep_2,\ldots\}$. \hfill\markfull
\end{definition}

\OMIT{
\begin{definition}[\textbf{Trigger Application}]\label{def:chase-step}
	Consider the (infinitary) TGD$^\naf$ programs $\dep \in \left[2^{\ground{\dep_{\Pi}^{\exists}}}\right]_{=}$ and $\dep' \in 2^{\ground{\dep_{\Pi}^{\not\exists}}}$ for some GDatalog$^\naf[\Delta]$ program $\Pi$.
	A {\em trigger for $\dep'$ on $\dep$} is an atom $\alpha = \text{\rm Active}_{|\bar q|}^{\delta}(\bar p,\bar q) \in \heads{\dep'}$ such that the (partial) function $\mi{AtR}_\dep$ is undefined on $\alpha$.
	An {\em application of $\alpha$ to $\dep$} returns the set of (infinitary) TGD programs $\{\dep_1,\dep_2,\ldots\}$ such that the following hold:
	\begin{itemize}
		\item for each integer $i >0$, $\dep_i = \dep \cup \left\{\alpha \ra \text{\rm Return}^{\delta}_{|\bar q|}(\bar p,\bar q,o)\right\}$, where $o \in \ins{C}$ and $\delta\params{\bar p}(o) > 0$, and
		\item for each constant $o \in \ins{C}$ with $\delta\params{\bar p}(o) > 0$, there exists $i > 0$ such that $\dep_i = \dep \cup \left\{\alpha \ra \text{\rm Return}^{\delta}_{|\bar q|}(\bar p,\bar q,o)\right\}$.
	\end{itemize}
	Such a trigger application is denoted as $\dep \params{\alpha} \{\dep_1,\dep_2,\ldots\}$. \hfill\markfull
\end{definition}
}

Having the notion of trigger application, we can now introduce the central notion of chase tree, which is relative to a grounder.

\begin{definition}[\textbf{Chase Tree}]\label{def:chase-tree}
Consider a GDatalog$^\naf[\Delta]$ program $\Pi$, and a database $D$ of $\edb{\Pi}$. Let $G$ be a grounder of $\Pi[D]$. A {\em chase tree for $D$ w.r.t.~$\Pi$ relative to $G$} is a (possibly infinite) rooted labelled tree $T = (N,E,\lambda)$, where $\lambda : N \ra 2^{\ground{\dep_{\Pi[D]}^{\exists}}}$, such that:
\begin{itemize}
	\item the root node is labelled by $\emptyset$,
	\item for each non-leaf node $v$ (i.e., $v$ has an outgoing edge) with children $v_1,v_2,\ldots$ there exists a trigger $\alpha$ for $G(\lambda(v))$ on $\lambda(v)$ such that $\lambda(v)\params{\alpha}\{\lambda(v_1),\lambda(v_2),\ldots\}$, and
	\item for each leaf node $v$ (i.e., $v$ has no outgoing edges), there is no trigger $\alpha$ for $G(\lambda(v))$ on $\lambda(v)$. \hfill\markfull
\end{itemize}
\end{definition}

A chase tree $T = (N,E,\lambda)$ essentially encodes iterative trigger applications.
A {\em maximal path} of $T$ is either a finite path $v_1,\ldots,v_n$ in $T$, where $v_1$ is the root node and $v_n$ is a leaf node, or an infinite path $v_1,v_2,\ldots$ in $T$, where $v_1$ is the root node.
Let $\paths{T}$ be the set of all maximal paths of $T$; we further write $\mathsf{paths}^{\fin}(T)$ (resp., $\mathsf{paths}^{\infty}(T)$) for the set of finite (resp., infinite) maximal paths of $T$.
The {\em result} of a finite (resp., infinite) maximal path $\pi = v_1,\ldots,v_n$ (resp., $\pi = v_1,v_2,\ldots$) of $T$, denoted $\sem{\pi}$, is the program (resp., infinitary program) $\lambda(v_n)$ (resp., $\bigcup_{i>0} \lambda(v_i)$). 
We further write $\sem{\pi}_G$ for $\sem{\pi} \cup G(\sem{\pi})$.
The notations $\sem{\cdot}$ and $\sem{\cdot}_G$ naturally extend to sets of maximal paths, i.e., for a set $M$ of maximal paths, $\sem{M} = \{\sem{\pi} \mid \pi \in M\}$ and $\sem{M}_G = \{\sem{\pi} \cup G(\sem{\pi}) \mid \pi \in M\}$.
We may simply say chase tree for $D$ w.r.t.~$\Pi$ meaning a chase tree relative to some grounder of $\Pi[D]$.
%; whenever the grounder is important, we will explicitly mention it.
%
The next easy lemma collects some useful properties of chase trees.

%Such an infinite maximal path of $T$ is called {\em fair} if, for each $i >0$, and for every trigger $\alpha$ for $G(\lambda(v_i))$ on $\lambda(v_i)$, there exists $j > i$ such that $\alpha$ is {\em not} a trigger for $G(\lambda(v_j))$ on $\lambda(v_j)$ (i.e., with $\alpha = \text{\rm Active}^{\delta}_{|\bar q|}(\bar p,\bar q)$, $\lambda(v_j)$ contains a TGD of the form $\alpha \ra \text{\rm Result}^{\delta}_{|\bar q|}(\bar p,\bar q,o)$ for some $o \in \ins{C}$ with $\delta\params{\bar p}(o) > 0$). We call the chase tree $T$ {\em fair} if every infinite maximal path of $T$ is fair.
%

%It is not difficult to show that the nodes of a chase tree are labelled with sets of ground AtR TGDs that are functionally consistent.

\begin{lemma}\label{lem:chase-tree-properties}
	Consider a GDatalog$^\naf[\Delta]$ program $\Pi$, and a database $D$ of $\edb{\Pi}$. Let $T = (N,E,\lambda)$ be a chase tree for $D$ w.r.t.~$\Pi$. Then:
	\begin{enumerate}
		\item For every $v \in N$, $\lambda(v) \in [2^{\ground{\dep_{\Pi[D]}^{\exists}}}]_{=}$.
		\item For every $u,v \in N$, $u \neq v$ implies $\lambda(u) \neq \lambda(v)$.
	\end{enumerate}
\end{lemma}

Item (1) essentially tells us that the nodes of a chase tree are labelled with sets of ground AtR TGDs that are functionally consistent, while item (2) establishes an injectivity property.
Another crucial property is the fact that two different chase trees, no matter in which order the triggers are applied, always lead to the same set of programs, that is, the set of programs obtained after collecting the results of their finite paths coincide. This is established by the next technical lemma, which exploits Lemma~\ref{lem:chase-tree-properties}, as well as the monotonicity of grounders for generative Datalog$^\naf$ programs.

\begin{lemma}\label{lem:chase-tree-uniqueness}
	Consider a GDatalog$^\naf[\Delta]$ program $\Pi$, and a database $D$ of $\edb{\Pi}$. Let $G$ be a grounder of $\Pi[D]$, and $T,T'$ be chase trees for $D$ w.r.t.~$\Pi$ relative to $G$. It holds that
	\[
	\sem{\mathsf{paths}^{\fin}(T)}\ =\ \sem{\mathsf{paths}^{\fin}(T')}.
	\]
\end{lemma}

\smallskip
We finally establish the key connection between chase trees and possible outcomes, which heavily relies on Lemma~\ref{lem:chase-tree-uniqueness}.

\OMIT{
Given a chase tree $T$ for a database $D$ w.r.t.~a GDatalog$^\naf[\Delta]$ program $\Pi$ relative to a grounder $G$ of $\Pi[D]$, we define the binary relation 
\[
\mathsf{R}_T\ =\ \left\{(\pi, \sem{\pi} \cup G(\sem{\pi})) \mid \pi \in \mathsf{paths}^\fin(T)\right\}.
\]
over $\mathsf{paths}^\fin(T) \times 2^{\ground{\dep_{\Pi(D)}}}$. By exploiting Lemma~\ref{lem:chase-tree-uniqueness} we can show the following key technical result:

\begin{lemma}\label{lem:bijection}
	Consider a GDatalog$^\naf[\Delta]$ program $\Pi$, and a database $D$ of $\edb{\Pi}$. Let $G$ be a grounder of $\Pi[D]$, and $T$ be chase tree for $D$ w.r.t.~$\Pi$ relative to $G$. $\mathsf{R}_T$ is a bijection from $\mathsf{paths}^\fin(T)$ to $\Omega_{\Pi,G}^{\fin}(D)$.
\end{lemma}
}

\begin{lemma}\label{lem:bijection}
	Consider a GDatalog$^\naf[\Delta]$ program $\Pi$, and a database $D$ of $\edb{\Pi}$. Let $G$ be a grounder of $\Pi[D]$, and $T$ be chase tree for $D$ w.r.t.~$\Pi$ relative to $G$. The binary relation
	\[
	\left\{(\pi, \sem{\pi}_G) \mid \pi \in \mathsf{paths}^\fin(T)\right\}.
	\]
	is a bijection from $\mathsf{paths}^\fin(T)$ to $\Omega_{\Pi,G}^{\fin}(D)$.
\end{lemma}

%\medskip
\noindent \paragraph{Fixpoint Probabilistic Semantics.} We now proceed to introduce the desired fixpoint semantics by exploiting the chase procedure presented above.
Consider a GDatalog$^\naf[\Delta]$ program $\Pi$, and a database $D$ of $\edb{\Pi}$. Let $G$ be a grounder of $\Pi[D]$, and $T$ be a chase tree for $D$ w.r.t.~$\Pi$ relative to $G$. Our goal is to define a probability space $\mi{PS}_T$ over $\paths{T}$, and then show that from $\mi{PS}_T$ we can get a probability space that faithfully mimics $\Pi_G(D)$

We start by observing that each path $\pi \in \paths{T}$ induces a (possibly infinite) subset of $\sms{\dep_{\Pi[D]}}$, that is, $\sms{\sem{\pi}_G}$. Note also that for $\pi,\pi' \in \paths{T}$, it might be the case that $\pi \neq \pi'$ (and thus, by Lemma~\ref{lem:chase-tree-properties}(2), $\sem{\pi} \neq \sem{\pi'}$), but they induce the same subset of $\sms{\dep_{\Pi[D]}}$. We now proceed to define a probability space over $\Omega_T = \paths{T}$. 
Let $F_T$ be the subset of $2^{\paths{T}}$ consisting of $\mathsf{paths}^{\infty}(T)$, and all the maximal subsets $E$ of $\mathsf{paths}^{\fin}(T)$ such that, for all $\pi,\pi' \in E$, $\sms{\sem{\pi}_G} = \sms{\sem{\pi'}_G}$. Let $\mathcal{F}_T$ be the $\sigma$-algebra generated by $F_T$. We finally define the function $P_T : \mathcal{F}_T \ra [0,1]$ as follows.
For every $\pi \in \mathsf{paths}^{\fin}(T)$, let
\[
\mi{Pr}(\pi)\ =\ \prod_{\textrm{Result}_{|\bar q|}^{\delta}(\bar p, \bar q, o) \in \heads{G(\sem{\pi})}}{\delta\params{\bar p}(o)}.
\]
Then, for every (countable) set $E \in \mathcal{F}_T$ with $E \in 2^{\mathsf{paths}^{\fin}(T)}$, let
\[
P_T(E)\ =\ \sum_{\pi \in E}{\mi{Pr}(\pi)}.
\]
Clearly, $\mathsf{paths}^{\fin}(T) \in \mathcal{F}_T$, and thus, we can let
\[
P_T\left(\mathsf{paths}^{\infty}(T)\right)\ =\ 1 - P_T\left(\mathsf{paths}^{\fin}(T)\right).
\]
Finally, $P_T$ extends to countable unions via countable addition.
\OMIT{
Finally, for every (countable) set $E = \mathsf{paths}^{\infty}(T) \cup \bigcup_{i>0} E_i \in \mathcal{F}_T$, 
\[
P_T(E)\ =\ P_T\left(\mathsf{paths}^{\infty}(T)\right) + \sum_{i>0} P_T(E_i).
\]
}
Clearly, $\mi{PS}_T = (\Omega_T,\mathcal{F}_T,P_T)$ is a probability space.
Consider now the triple $\sem{\mi{PS}_T}_G$ obtained from $\mi{PS}_T$ by replacing every $\pi \in \Omega_T$ with $\sem{\pi}_G$. Formally, $\sem{\mi{PS}_T}_G = (\sem{\Omega_T}_G,\sem{\mathcal{F}_T}_G,\sem{P_T}_G)$, where 
\begin{itemize}
	\item $\sem{\mathcal{F}_T}_G = \{\sem{E}_G \mid E \in \mathcal{F}_T\}$, and 
	\item $\sem{P_T}_G : \sem{\mathcal{F}_T}_G \ra [0,1]$ is such that, for every $\sem{E}_G \in \sem{\mathcal{F}_T}$ (by definition, $E \in \mathcal{F}_T$), $\sem{P_T}_G(\sem{E}_G) = P_T(E)$.
\end{itemize}
We claim that $\sem{\mi{PS}_T}_G$ is a probability space that faithfully mimics the probability space $ \Pi_G(D) = \left(\Omega_{\Pi,G}(D),\mathcal{F}_{\Pi,G}^{D},P_{\Pi,G}^{D}\right)$.
More precisely, by exploiting Lemma~\ref{lem:bijection}, we can show the following:

\begin{theorem}\label{the:chase-semantics}
	Consider a GDatalog$^\naf[\Delta]$ program $\Pi$, and a database $D$ of $\edb{\Pi}$. Let $G$ be a grounder of $\Pi[D]$, and $T$ be a chase tree for $D$ w.r.t.~$\Pi$ relative to $G$. The following hold:
	\begin{enumerate}
		\item\label{the:chasesemantics:1} $\sem{\mathsf{paths}^\fin(T)}_G = \Omega_{\Pi,G}^{\fin}(D)$.
		
		\item\label{the:chasesemantics:2} There exists a bijection $g : \sem{\mathcal{F}_T}_G \ra \mathcal{F}_{\Pi,G}^{D}$ such that, for every $E \in \sem{\mathcal{F}_T}_G$, $E \cap \sem{\mathsf{paths}^\fin(T)}_G = g(E) \cap \Omega_{\Pi,G}^{\fin}(D)$.
		
		\item\label{the:chasesemantics:3} For every $E \in \sem{\mathcal{F}_T}_G$, $\sem{P_T}_G(E) = P_{\Pi,G}^{D}(g(E))$. 
	\end{enumerate}
\end{theorem}

\OMIT{
\begin{theorem}\label{the:chase-semantics}
	Consider a GDatalog$^\naf[\Delta]$ program $\Pi$, and a database $D$ of $\edb{\Pi}$. Let $G$ be a grounder of $\Pi[D]$, and $T$ be a chase tree for $D$ w.r.t.~$\Pi$ relative to $G$. The following hold:
	\begin{enumerate}
		\item The function $f : \sem{\mathsf{paths}^\fin(T)} \ra \Omega_{\Pi,G}^{\fin}(D)$ such that $f(\dep) = \dep \cup G(\dep)$, for every $\dep \in \sem{\mathsf{paths}^\fin(T)}$, is a bijection.
		
		\item There exists a bijection $g : \sem{\mathcal{F}_T} \ra \mathcal{F}_{\Pi,G}^{D}$ such that, for every $E \in \sem{\mathcal{F}_T}$, $E \cap \sem{\mathsf{paths}^\fin(T)} = f^{-1}(g(E) \cap \Omega_{\Pi,G}^{\fin}(D))$.
		
		\item For every $E \in \sem{\mathcal{F}_T}$, $\sem{P_T}(E) = P_{\Pi,G}^{D}(g(E))$. 
	\end{enumerate}
\end{theorem}
}

The above theorem essentially tells us that there is en effective way of obtaining the probability space $\Pi_{G}(D)$ from the probability space $\mi{PS}_T$, no matter which chase tree $T$ we consider.

\OMIT{
{\color{red} We need to formalize what we mean by this (we cannot simply say that $\sem{\mi{PS_T}}$ coincides with $\Pi_{G}(D)$ since they differ on the infinite possible outcomes due to the fact that we do not enforce fairness). We then proceed to show the claim. To this end, we need to establish the following technical result, which exploits Lemma~\ref{lem:chase-tree-uniqueness}.

\begin{lemma}	
	The binary relation 
	\[
	\left\{(\pi, \sem{\pi} \cup G(\sem{\pi})) \mid \pi \in \mathsf{paths}^\fin(T)\right\}.
	\]
	%over $\mathsf{paths}^\fin(T) \times 2^{\ground{\dep_{\Pi(D)}}}$.
	is a bijective function of the form $\mathsf{paths}^\fin(T) \ra \Omega_{\Pi,G}^{\fin}(D)$.
\end{lemma}
}

\OMIT{
\begin{lemma}	
	Let $f$ be the binary relation over $\mathsf{paths}^\fin(T) \times 2^{\ground{\dep_{\Pi(D)}}}$ defined as $\{(\pi, \sem{\pi} \cup G(\sem{\pi})) \mid \pi \in \mathsf{paths}^\fin(T)\}$. Then, $f$ is a bijective function from $\mathsf{paths}^\fin(T)$ to $\Omega_{\Pi,G}^\fin(D)$.
\end{lemma}
}
}
%and then establish its uniqueness in the sense that different chase trees lead to the same probability space (this is precisely the place where the monotonicity of grounders for generative Datalog$^\naf$ comes into the picture). The latter uniqueness property allows us to refer to {\em the} chase-based probabilistic semantics. We will then show that the chase-based probabilistic semantics matches (modulo some representation details) the one presented in Section~\ref{sec:gen-datalog}.

\section{Stratified Negation}\label{sec:stratified-negation}

We have seen in Section~\ref{sec:gen-datalog} that for generative Datalog programs without negation, there is no need to go beyond the probabilistic semantics induced by the simple grounder, i.e., the semantics induced by the simple grounder is as good as the semantics induced by any other grounder (Theorem~\ref{the:gsimple-best}).
In this section, we establish an analogous result for the important class of generative Datalog$^\naf$ programs with stratified negation. But let us first recall this class.

\medskip
\noindent \paragraph{Generative Datalog with Stratified Negation.} Consider a GDatalog$^{\naf}$ program $\Pi$. The {\em dependency graph} of $\Pi$, denoted $\dg{\Pi}$, is a multigraph $(V,E)$, where (i) $V = \sch{\Pi}$, (ii) for each $\rho \in \Pi$ with $P$ being the predicate of $H(\rho)$, and for every predicate $R$ occurring in $B^+(\rho)$ (resp., $B^-(\rho)$), there is a {\em positive edge} $(R,P)$ (resp., a {\em negative edge} $(R,P)$) in $E$, and (iii) there are no other edges in $E$.
We say that {\em $P$ depends on $R$ (relative to $\Pi$)} if there is a path in $\dg{\Pi}$ from $R$ to $P$. A {\em cycle} in $\dg{\Pi}$ is a path from a predicate to itself.
We say that the program $\Pi$ has {\em stratified negation} if $\dg{\Pi}$ does not contain a cycle that goes through a negative edge. Henceforth, we refer to such kind of programs as GDatalog$^{\naf s}[\Delta]$ programs.

Interestingly, when the negation is stratified, a certain stratification of the given program can be induced that guarantees that a negative literal is affected only by rules from previous strata. We proceed to formalize this.
A {\em strongly connected component} of $\dg{\Pi} = (V,E)$ is a subset-maximal set $C \subseteq V$ such that $P$ depends on $R$ for every two distinct predicates $P,R \in C$. Let $\scc{\Pi}$ be the set of strongly connected components of $\dg{\Pi}$. A {\em topological ordering} of $\scc{\Pi}$ is a
%\todo{it is not unique, I think ``a'' would be more appropriate} 
linear order $C_1,\ldots,C_n$, for $n = |\scc{\Pi}|$, over $\scc{\Pi}$ such that, for $1 \leq i < j \leq n$, there is no predicate of $C_i$ that depends on one of $C_j$.
We denote by $\Pi_{|C_i}$
%\todo{clash with $\Pi[D]$?} 
the subset of $\Pi$ that keeps only the rules whose head mentions a predicate from $C_i$.

\medskip
\noindent \paragraph{Perfect Grounding for GDatalog$^{\naf s}[\Delta]$.} Unlike positive programs, we can show that for generative Datalog$^{\naf s}$ programs there exists a better grounder than the simple one.
%
%the simple grounder for generative Datalog$^{\naf s}$ programs is not as good as any other grounder. enough even if the negation is stratified.
%\todo[inline]{I'm not sure to get this point. I would change as:The simple grounder for generative Datalog$^\naf$ programs can be improved if the negation is stratified.
%}
More precisely, for a GDatalog$^{\naf s}$ program $\Pi$, and a database $D$ of $\edb{\Pi}$, there exists a grounder $G$ of $\Pi[D]$ such that $\Pi_G(D)$ is as good as $\Pi_{\mi{GSimple}_{\Pi[D]}}(D)$, but not vice versa. We proceed to introduce a grounder that leads to the ultimate semantics in the case of stratified negation.

Given a GDatalog$^{\naf s}[\Delta]$ program $\Pi$, we define the so-called perfect grounder $\mi{GPerfect}_{\Pi}$, which can be think as the version of the simple grounder (see Definition~\ref{def:simple-grounder}) that considers the rules of $\Pi$ in a certain order according to the induced stratification of $\Pi$.
Similarly to the simple grounder, $\mi{GPerfect}_{\Pi}$ relies on an operator that operates on (infinitary) existential-free TGD$^\naf$ programs. Consider such a program $\dep$. We define $\mi{Perfect}_{\dep}$ as the monotonic function from $2^{\ground{\dep}}$ to $2^{\ground{\dep}}$ such that, for every $\dep' \in 2^{\ground{\dep}}$,
\begin{multline*}
\mathit{Perfect}_\dep(\dep')\ =\ \dep' \cup \left\{
h(\sigma)\ \mid\ 
\sigma \in \dep,\ h(B^+(\sigma)) \subseteq \heads{\dep'}\ \text{and}\right.\\
\left.h(B^-(\sigma)) \cap \heads{\dep'} = \emptyset)
\right\}.
\end{multline*}
The set $\mathit{Perfect}_{\dep}^{\infty}(\dep')$ is defined as expected; due to the monotonicity of $\mathit{Perfect}_{\dep}$, it is the least fixpoint of $\mathit{Perfect}_{\dep}$ that contains $\dep'$.

Assume now that $C_1,\ldots,C_n$ is a topological ordering over $\scc{\Pi}$. For a set $\dep \in [2^{\ground{\dep_{\Pi}^{\exists}}}]_{=}$, we define $\dep \uparrow C_i \in 2^{\ground{\dep_{\Pi}^{\not\exists}}}$, which, intuitively speaking, is the set assigned to $\dep$ by the perfect grounder considering the rules stratum by stratum up to the $i$-th stratum. Eventually, $\mi{GPerfect}_{\Pi}$ will assign to $\dep$ the (infinitary) TGD$^\naf$ program $\dep \uparrow C_n$.
By convention, let $\dep \uparrow C_0 = \emptyset$. Then, for every $i \in [n]$, with $\dep'_i = \dep_{\Pi_{|C_i}}^{\not \exists} \cup \dep \uparrow C_{i-1} \cup \dep$,
\begin{align*}
\dep \uparrow C_{i}\ =\
\begin{cases}
\mathit{Perfect}_{\dep'_i}^{\infty}(\emptyset) \setminus \dep, & \mbox{if } \mi{AtR}_\dep \comp \dep \uparrow C_{i-1}\\
%& \\
\dep \uparrow C_{i-1}, & \mbox{otherwise}.
\end{cases}
\end{align*}
%\todo{should be $Perfect(\emptyset) \setminus \dep$?}
We are now ready to introduce the function $\mi{GPerfect}_{\Pi}$.

\begin{definition}[\textbf{Perfect Grounder}]\label{def:perfect-grounder}
	Consider a GDatalog$^{\naf s}[\Delta]$ program $\Pi$, and let $C_1,\ldots,C_n$ be a topological ordering over $\scc{\Pi}$. The function $\mi{GPerfect}_{\Pi} : [2^{\ground{\dep_{\Pi}^{\exists}}}]_{=} \ra 2^{\ground{\dep_{\Pi}^{\not\exists}}}$ is such that
	$\mi{GPerfect}_{\Pi}(\dep) = \dep \uparrow C_n$, for every $\dep \in [2^{\ground{\dep_{\Pi}^{\exists}}}]_{=}$. \hfill\markfull
\end{definition}

An example that illustrates $\mi{GPerfect}_{\Pi}$ can be found in the appendix.
We can show that $\mi{GPerfect}_{\Pi}$ is indeed a grounder. This relies on a technical lemma that establishes the following: every (infinitary) TGD$^\naf$ program $\dep$ that belongs to the $\mi{GPerfect}_{\Pi}$-grounding of $\Pi$ has a {\em unique} stable model, namely the instance $\heads{\dep}$.

\begin{proposition}\label{pro:gperfect-grounder}
	Consider a GDatalog$^{\naf s}[\Delta]$ program $\Pi$. It holds that $\mi{GPerfect}_{\Pi}$ is a grounder of $\Pi$.
\end{proposition}

We now show the main result of this section: the perfect grounder leads to the ultimate semantics when the negation is stratified.

\begin{theorem}\label{the:gperfect-best}
	Consider a GDatalog$^{\naf s}[\Delta]$ program $\Pi$, and a database $D$ of $\edb{\Pi}$. Let $G$ be an arbitrary grounder of $\Pi[D]$. It holds that $\Pi_{\mi{GPerfect}_{\Pi[D]}}(D)$ is as good as $\Pi_{G}(D)$.
\end{theorem}

\OMIT{
%\medskip
\noindent \paragraph{Invariance Under SMS-Equivalence.}  At this point, one may wonder whether an invariance property, analogous to Theorem~\ref{the:stable-invariance-simple} for positive programs, can be established in the case of stratified negation by adopting the perfect grounder. We proceed to show that this is the case.
Recall that two GDatalog$^\naf[\Delta]$ programs $\Pi$ and $\Pi'$ are SMS-equivalent if $\sms{\dep_\Pi} = \sms{\dep_{\Pi'}}$.
The notion of {\em perfect-semantic equivalence} is defined in exactly the same way as the notion of simple-semantic equivalence, with the difference that the perfect grounder is adopted.
We can then establish the following implication; note that the converse does not hold in general:

\begin{theorem}\label{the:stable-invariance-perfect}
	If two GDatalog$^{\naf s}[\Delta]$ programs are SMS-equivalent, then they are perfect-semantic equivalent.
\end{theorem}

This tells us that rewriting a GDatalog$^{\naf s}[\Delta]$ program in a way that preserves SMS-equivalence is
safe, as the probabilistic semantics induced by the perfect grounder is in essence preserved.
}

\section{Conclusions}\label{sec:conclusions}

We have introduced generative Datalog with negation that allows sampling from discrete probability distributions, and defined its probabilistic semantics by interpreting negation according to the standard stable model semantics. We have also shown that the semantics can be equivalently defined via a fixpoint procedure based on a new chase procedure that operates on ground programs. 
Here we focused on the generative component, but one may also incorporate conditioning under events of positive probability just as in~\cite{BCKOV17}.
An interesting direction for future research, which is highly relevant for practical implementations, is to devise sophisticated grounders that avoid as much as possible the generation of superfluous ground rules.
Furthermore, in view of the fact that continuous distributions appear in a variety of application scenarios for probabilistic databases, it is natural to extend generative Datalog with negation to support continuous distributions.

\bibliographystyle{ACM-Reference-Format}

%\bibliography{references}


\balance

%\ifdefined\withappendix
\newpage
\appendix
\onecolumn

\section{Related Work}\label{sec:apx-related-work}

The probabilistic programming and the probabilistic database communities have developed several different models and systems that allow to specify probability distributions over data.
Since the paper by B\'ar\'any et al.~\cite{BCKOV17}, which introduced generative Datalog, already presents a thorough discussion on related work (see Section 7 of~\cite{BCKOV17}), to which we refer the reader, we proceed to discuss only the related work that is closer to our work.
Conceptually closer to generative Datalog (with or without negation) are languages studied in statistical relational AI (StarAI)~\cite{2016Raedt}, which aim at the combination of predicate logic and probabilities.
Here are some prominent examples of such languages:
\begin{itemize}
\item ProbLog~\cite{RaKT07,FBRSGTJR15} is a probabilistic extension of Prolog where standard Prolog rules can be annotated with a probability value (i.e., it allows uncertainty at the level of rules).
Hybrid ProbLog~\cite{GuJR10} is an extension of ProbLog that allows continuous attribute-level uncertainty in rule-heads.

\item Probabilistic Answer Set Programming~\cite{ijar:CozmanM20} extends the
  well-known logic programming language of Answer Set
  Programming~\cite{GeLi88,ngc:GelfondL91} with probabilistic facts. The
  semantics of this language is defined via upper and lower probability bounds
  for stable models. A number of similar extensions have been proposed in the
  literature, with the general idea to annotate facts or rules with
  probabilities~\cite{tocl:Lukasiewicz01,lpnmr:WanK09,tplp:BaralGR09}. A more
  complete overview is given by Cozman and Mau{\'{a}}~\cite{ijar:CozmanM20}.

\item Markov Logic Networks (MLNs)~\cite{RiD06} combine first-order logic and Markov Networks, and they describe joint distributions of variables based on weighted first-order constraints.
Hybrid MLNs~\cite{WaDo08} extend MLNs to continuous distributions.
Infinite MLNs~\cite{SiDo07} allow for countably
infinitely many variables with countable domains.

\item Probabilistic Soft Logic~\cite{KimmingBBHG12} is another statistical relational learning framework that combines first-order logic and probabilistic graphical models that specifies joint distributions with weighted rules, but also ``soft'' truth values.

\item Another formal language that combines first-order logic and probabilistic graphical models, which also supports continuous distributions, is Bayesian Logic~\cite{MMRSOK05}, which in turn builds on Bayesian Networks.

\item A probabilistic extension of Datalog$^\pm$ (a family of knowledge representation languages that enrich Datalog with features such as value invention and equality in head-rules~\cite{CGLMP10}) that is based on MLNs as underlying probabilistic semantics has been proposed in~\cite{GLMS13}.
Various probabilistic extensions of plain Datalog can be also found in~\cite{DeKM10,Fuhr95,Fuhr00}.
\end{itemize}

Although the probabilistic formalisms discussed above (as well as those discussed in Section 7 of~\cite{BCKOV17}), might share individual features with generative Datalog (with or without negation), the combination of plain Datalog with probabilistic programming languages should be attributed to~\cite{BCKOV17}.
Moreover, the fact that generative Datalog (with or without negation) can express attribute-level uncertainty in rule-heads (via the $\Delta$-terms appearing in rule-heads) differs from other probabilistic extensions of Datalog, where the uncertainty is at the at level of rules specified by attaching probabilities to them (like, for example, the languages from~\cite{GLMS13} and~\cite{DeKM10,Fuhr95,Fuhr00}).

\section{Section~\ref{sec:preliminaries}: Preliminaries}\label{sec:apx-preliminaries}

\subsection*{Parameterized Probability Distribution - Example}

We provide an example that illustrates the notion of parameterized probability distribution.
Throwing a die is a classical example of a random process. The throwing of an {\em unbiased} die can be modelled via the discrete probability space $(\Omega,P)$, where $\Omega = \{1,\ldots,6\}$, the six faces of a die, and $P$ is such that $P(i) = \frac{1}{6}$ for each $i \in \Omega$.
On the other hand, we can model the throwing of a {\em biased} die by using a parameterized probability distribution. In particular, this can be done via $\mathsf{Die} : \mathbb{R}^6 \ra \mathcal{P}_{\Omega}$ with $\Omega = \{0,1,\ldots,6\}$ defined as follows: for every $\bar p = (p_1,\ldots,p_6) \in \mathbb{R}^6$, $\sum_{i=1}^{6} p_i = 1$ implies $\mathsf{Die}\params{\bar p}(0) = 0$ and $\mathsf{Die}\params{\bar p}(i) = p_i$ for each $i \in [6]$; otherwise, $\mathsf{Die}\params{\bar p}(0) = 1$ and $\mathsf{Die}\params{\bar p}(i) = 0$ for each $i \in [6]$. Note that the outcome $0$ is associated with incorrect instantiations of the parameters.

\section{Section~\ref{sec:gen-datalog}: Generative Datalog with Negation}\label{sec:apx-gen-datalog}

\newcommand{\bckov}{\ensuremath{\mathrm{BCKOV}}\xspace}
\newcommand{\gsimppi}{\ensuremath{\mi{GSimple}_{\Pi[D]}}\xspace}
\newcommand{\gsimppionly}{\ensuremath{\mi{GSimple}_{\Pi}}\xspace}

\subsection*{Proof of Proposition~\ref{pro:gsimple-grounder}}

Consider a set $\dep \in [2^{\ground{\dep_{\Pi}^{\exists}}}]_{=}$ such that $\mi{AtR}_\dep \comp \gsimppionly(\dep)$. We need to show that, for every totalizer $\dep'$ of $\mi{AtR}_\dep$, it holds that
\[
\sms{\gsimppionly(\dep)\cup \dep}\ =\ \sms{\dep^{\not\exists}_\Pi \cup \dep'}.
\]
Fix an arbitrary totalizer $\dep'$ of $\mi{AtR}_\dep$. We proceed to show the above equality. For brevity, let $\dep_1 = \gsimppionly(\dep)\cup \dep$ and $\dep_2 = \dep^{\not\exists}_\Pi \cup \dep'$.

$(\subseteq)$ Consider an instance $I \in \sms{\dep_1}$. We need to show that $I \in \sms{\dep_2}$, or, equivalently, that $I$ is a model of
\[
\SM[\dep_2]\ =\ \underbrace{\bigwedge_{\sigma \in \dep_2} \sigma}_{\Phi}\ \wedge\ \neg \underbrace{\exists \ins{X}_{\dep_2} \left( (\ins{X}_{\dep_2} < \ins{R}) \wedge \repl_{\dep_2}(\dep_2)\right)}_{\Psi}.
\]
Note that, for each $\sigma \in \dep_\Pi^{\not\exists}$, and for each $h$ such that $h(\sigma) \notin \dep_1$, it holds that $B^+(h(\sigma)) \not\subseteq \heads{\dep_1} \supseteq I$.
Therefore, for every $h$ and $\sigma$ such that $B^+(h(\sigma)) \subseteq I$, it holds that $B^-(h(\sigma)) \cap I = \emptyset$ since $I$ is a stable model of $\dep_1$; thus, $I$ is a model of every $\sigma \in \dep_\Pi^{\not\exists}$. 
The same can be observed for $\dep' \supseteq \dep$ which contains no negation, and thus, $I$ is a model of every $\sigma \in \dep'$. Consequently, $I$ is a model of $\Phi$.
Assume now, by contradiction, that $I$ is a model of $\Psi$. This would imply that there exists an instance $J \subsetneq I$ that is a model of each TGD of the set obtained from $\{\sigma \in \ground{\dep_1} \mid B^-(\sigma) \cap I = \emptyset\}$ by eliminating all the negative literals. This in turn implies that $I \not\in \sms{\dep_1}$, which is not possible. Hence, $I$ is a model of $\neg \Psi$, which in turn implies that $I$ is a model of $\SM[\dep_2]$, as needed.

%Moreover, $I \models \exists \ins{X_\dep} \left((\ins{X_\dep} < \ins{R}) \land \repl(\dep_\Pi^{\not\exists} \cup \dep')\right)$ would imply the existence of $J \subset I$ such that $J \models (\gsimppionly(\dep) \cup \dep)^I$, which is not possible.

$(\supseteq)$ This direction is similar. Let $I \in \sms{\dep_2}$, and we need to show that $I \in \sms{\dep_1}$, or, equivalently, that $I$ is a model of
\[
\SM[\dep_1]\ =\ \underbrace{\bigwedge_{\sigma \in \dep_1} \sigma}_{\Phi'}\ \wedge\ \neg \underbrace{\exists \ins{X}_{\dep_1} \left( (\ins{X}_{\dep_1} < \ins{R}) \wedge \repl_{\dep_1}(\dep_1)\right)}_{\Psi'}.
\]
Since all the TGDs of $\dep_1$ are obtained by applying some homomorphism to TGDs of $\dep_2$, we get that $I$ is a model of $\Phi'$.
Moreover, as above, assuming that $I$ is a model of $\Psi'$, allows us to conclude that $I \notin \sms{\dep_2}$, which contradicts our hypothesis. Therefore, $I$ is a model of $\neg \Psi'$, which in turn implies that $I$ is a model of $\SM[\dep_1]$, as needed.

%the existence of $J \subset I$ such that $J \models (\gsimppionly(\dep)\cup \dep)^I$ would imply $J \models \exists \ins{X_\dep} \left((\ins{X_\dep} < \ins{R}) \land \repl(\dep_\Pi^{\not\exists} \cup \dep')\right)$, which is not possible.

\subsection*{Proof of Theorem~\ref{the:semantics}}
Recall that $\Pi_{G}^{D}$ is defined as the triple $\left(\Omega_{\Pi,G}(D),\mathcal{F}_{\Pi,G}^{D},P_{\Pi,G}^{D}\right)$. 
By construction, $\mathcal{F}_{\Pi,G}^{D}$ is a $\sigma$-algebra over the sample space $\Omega_{\Pi,G}(D)$.
It remains to show that the function $P_{\Pi,G}^{D} : \mathcal{F}_{\Pi,G}^{D} \ra [0,1]$ is a probability measure. By definition,
\[
\Omega_{\Pi,G}(D)\ =\ \Omega_{\Pi,G}^{\fin}(D) \cup \Omega_{\Pi,G}^{\infty}(D) \qquad \text{and} \qquad P_{\Pi,G}^{D}\left(\Omega_{\Pi,G}^{\infty}(D)\right) = 1 - P_{\Pi,G}^{D}\left(\Omega_{\Pi,G}^{\fin}(D)\right).
\]
Therefore, we immediately get that
\[
P_{\Pi,G}^{D}\left(\Omega_{\Pi,G}(D)\right)\ =\ 1.
\]
Observe now that the members of $F_{\Pi,G}^D$ are pairwise disjoint, and that $F_{\Pi,G}^D$ itself is countable. The induced $\sigma$-algebra $\mathcal{F}_{\Pi,G}^D$ will therefore only contain events of the form $E = \bigcup_{S \in F_E} S$ with $F_E \subseteq F_{\Pi,G}^D$ that are countable; the complement of any such event $E$ is the event $\bigcup_{S \in \left(F_{\Pi,G}^D \setminus F_E \right)} S$.   
Hence, for any countable set $\mathcal{E}$ of pairwise disjoint events, we can conclude that
\[
P_{\Pi,G}^D \left(\bigcup_{E \in \mathcal{E}} E \right)\ =\  P_{\Pi,G}^D \left(\bigcup_{E\in \mathcal{E}} \bigcup_{S \in F_{E}} S \right)\ =\ \sum_{E\in \mathcal{E}} \sum_{S \in F_E} P_{\Pi,G}^D \left(S \right)\ =\ \sum_{E\in \mathcal{E}} P_{\Pi,G}^D \left(E \right),
\]
as needed. Note that, since all the sets in $\mathcal{E}$ are pairwise disjoint, also $F_E \cap F_{E'} = \emptyset$ for any distinct  $E, E'$ in $\mathcal{E}$.

\subsection*{Positive Programs}

As briefly discussed in the main body of the paper (in Section~\ref{sec:gen-datalog}), we can show that the semantics induced by the simple grounder coincide with the semantics of ~\citet{BCKOV17} for positive generative Datalog whenever we focus on positive programs that guarantee that all the possible outcomes are finite. We proceed to formally state and prove this claim. To this end, we first recall the semantics for positive generative Datalog from~\cite{BCKOV17}, and then establish the above claim.

\medskip

\noindent\underline{\textbf{Existing Semantics for Positive Generative Datalog}}

\smallskip

%\subsubsection*{Existing Semantics}
%
\noindent To distinguish our semantics from the semantics of~\citet{BCKOV17}, we will refer to their semantics as \bckov semantics (in reference to the authors of \cite{BCKOV17}), and their possible outcomes as \bckov possible outcomes. Below we give an abridged version of the key notions of the \bckov semantics for positive GDatalog$[\Delta]$ programs.
%
%Let $\Pi$ be a (positive) GDatalog$[\Delta]$ program. 
Similar to the definition of our semantics, also the \bckov semantics rely on a translation into a TGD program, which is slightly different than our translation presented in Section~\ref{sec:gen-datalog}. Let us recall their translation.
%
%$\Sigma_\mathcal{G}$, programs in \bckov semantics are also rewritten into TGDs first. In particular, 
Consider a (positive) GDatalog$[\Delta]$ rule $\rho$
\[
R_1(\bar u_1),\ldots,R_n(\bar u_n)\ \ra\ R_0(\bar w),
\]
where $\bar w = (w_1,\ldots,w_{\ar{R_0}})$ with $w_{i_1} = \delta_1\params{\bar p_{1}}[\bar q_{1}],\ldots,w_{i_r} = \delta_r\params{\bar p_{r}}[\bar q_{r}]$, for $1 \leq i_1 < \cdots < i_r \leq \ar{R_0}$, be all the $\Delta$-terms in $\bar w$.\footnote{Strictly speaking, \cite{BCKOV17} assumes that there is only one $\Delta$-term per rule head. This is an assumption that can be made without loss of generality as explained in~\cite{BCKOV17}. For the sake of consistency, though, we consider multiple $\Delta$-terms in the head.} 
%Having only one $\Delta$-term in the head this definition matches exactly the presentation in~\cite{BCKOV17} .}
%
If $r=0$ (i.e., there are no $\Delta$-terms in $\bar w$), then $\widetilde{\rho_{\exists}}$ is defined as the singleton set $\{\dep_\rho\}$, where $\dep_\rho$ is the TGD without existentially quantified variables obtained after converting $\rho$ into a first-order sentence in the usual way. In fact, $\widetilde{\rho_{\exists}}$ coincides with $\rho_{\exists}$ defined in Section~\ref{sec:gen-datalog}.
Otherwise, $\widetilde{\rho_{\exists}}$ is defined as the set consisting of the following TGDs:
\begin{eqnarray*}
R_1(\bar u_1),\ldots,R_n(\bar u_n)\ \ra\ \exists y_j\ \textrm{Result}_{|\bar q_j|}^{\delta_j}(\bar p_{j},\bar q_{j},y_j)
\end{eqnarray*}
for each $j \in [r]$, and
\[
\text{\rm Result}_{|\bar q_1|}^{\delta_1}(\bar p_{1},\bar q_{1},y_1),\ldots, \text{\rm Result}_{|\bar q_r|}^{\delta_r}(\bar p_{r},\bar q_{r},y_r),
R_1(\bar u_1),\ldots,R_n(\bar u_n)\ \ra\ 
R_0(\bar w'),
\]
where $\textrm{Result}_{|\bar q_j|}^{\delta_j}$ are fresh $(|\bar p_j|+|\bar q_j|+1)$-ary predicates, not occurring in $\sch{\Pi}$, $y_1,\ldots,y_r$ are distinct variables not occurring in $\rho$, and $\bar w' = (w_1,\ldots,w_{i_1-1},y_1,w_{i_1+1},\ldots,w_{i_r-1},y_r,w_{i_r+1},\ldots,w_{\ar{R_0}})$ is obtained from $\bar w$ by replacing $\Delta$-terms with variables $y_1,\ldots,y_r$.
Finally, for a (positive) GDatalog$[\Delta]$ program $\Pi$, we define $\widetilde{\dep_\Pi}$ as the TGD program $\bigcup_{\rho \in \Pi} \widetilde{\rho_\exists}$.
Note that the key difference of the above translation compared to our translation given in Section~\ref{sec:gen-datalog} is that the intermediate active-to-result TGDs are missing.
We can now recall the notion of possible outcome as defined in~\cite{BCKOV17}:

%Let $\widehat{\mathcal{G}}$ be the set of TGDs $\bigcup{\rho\in \mathcal{G}} \rho_\exists$.

\begin{definition}[\textbf{BCKOV Possible Outcome}]
	Consider a GDatalog$[\Delta]$ program $\Pi$, and a database $D$ of $\edb{\Pi}$. A {\em \bckov possible outcome} of $D$ w.r.t.~$\Pi$ is a minimal model $I$ of $\widetilde{\dep_\Pi}$ such that $\delta\params{\bar p}(o) > 0$ for every $\Delta$-atom $\textrm{Result}_{|\bar q|}^{\delta}(\bar p,\bar q,o) \in I$. We denote the set of all \bckov possible outcomes of $D$ w.r.t.~$\Pi$ as $\Omega^\bckov_{\Pi}(D)$. \hfill\markfull
 \end{definition}

%A \emph{\bckov possible outcome} for GDatalog$[\Delta]$ program $\mathcal{G}$, and database $D$ of $\edb{\mathcal{G}}$, is minimal model $J$ of $\widehat{\mathcal{G}}$ such that $\delta\params{\bar p}(o) > 0$ for every $\Delta$-atom $\textrm{Result}_{|\bar q|}^{\delta}(\bar p,\bar q,o) \in J$. We denote the set of all \bckov possible outcomes of $\mathcal{G}$ and $D$ as $\Omega^\bckov_{\mathcal{G}}(D)$.

Consider a GDatalog$[\Delta]$ program $\Pi$, and a database $D$ of $\edb{\Pi}$. With $F$ being a set of atoms, we write $\Omega_{\Pi}^{F\subseteq}(D)$ for the set of possible outcomes
$\{ J \in \Omega^\bckov_{\Pi}(D) \mid F \subseteq  J\}$, that is, the set that collects all those possible outcomes that contain all the facts of $F$.
Let $\Omega_{\Pi}^{\subseteq_{\fin}}(D)$ be the set of sets $\Omega_{\Pi}^{F\subseteq}(D)$ where $F$ is finite.
A finite sequence of atoms $\mathbf{f} = (f_1,f_2,\dots,f_n)$ is a derivation if, for all $i\in [n]$, $f_i = \head{h(\sigma)}$ for some TGD $\sigma \in \widetilde{\dep_\Pi}$ and homomorphism $h$ from $B(\sigma)$ to $J_{i-1}$, where $J_{k} = D \cup \{f_1,f_2,\dots,f_k\}$. 
For a derivation $f_1,\dots,f_n$, the set $\{f_1,\dots,f_n\}$ is called a derivation set. Intuitively, a finite derivation set is an intermediate instance of some chase tree. The probability of a set of atoms $F$ is defined in essentially the same way as our definition of $\mi{Pr}(\dep)$. Formally, for an atom $f$ of the form $\mathrm{Result}_n^\delta(\bar p, \bar q, o)$, let $\mi{Pr}^\bckov(f)=\delta\left<\bar p\right>(o)$. If $f$ is of any other form, then $\mi{Pr}^\bckov(f)=1$. For a set of atoms $F$, let $P^\bckov(F) = \prod_{f\in F} \mi{Pr}^\bckov(f)$.
The following holds:
%We are now ready to recall the output of a generative Datalog program $\Pi$ on a database $D$ according to the \bckov semantics.

\begin{theorem}[\cite{BCKOV17}]
	\label{thm:todsprob}
	Consider a GDatalog$[\Delta]$ program $\Pi$, and a database $D$ of $\edb{\Pi}$. There exists a unique probability space $(\Omega,\mathcal{F},P)$ that satisfies all of the following:
	\begin{enumerate}
		\item $\Omega\ =\ \Omega^\bckov_{\Pi}(D)$,
		\item $\mathcal{F}$ is the $\sigma$-algebra generated by $\Omega_{\Pi}^{\subseteq_{\fin}}(D)$, and
		\item $P(\Omega_{\Pi}^{F\subseteq})\ =\ P^\bckov(F)$ for every derivation set $F$.
	\end{enumerate}
\end{theorem}

For a (positive) GDatalog$[\Delta]$ program $\Pi$, and a database $D$ of $\edb{\Pi}$, we denote by
\[
\Pi^\bckov(D)\ =\ \left(\Omega^\bckov_{\Pi}(D), \mathcal{F}^\bckov_{\Pi,D}, P^\bckov_{\Pi,D}\right)
\]
the unique probability space provided by Theorem~\ref{thm:todsprob}, which forms the {\em output of $\Pi$ on $D$} according to the \bckov semantics of~\cite{BCKOV17}.

\OMIT{
\begin{theorem}[\cite{BCKOV17}]
  \label{thm:todsprob}
  Consider a GDatalog$[\Delta]$ program $\Pi$, and a database $D$ of $\edb{\Pi}$. There exists a unique probability space $(\Omega^\bckov_\mathcal{G}, \mathcal{F}^{\bckov}_\mathcal{G}, P^\bckov_\mathcal{G})$, that satisfies all of the following:
  \begin{enumerate}
  \item $\Omega^{\bckov} = \Omega^\bckov_{\mathcal{G}}(D)$,
  \item $\mathcal{F}^\bckov_\mathcal{G} = \sigma(\Omega_{\mathcal{G}}^{\mathit{fin}\subseteq}(D))$,
  \item and $P^\bckov_\mathcal{G}(\Omega_{\mathcal{G}}^{F\subseteq}) = \mathbf{P}^\bckov(F)$ for every derivation set $F$.
  \end{enumerate}
\end{theorem}
}

%For the rest of this section we will usually consider the GDatalog$[\Delta]$ program $\mathcal{G}[D]$ instead of separate $\mathcal{G}$ and $D$ with the natural understanding that this refers to the program $\mathcal{G}[D]$ with the empty database.

\medskip

\noindent\underline{\textbf{BCKOV Semantics vs. Our Semantics}}

\smallskip
%\subsubsection*{Existing Semantics vs. Our Semantics}

\noindent We move on to show that for positive programs (i.e., programs with no negation) where all possible outcomes are finite, our semantics induced by the simple grounder coincide with the \bckov semantics, that is, the resulting probability spaces are isomorphic.

\begin{definition}[\textbf{Finite Grounding}]
  Consider a GDatalog$^\naf[\Delta]$ program $\Pi$, a database $D$ of $\edb{\Pi}$, and a grounder $G$ of $\Pi[D]$. We say that $G$ is \emph{finitely grounding} (for $\Pi[D]$) if every possible outcome of $D$ w.r.t.~$\Pi$ relative to $G$ is finite. \hfill\markfull
\end{definition}

We further need the notion of isomorphism between probability spaces.
An \emph{isomorphism} between two probability spaces $PS_1 = (\Omega_1, \mathcal{F}_1, P_1)$ and $PS_2 = (\Omega_2, \mathcal{F}_2, P_2)$ is a bijection $f \colon \Omega_1 \to \Omega_2$ such that (i) for every $F \in 2^{\Omega_1}$, $F \in \mathcal{F}_1$ iff $f(F) \in \mathcal{F}_2$, and (ii) for every $F \in \mathcal{F}_1$, $P_1(F) = P_2(f(F))$. 
% we can cite An Introduction to Probability Theory by K. Ito if we want to
%
We say that $PS_1$ and $PS_2$ are {\em isomorphic}, denoted $PS_1 \simeq PS_2$, if there exists an isomorphism between them.
We are now ready to state the desired result:

\begin{theorem}\label{the:isomorphic-spaces}
Consider a GDatalog$[\Delta]$ program $\Pi$, and a database $D$ of $\edb{\Pi}$ such that $\gsimppi$ is finitely grounding. Then
\[
\Pi_{\gsimppi}(D)\ \simeq\ \Pi^{\bckov}(D).
\]
\end{theorem}

The rest of the section is devoted to proving the above theorem. To this end, we first need to establish some auxiliary technical lemmas.
Recall that our translation of a GDatalog$[\Delta]$ progrma into a TGD program is the same as the translation used in the definition of the \bckov semantics, except for the addition of the intermediate ``active'' predicates, and the respective active-to-result TGDs that link the ``active'' with the ``result'' predicates. It will therefore be useful to talk about stable models of a TGD program $\dep$ without the ``active'' predicates. In particular, for a stable model $I$, we refer to the subset of $I$ after eliminating all the atoms of the form $\textrm{Active}_{|\bar q|}^{\delta}(\bar p,\bar q)$ as $I$ {\em modulo active}.

\begin{lemma}
	\label{lem:positivemodels}
	Consider a GDatalog$[\Delta]$ program $\Pi$, and a database $D$ of $\edb{\Pi}$. The following hold:
	\begin{enumerate}
		\item For every possible outcome $\dep \in \Omega_{\Pi,\gsimppi}(D)$, $\Sigma$ has exactly one stable model. 
		\item For every distinct $\dep,\dep' \in \Omega_{\Pi,\gsimppi}(D)$, the stable model of $\Sigma$ modulo active and the stable model of $\dep'$ modulo active are different.
	\end{enumerate}
\end{lemma}

\begin{proof}
	{\em Item (1).} Since $\Pi$ is positive, so is the TGD program $\dep_\Pi$. Therefore, every possible outcome $\dep \in \Omega_{\Pi,\gsimppi}(D)$ is positive and without existential quantification. It is well-known that such programs have exactly one stable model; thus, $\sms{\Sigma}$ is a singleton.
	
	{\em Item (2).} Let $\Sigma, \Sigma'$ be two distinct possible outcomes of $\Omega_{\Pi,\gsimppi}(D)$. This means that $\Sigma$ and $\Sigma'$ belong to $\gsimppi$-$\ground{\Pi[D]}$, and therefore, $\Sigma$ and $\Sigma'$ are generated by distinct terminals $T$, $T'$, i.e., different subsets of $[2^{\ground{\dep_{\Pi[D]}^{\exists}}}]_{=}$ that are incomparable by $\subseteq$.
	Since $T$ and $T'$ are minimal w.r.t. $\terminals{\gsimppi}$, all body atoms in $T$ and $T'$ occur in the heads of $\gsimppi(T)$ and $\gsimppi(T')$, respectively. Hence, by definition of $\gsimppi$, the single stable model of $\Sigma$ contains exactly $\heads{T}$  as the only atoms with a predicate of the form $\textrm{Result}^\delta_n$ (the ``result'' atoms), while the stable model of $\Sigma'$ contains exactly $\heads{T'}$ as its only ``result'' atoms. Since $T$ and $T'$ are incomparable by $\subseteq$, so are $\heads{T}$ and $\heads{T'}$. Hence, the stable models of $\dep$ and $\dep'$ modulo active are different, as needed.
\end{proof}

\begin{lemma}
	\label{lem:PitoG}
	Consider a GDatalog$[\Delta]$ program $\Pi$, and a database $D$ of $\edb{\Pi}$. For every $\dep \in \gsimppi\text{-}\ground{\Pi[D]}$, the following statements are equivalent:
	\begin{enumerate}
	\item $\dep \in \Omega_{\Pi,\gsimppi}(D)$.
	\item The stable model of $\dep$ modulo active belongs to $\Omega_{\Pi}^{\bckov}(D)$.
	\end{enumerate}
\end{lemma}

\begin{proof}
	$(1) \Rightarrow (2)$. By definition, $\dep = \dep' \cup \gsimppi(\dep')$ for some minimal $\dep' \in \terminals{\gsimppi}$. We define the instance
	\[
	I\ =\ D \cup \left\{\textrm{Result}_{|\bar q|}^{\delta}(\bar p,\bar q, o) \mid \textrm{Active}_{|\bar q|}^{\delta}(\bar p,\bar q) \ra \textrm{Result}_{|\bar q|}^{\delta}(\bar p,\bar q, o) \in \dep'\right\},
	\]
	i.e., $I = D \cup \heads{\dep'}$. It is easy to see that $I \in \Omega_{\Pi}^{\bckov}(D)$. Moreover, $I$ is the stable model of $\dep$ modulo active, and the claim follows.
	
	\OMIT{
	Let $J$ be a \bckov possible outcome of $\mathcal{G}$. Recall that $J$ is a minimal model, we will define set  $T_J$ as the ground $\mi{AtR}$ TGDs corresponding to the result atoms in $J$.  That is, let $T_J$ be the smallest set such that for every $\textrm{Result}_{|\bar q|}^{\delta}(\bar p,\bar q, y) \in J$, the rule
	\[
	\textrm{Active}_{|\bar q|}^{\delta}(\bar p,\bar q) \ra \textrm{Result}_{|\bar q|}^{\delta}(\bar p,\bar q, o)
	\]
	is in $T_J$. It is then straightforward to verify that $\dep = G(T_J)\cup T_J \in \Omega_{\Pi,\gsimppi}(D)$ and that the stable model of $\dep$ modulo \textrm{Active} is exactly $J$

	($\Rightarrow$)
	%Let $T \in \terminals(\gsimppi)$ be the ground $\mi{AtR}$ TGDs that generate $\Sigma$.
	let $J$ be the stable model of a possible outcome $\Sigma \in \Omega_{\Pi, \gsimppi}(D)$.
	For positive programs, the definition of $\mathit{GSimple}$ corresponds directly to the standard fixed-point definition of minimal models (cf.,~\cite{}) for some fixed instantiations of $\Delta$-terms. 
	Recall that the only difference in the  TGD$^\naf$ rewriting $\sigma_\Pi$ of our semantics to the TGDs $\widehat{\mathcal{G}}$ in \bckov semantics, are the intermediate $\mi{AtR}$ TGDs which are of no consequence w.r.t. the resulting model (except for the additional \textrm{Active} facts). 
	It is then straightforward to verify that $J$ modulo \textrm{Active} is a minimal model of $\widehat{\mathcal{G}}$.}

	$(2) \Rightarrow (1)$. We need to show that $\delta\params{\bar p}(o) > 0$ for every $\Delta$-atom $\textrm{Result}_{|\bar q|}^{\delta}(\bar p,\bar q,o) \in \heads{\dep}$.
	Observe that every such $\Delta$-atom occurs in the stable model of $\dep$ modulo active. Thus, by hypothesis, every $\Delta$-atom $\textrm{Result}_{|\bar q|}^{\delta}(\bar p,\bar q,o) \in \heads{\dep}$ belongs to a \bckov possible outcome. This allows us to conclude that $\delta\params{\bar p}(o) > 0$ for every $\Delta$-atom $\textrm{Result}_{|\bar q|}^{\delta}(\bar p,\bar q,o) \in \heads{\dep}$, and the claim follows.
\end{proof}

\begin{lemma}
	\label{lem:finiteiff}
	Consider a GDatalog$[\Delta]$ program $\Pi$, and a database $D$ of $\edb{\Pi}$. The following are equivalent:
	\begin{enumerate} 
	\item $\gsimppi$ is finitely grounding.
	\item Every instance $I \in \Omega_{\Pi}^{\bckov}(D)$ is finite.
	\end{enumerate}
\end{lemma}

\begin{proof}
	$(1) \Rightarrow (2)$. Fix an arbitrary \bckov possible outcome $I \in \Omega^{\bckov}_{\Pi}$. By Lemma~\ref{lem:PitoG}, there exists a possible outcome $\Sigma_I$ of $\Pi[D]$ relative to $\gsimppi$ such that the single stable model of $\Sigma_I$ modulo active is precisely the instance $I$. Since, by hypothesis, $\gsimppi$ is finitely grounding, we get that $\Sigma$ is finite. This in turn implies that $I$ is also finite, and the claim follows.
	
	$(2) \Rightarrow (1)$. Fix an arbitrary possible outcome $\dep \in \Omega_{\Pi,\gsimppi}(D)$. We need to show that $\dep$ is finite. By Lemma~\ref{lem:PitoG}, the stable model of $\dep$ modulo active is a \bckov possible outcome, and thus, by hypothesis, is finite. Since there exists at most one ``active'' atom that corresponds to a ``result'' atom in the stable model of $\dep$ modulo active, we can conclude that the stable model of $\dep$ is finite. 
	By definition (for positive programs), the simple grounder derives only ground rules whose head needs to be satisfied in the stable model of $\dep$. Thus, since the stable model of $\dep$ is finite, also $\dep$ itself is finite, and the claim follows.
\end{proof}

Having the above auxiliary lemmas in place, we can now complete the proof of Theorem~\ref{the:isomorphic-spaces} stated above.

\OMIT{
To conclude this section, we show that the probability space resulting from our semantics, and the probability space of \bckov semantics are isomorphic. Formally,
a \emph{isomorphism} between two probability spaces $(\Omega_1, \mathcal{F}_1, P_1)$, $(\Omega_2, \mathcal{F}_2, P_2)$ is a bijection $f \colon \Omega_1 \to \Omega_2$, such that
$F_1 \in \mathcal{F}_1 \iff f(F_1) \in \mathcal{F}_2$, and $P_1(F_1) = P_2(f(F_1))$ for all $F_1 \in \mathcal{F}_1$. % we can cite An Introduction to Probability Theory by K. Ito if we want to

\begin{theorem}
	Let $\Pi$ be a positive GDatalog$^\naf[\Delta]$ and $D$ be a database of $\edb{\Pi}$ such that $\gsimppi$ is finitely grounding. Let $\mathcal{G}$ be the GDatalog$[\Delta]$ program corresponding to $\Pi[D]$.
	The probability space $(\Omega_{\Pi,G}(D), \mathcal{F}_{\Pi,G}^D, P_{\Pi,G}^D)$ with $G= \gsimppi$ is isomorphic to the probability space $(\Omega^\bckov_{\mathcal{G}}, F^\bckov_{\mathcal{G}}, P^\bckov_{\mathcal{G}}$ from Theorem~\ref{thm:todsprob}.
\end{theorem}
}

\medskip

\textsc{Proof (of Theorem~\ref{the:isomorphic-spaces}).} For brevity, throughout the proof, we simply write $G$ instead of $\gsimppi$. We need to show that there exists an isomorphism between the probability spaces $\Pi_G(D)$ and $\Pi^{\bckov}(D)$.
First, observe that since $G$ is finitely grounding, $\Omega_{\Pi,G}^\infty(D) = \emptyset$, and thus, $\Omega_{\Pi,G}(D) = \Omega_{\Pi,G}^{\fin}(D)$.
We claim that the function $f : \Omega_{\Pi,G}^{\fin}(D) \ra \Omega_{\Pi}^{\bckov}(D)$ such that, for every $\dep \in \Omega^{\fin}_{\Pi,G}(D)$, $f(\dep)$ is the unique stable model of $\dep$ modulo active, is the required isomorphism. We need to show the following statements:
\begin{enumerate}
	\item $f$ is a bijection,
	%for every $\dep,\dep' \in \in \Omega^{\fin}_{\Pi,G}(D)$, $\dep \neq \dep'$ implies $f(\dep) \neq f(\dep')$,
	%\item for every $\dep \in \Omega_{\Pi}^{\bckov}(D)$, there exists $\dep' \in \Omega^{\fin}_{\Pi,G}(D)$ such that $\dep = f(\dep')$,
	\item for every $F \in 2^{\Omega_{\Pi,G}^{\fin}(D)}$, $F \in \mathcal{F}_{\Pi,G}^{D}$ iff $f(F) \in \mathcal{F}^\bckov_{\Pi,D}$, and
	\item for every $F \in \mathcal{F}_{\Pi,G}^{D}$, $P_{\Pi,G}^{D}(F) = P_{\Pi,D}^{\bckov}(f(F))$.
\end{enumerate}

\medskip

{\em Item (1).} For showing that $f$ is a bijection, we need to show that (i) for every $\dep,\dep' \in \Omega^{\fin}_{\Pi,G}(D)$, $\dep \neq \dep'$ implies $f(\dep) \neq f(\dep')$, and (ii) for every $\dep \in \Omega_{\Pi}^{\bckov}(D)$, there exists $\dep' \in \Omega^{\fin}_{\Pi,G}(D)$ such that $\dep = f(\dep')$. The former is an immediate consequence of Lemma~\ref{lem:positivemodels}, while the latter follows from Lemma~\ref{lem:PitoG}.

\medskip

{\em Item (2).} It suffices to show that $f(\mathcal{F}_{\Pi,G}^{D}) = \mathcal{F}_{\Pi,D}^{\bckov}$.
Recall that $\mathcal{F}^D_{\Pi,G}$ is the $\sigma$-algebra generated by $F^D_{\Pi,G}$ consisting of $\Omega^\infty_{\Pi,G}(D)$, and all the maximal sets $E \subseteq \Omega^{\fin}_{\Pi,G}(D)$ such that $\sms{\dep}=\sms{\dep'}$ for all $\dep,\dep'\in E$. Note that, by Lemma~\ref{lem:positivemodels},  $F^D_{\Pi,G}$ consists only of singleton sets.
We proceed to show that $f(\mathcal{F}_{\Pi,G}^{D}) = \mathcal{F}_{\Pi,D}^{\bckov}$.
$(\supseteq)$ By Lemma~\ref{lem:finiteiff}, every instance $I \in \Omega^\bckov_\Pi(D)$ is finite, and, by Lemma~\ref{lem:PitoG}, there exists $\Sigma \in \Omega^{\fin}_{\Pi,G}(D)$ such that $f(\Sigma)=I$.
Since all possible outcomes of $D$ w.r.t.~$\Pi$ are finite, and we assume countable sample spaces, there are only countably many possible outcomes of $D$ w.r.t.~$\Pi$. Therefore, all sets in $\Omega_{\Pi}^{\subseteq_\fin}$ are also countable. Hence, any set of $\mathcal{F}_{\Pi,D}^{\bckov}$ is a countable union of sets of $f(\mathcal{F}_{\Pi,G}^{D})$, and thus, $f(\mathcal{F}_{\Pi,G}^{D}) \supseteq \mathcal{F}_{\Pi,D}^{\bckov}$.
$(\subseteq)$ Similarly, since all $I \in \Omega^{\bckov}_{\Pi}(D)$ are finite, $\{I\} = \Omega_{\Pi}^{I\subseteq}$, and thus also $\{I\}\in \Omega_{\Pi}^{F\subseteq_\fin}$. Hence, $f(F_{\Pi,G}^D) \subseteq  \Omega_{\Pi}^{\subseteq_{\fin}}$, and then clearly also $f(\mathcal{F}_{\Pi,G}^{D}) \subseteq \mathcal{F}_{\Pi,D}^{\bckov}$.

\medskip

{\em Item (3).} Since the sample spaces are countable, and every sample exists as a singleton in the respective $\sigma$-algebra, it suffices to prove the claim only for the singleton events. The probability of all other events is determined by them via the axioms of probability measures. For every $\dep \in \Omega^{\fin}_{\Pi,G}(D)$, let $I_\dep$ be the minimal model of $\dep$. Again, we have that $I_\dep$ is finite; hence, $I_\dep$ is also a finite derivation set and $\{I_\dep\} = \Omega^{I_\dep\subseteq}_\Pi$. Then, $P^\bckov_\Pi(\{I_\dep \})$ is precisely the product of probabilities of the ``result'' atoms occurring in $I_\dep$. By the definition of $\gsimppi$, and grounders in general, all ``result'' atoms in heads of $\Pi$ are also in the respective minimal model. Hence, $P_{\Pi,G}^D(F)=P^\bckov_{\Pi}(f(F))$, for every singleton $F \in \mathcal{F}_{\Pi,G}^{D}$, and the claim follows. \qed

%%%%%%%%%%%%%%%%%%%%%%%%%%%%%

\OMIT{
%\begin{proposition}
%  \label{prop:finite-program-0infty}
%  Let $\Pi$ be a GDatalog$^\naf[\Delta]$ program, let $D$ a database of $\edb{\Pi}$, and let $G$ be a finitely grounding grounder for $\Pi[D]$. Then $P^D_{\Pi,G}(\Omega_{\Pi,G}^\infty)=0$.
%\end{proposition}

In the following we compare the interpretation of  positive GDatalog$^ \naf[\Delta]$ under our proposed semantics to their interpretation as GDatalog$[\Delta]$ programs under \bckov semantics. For clarity, we will use two different names for the same program in the following, one name where we consider the program as a  GDatalog$^ \naf[\Delta]$ program and another when considered as the \emph{corresponding} GDatalog$[\Delta]$ program. Despite the distinct names used below, both objects represent the syntactically same program. We will always consider our semantics for positive GDatalog$^\naf[\Delta]$ programs, in comparision to \bckov semantics for the corresponding GDatalog$[\Delta]$ program.

Our rewriting of GDatalog$^\naf[\Delta]$ programs to TGDs is the same as the rewriting used in \bckov semantics, except for the addition of intermediate $\textrm{Active}$ predicates and the respective TGDs that link \textrm{Active} to \textrm{Result}. It will therefore be useful to talk about stable models of a GDatalog$^\naf[\Delta]$ program $\Pi$ \emph{modulo \textrm{Active}}, which we define as a stable model of some possible outcome $\Sigma \in \Omega_{\Pi,G}(D)$ with all atoms of the form $\textrm{Active}_{|\bar q|}^{\delta}(\bar p,\bar q)$ removed.

\begin{lemma}
  \label{lem:positivemodels}
  Let $\Pi$ be a positive GDatalog$^\naf[\Delta]$ program and let $D$ be a database of $\edb{\Pi}$. For any possible outcome $\Sigma \in \Omega_{\Pi,\gsimppi}(D)$, $\Sigma$ has exactly one stable model. Moreover, for every $\Sigma' \in \Omega_{\Pi,\gsimppi}(D)$ with $\Sigma \neq \Sigma'$, the stable model of $\Sigma$ modulo \textrm{Active} is distinct from the stable model of $\Sigma'$ modulo \textrm{Active} of every other possible outcome .% with the same stable, and for any $G'\in \mathit{ground}(\simpgrounder, \Pi)$ with $G\neq G'$, it holds that $\sms{G} \neq \sms{G'}$.
\end{lemma}
\begin{proof}
  Since $\Pi$ is positive, so is $\sigma_\Pi$ and therefore also every possible outcome $\Sigma$ is positive and no existential quantification. It is well known that such programs have exactly one stable model, i.e., $\sms{\Sigma}$ is a singleton.

  Let $\Sigma, \Sigma'$ be two distinct possible outcomes. That is $\Sigma$ and $\Sigma'$ are in $\gsimppi$-$\mi{ground}(\Pi[D])$, and therefore $\Sigma$ and $\Sigma'$ must be generated by distinct terminals $T$, $T'$, i.e., different subsets of $[2^{\ground{\dep_{\Pi[D]}^{\exists}}}]_{=}$ that are incomparable by $\subseteq$.
  Since $T$ and $T'$ are minimal w.r.t. $\terminals{\gsimppi}$, all body atoms in $T$ and $T'$ occur in the heads of $\gsimppi(T)$ and $\gsimppi(T')$, respectively. Hence, by definition of $\gsimppi$, the single stable model of $\Sigma$ contains exactly $\heads{T}$  as the only atoms of the form $\textrm{Result}^\delta_n$ (the \emph{result} atoms), while the stable model of $\Sigma'$ contains exactly $\heads{T'}$ as its only result atoms. Since $T$ and $T'$ are incomparable by $\subseteq$, so are $\heads{T}$ and $\heads{T'}$. Hence, the stable models modulo \textrm{Active} of any distinct $\Sigma$ and $\Sigma'$ are also distinct.
\end{proof}

\begin{lemma}
  \label{lem:PitoG}
  Let $\Pi$ be a positive GDatalog$^\naf[\Delta]$ program, and let $D$ be a database of $\edb{\Pi}$.
  Then $\dep \in \Omega_{\Pi,\gsimppi}(D)$ if and only if the stable model of $\dep$ modulo \textrm{Active} is a \bckov possible outcome of the GDatalog$[\Delta]$ program $\mathcal{G}$ corresponding to $\Pi[D]$.
\end{lemma}
\begin{proof}
    ($\Leftarrow$) Let $J$ be a \bckov possible outcome of $\mathcal{G}$. Recall that $J$ is a minimal model, we will define set  $T_J$ as the ground $\mi{AtR}$ TGDs corresponding to the result atoms in $J$.  That is, let $T_J$ be the smallest set such that for every $\textrm{Result}_{|\bar q|}^{\delta}(\bar p,\bar q, y) \in J$, the rule
    \[
      \textrm{Active}_{|\bar q|}^{\delta}(\bar p,\bar q) \ra \textrm{Result}_{|\bar q|}^{\delta}(\bar p,\bar q, o)
    \]
    is in $T_J$. It is then straightforward to verify that $\dep = G(T_J)\cup T_J \in \Omega_{\Pi,\gsimppi}(D)$ and that the stable model of $\dep$ modulo \textrm{Active} is exactly $J$.

    ($\Rightarrow$)
    %Let $T \in \terminals(\gsimppi)$ be the ground $\mi{AtR}$ TGDs that generate $\Sigma$.
    let $J$ be the stable model of a possible outcome $\Sigma \in \Omega_{\Pi, \gsimppi}(D)$.
    For positive programs, the definition of $\mathit{GSimple}$ corresponds directly to the standard fixed-point definition of minimal models (cf.,~\cite{}) for some fixed instantiations of $\Delta$-terms. 
    Recall that the only difference in the  TGD$^\naf$ rewriting $\sigma_\Pi$ of our semantics to the TGDs $\widehat{\mathcal{G}}$ in \bckov semantics, are the intermediate $\mi{AtR}$ TGDs which are of no consequence w.r.t. the resulting model (except for the additional \textrm{Active} facts). 
    It is then straightforward to verify that $J$ modulo \textrm{Active} is a minimal model of $\widehat{\mathcal{G}}$.
\end{proof}

\begin{lemma}
  \label{lem:finiteiff}
  Let $\Pi$ be a positive GDatalog$^\naf[\Delta]$ program, and let $D$ be a database of $\edb{\Pi}$. Let $\mathcal{G}$ be the GDatalog$[\Delta]$ program corresponding to $\Pi[D]$. 
  Then $\gsimppi$ is finitely grounding for $\Pi[D]$ if and only if all possible outcomes of the corresponding GDatalog$[\Delta]$ program $\mathcal{G}$ are finite.
\end{lemma}
\begin{proof}
  For the implication from left to right, observe that for every \bckov possible outcome $J \in \Omega^{\bckov}_\mathcal{G}$ there is a possible outcome $\Sigma$ of $\Pi[D]$ relative to $\gsimppi$, such that the single stable model of $\Sigma$ modulo \textrm{Active} is $J$ (Lemma~\ref{lem:PitoG}). Since $\gsimppi$ is finitely grounding, $\Sigma$ is finite and thus so is $J$. Hence, all possible outcomes of $\mathcal{G}$ are finite.
  For the other direction we have that for every possible outcome $\Sigma$ of $\Pi[D]$ relative to $\gsimppi$, the stable model modulo \textrm{Active} $J$ of $\Sigma$ is a \bckov possible outcome. By assumption, every \bckov possible outcome is finite. There is at most one satisfied \textrm{Active} fact per satisfied \textrm{Result} fact in the model, hence the stable model of $\Sigma$ is finite. By definition (for positive programs), $\mi{GSimple}$ derives only ground rules whose head needs to be satisfied in the model of $\Sigma$, thus since the model is finite and $D$ is finite, also $\Sigma$ will be finite.
\end{proof}

To conclude this section, we show that the probability space resulting from our semantics, and the probability space of \bckov semantics are isomorphic. Formally,
a \emph{isomorphism} between two probability spaces $(\Omega_1, \mathcal{F}_1, P_1)$, $(\Omega_2, \mathcal{F}_2, P_2)$ is a bijection $f \colon \Omega_1 \to \Omega_2$, such that
$F_1 \in \mathcal{F}_1 \iff f(F_1) \in \mathcal{F}_2$, and $P_1(F_1) = P_2(f(F_1))$ for all $F_1 \in \mathcal{F}_1$. % we can cite An Introduction to Probability Theory by K. Ito if we want to

\begin{theorem}
  Let $\Pi$ be a positive GDatalog$^\naf[\Delta]$ and $D$ be a database of $\edb{\Pi}$ such that $\gsimppi$ is finitely grounding. Let $\mathcal{G}$ be the GDatalog$[\Delta]$ program corresponding to $\Pi[D]$.
  The probability space $(\Omega_{\Pi,G}(D), \mathcal{F}_{\Pi,G}^D, P_{\Pi,G}^D)$ with $G= \gsimppi$ is isomorphic to the probability space $(\Omega^\bckov_{\mathcal{G}}, F^\bckov_{\mathcal{G}}, P^\bckov_{\mathcal{G}}$ from Theorem~\ref{thm:todsprob}.
\end{theorem}
\begin{proof}
  Throughout the proof we will use $G$ for $\gsimppi$.
  First, observe that if $\gsimppi$ is finitely grounding, then $\Omega_{\Pi,G}^\infty(D) = \emptyset$.
  We argue that the function $f$, that takes every $\dep \in \Omega^{fin}_{\Pi,G}(D)$ to the single stable model of $\dep$ modulo \textrm{Active}, is the required isomorphism.  Note that $\Omega^{fin}_{\Pi,G}(D) = \Omega_{\Pi,G}(D)$ since \gsimppi is finitely grounding. By  Lemma~\ref{lem:positivemodels}, every $\dep \in \Omega^{fin}_{\Pi,G}(D)$ has a single stable model. By Lemma~\ref{lem:PitoG}, the function $f$ is a one-to-one correspondence between ground programs in $\Omega^{fin}_{\Pi,G}(D)$ and models in $\Omega^\bckov_{\mathcal{G}}$.

  We move on to verifying that $f$ is indeed measure preserving, that is, for every $X \in \mathcal{F}^D_{\Pi,G}$ we have $f(X) \in \mathcal{F}^\bckov_{\mathcal{G}}$ and $P^D_{\Pi,G}(X)=P^\bckov_{\mathcal{G}}(f(X))$.
   Let us write $\sigma(F)$ for the closure of set $F$ under countable union and complement. 
 Recall that $\mathcal{F}^D_{\Pi,G} =  \sigma(F^D_{\Pi,G})$, where $F^D_{\Pi,G}$ consists of $\Omega^\infty_{\Pi,G}(D)$ and all maximal $E \subseteq \Omega^{fin}_{\Pi,G}(D)$ such that $\sms{\dep}=\sms{\dep'}$ for all $\dep,\dep'\in E$. By Lemma~\ref{lem:positivemodels},  $F^D_{\Pi,G}$ consists only of singletons. We argue that $f(\sigma(F_\Pi)) = \sigma(\Omega_{\mathcal{G}}^{\mathit{fin}\subseteq})$.
    % By Lemma~\ref{lem:positivepi} we have that $F'_\Pi$ consists only of singletons, each containing exactly a possible outcome in $\Omega_{\mathcal{G}}$. All sets in $\sigma(\Omega_{\mathcal{G}}^{\mathit{fin}})$ are countable and thus, by Lemma~\ref{lem:positivepi}, expressible as a countable unions of images under $f$ of elements of $F'_\Pi$.
 For any $J \in \Omega^\bckov_\mathcal{G}$, we have by Lemma~\ref{lem:finiteiff} that $J$ is finite and by Lemma~\ref{lem:PitoG} there is a $\Sigma \in \Omega^{fin}_{\Pi,G}(D)$ such that $f(\Sigma)=J$.
 Since all possible outcomes of $\mathcal{G}$ are finite and we assume countable sample spaces, there are only countably many possible outcomes of $\mathcal{G}$. Therefore, all
 sets in $\Omega_{\mathcal{G}}^{\mathit{fin}\subseteq}$ are also countable. Hence, any set in $\sigma(\Omega_\mathcal{G}^{\mathit{fin}\subseteq})$ is a countable union of sets in $f(F^D_{\Pi,G})$ and we have $f(\sigma(F_{\Pi,G}^D)) \supseteq \sigma(\Omega_{\mathcal{G}}^{\mathit{fin}\subseteq})$. Similarly, since all $J \in \Omega^\bckov_{\mathcal{G}}$ are finite, $\{J\} = \Omega_{\mathcal{G}}^{J\subseteq}$ and thus also $\{J\}\in \Omega_{\mathcal{G}}^{fin\subseteq}$. Thus $f(F_{\Pi,G}^D) \subseteq  \Omega_{\mathcal{G}}^{\mathit{fin}\subseteq}$ and then clearly also $f(\sigma(F_{\Pi,G}^D)) \subseteq \sigma(\Omega_{\mathcal{G}}^{\mathit{fin}\subseteq})$.

 Finally, we observe that for every $X \in \sigma(F^D_{\Pi,G})$ we have $P_{\Pi,G}^D(X) = P^\bckov_{\mathcal{G}}(f(X))$. Since the sample spaces are countable, and every sample exists as a singleton in the respective $\sigma$ algebra, it suffices to observe the equivalence only for the singleton events. The probability of all other events is determined by them via the axioms of probability measures.
 For every $\dep \in \Omega^{fin}_{\Pi,G}(D)$, let $J$ be the minimal model of $\dep$. Again we have that $J$ is finite, hence $J$ is also a finite derivation set and $\{J\} = \Omega^{J\subseteq}_\mathcal{G}$. Then $P^\bckov_\mathcal{G}(\{J \})$ is precisely the product of probabilities of \textrm{Result} facts in $J$. By the definition of $\gsimppi$ and grounders in general, all result atoms in heads of $G$ are also in the respective minimal model. Hence, $P_{\Pi,G}^D(X)=P^\bckov_{\mathcal{G}}(f(X))$.
\end{proof}
}

%%%%%%%%%%%%%%%%%%%%%%%%

\subsection*{Proof of Theorem~\ref{the:gsimple-best}}
A GDatalog$[\Delta]$ program $\Pi$ is a special case of a GDatalog$^{\naf s}[\Delta]$ program that does not use negation.
Furthermore, the simple grounder and the perfect grounder coincide for positive programs since, for every positive (infinitary) existential-free TGD program $\dep$, $\mathit{Simple}_\dep$ and $\mathit{Perfect}_\dep$ is the same function from $2^{\ground{\dep}}$ to $2^{\ground{\dep}}$.
%$\mathit{Simple}_\dep(\dep') = \mathit{Perfect}_\dep(\dep')$ for any $\dep \in 2^{\ground{\dep}}$ since $B^-(\sigma) = \emptyset$ for every positive rule $\sigma$. 
Therefore, Theorem~\ref{the:gsimple-best} is a special case of Theorem~\ref{the:gperfect-best}, which is shown below.

\OMIT{
\subsection*{Proof of Proposition~\ref{pro:compare-semantics}}
We first observe that for every $\dep \in [2^{\ground{\dep_{\Pi}^{\exists}}}]_{=}$ and $\dep \cup G'(\dep) \in  \Omega^{fin}_{\Pi,G'}(D)$, there is a $\dep' \subseteq \dep$, such that $\dep' \cup G(\dep') \in \Omega^{fin}_{\Pi,G}(D)$ and $\sms{\dep \cup G'(\dep)} = \sms{\dep' \cup G(\dep')}$.

By assumption, if $G'(\dep)$ is finite, then so is $G(\dep)$, and $\mi{AtR}_\dep \comp G(\dep)$. Thus, either $\dep$  is in $\terminals{G}$ or there is a $\dep'\subseteq \dep$ with $\mi{AtR}_{\dep'} \comp G(\dep')$. By the definition of a grounder we then have
\[
  \sms{\dep' \cup G(\dep')} = \sms{\dep \cup G(\dep)} = \sms{\Sigma_\Pi^{\not\exists} \cup \dep^*} = \sms{\dep \cup G'(\dep)}
\]
where $\dep^*$ is  a totalizer of $\mi{Atr}_\dep$. Clearly, $\dep' \cup G(\dep') \subseteq \dep \cup G(\dep)$ is finite.

Now, for any $\dep' \in [2^{\ground{\dep_{\Pi}^{\exists}}}]_{=}$ such that $\dep' \cup G(\dep') \in  \Omega^{fin}_{\Pi,G}(D)$, let $\mathfrak{E}(\dep') \subseteq \Omega^{fin}_{\Pi,G'}(D)$ such that every program in $\mathfrak{E}$ is of the form $\dep \cup G'(\dep)$, with $\dep \supseteq \dep'$, and has the same stable models as $\dep' \cup G(\dep')$.
By definition of $P_{\Pi,G}^D$ for sets of finite possible outcomes it follows that
\[
  P_{\Pi,G}^D(\{\dep \cup G(\dep)\}) \geq P_{\Pi,G'}^D(\mathfrak{E}(\dep'))
\]
By the above observation, every set of possible outcomes in $\Omega^{fin}_{\Pi,G'}(D)$ is the subset of some union $\bigcup_{\dep' \in \mathfrak{S}}\mathfrak{E}(\dep')$, where $\mathfrak{S} \subseteq \Omega^{fin}_{\Pi,G}(D)$. Hence the proposition holds.
}

%%% Local Variables:
%%% mode: latex
%%% TeX-master: "main-generative-datalog"
%%% End:

\section{Section~\ref{sec:chase}: Fixpoint Probabilistic Semantics}\label{sec:apx-chase}

\subsection*{Proof of Lemma~\ref{lem:chase-tree-properties}}
{\em Item (1).} Intuitively, this follows inductively from Definition~\ref{def:chase-step} and
Definition~\ref{def:chase-tree}, since the definition requires that the label of each child of a chase tree node is the result of a trigger application.
We proceed to prove item (1) by induction on the height of the tree. 

\textbf{Base Case.} Only the root node $v_\mathit{root}$ of $T$ resides at depth $0$, with label
$\lambda(v_\mathit{root}) = \emptyset$, and it trivially holds that $\emptyset \in [2^{\ground{\dep_{\Pi[D]}^{\exists}}}]_{=}$.

\textbf{Inductive Step.} Assume now that the property
holds for all nodes up to depth $n$. Pick any such node $v$ at depth $n$ and let
$v_1, v_2, \ldots$ be its children. By Definition~\ref{def:chase-tree}, there
exists a trigger $\alpha$ for $G(\lambda(v))$ on $\lambda(v)$ such that
$\lambda(v)\params{\alpha}\{\lambda(v_1),\lambda(v_2),\ldots\}$. By
Definition~\ref{def:chase-step}, for each $i > 0$ we have that $\lambda(v_i)
\setminus \lambda(v) = \left\{\alpha \ra \text{\rm Return}^{\delta}_{|\bar
q|}(\bar p,\bar q,o)\right\}$, for some $o \in \ins{C}$, but no other AtR TGD
$\alpha \ra \text{\rm Return}^{\delta}_{|\bar q|}(\bar p,\bar q,o')$ can exist
in $\lambda(v)$. Hence, $\lambda(v_i)$ is functionally consistent, which
establishes the desired property for all nodes at depth $n+1$, proving the
claim.

\smallskip
{\em Item (2).} First, note that, by Definition~\ref{def:chase-tree}
and~\ref{def:chase-step}, for any parent node $v_p \in N$ with child node $v_c$,
it must hold that $\lambda(v_p) \subset \lambda(v_c)$, i.e., the labels of the
nodes increase monotonically along the parent-child axis. We proceed by considering the following two
cases:

\textbf{Case 1.} First, assume that neither $u$ nor $v$ are ancestors of each other.  Let $w$ be
the least common ancestor node of both $u$ and $v$, with $w_u$ and $w_v$ being
two different child nodes of $w$ and ancestor nodes of $u$ and $v$, respectively
(by assumption, these two nodes must clearly exist).  By Definition~\ref{def:chase-tree} and~\ref{def:chase-step}, $\lambda(w_u) =
\lambda(w) \cup \left\{\alpha \ra \text{\rm Return}^{\delta}_{|\bar q|}(\bar
p,\bar q,o)\right\}$, and $\lambda(w_v) = \left\{\alpha \ra \text{\rm
Return}^{\delta}_{|\bar q|}(\bar p,\bar q,o')\right\}$, where $o \neq o'$, for
some $o, o' \in \ins{C}$. But then, by monotonicity, $\lambda(w_u) \subseteq
\lambda(u)$, and $\lambda(w_v) \subseteq \lambda(v)$, proving the claim that
$\lambda(u) \neq \lambda(v)$.

\textbf{Case 2.} For the other case, assume, w.l.o.g., that $u$ is an ancestor of $v$. But then,
by monotonicity of the node labels, we have that $\lambda(u) \subset \lambda(v)$
and thus, clearly, $\lambda(u) \neq \lambda(v)$, as claimed.

\subsection*{Proof of Lemma~\ref{lem:chase-tree-uniqueness}}

Due to item (2) of Lemma~\ref{lem:chase-tree-properties}, it suffices to show the following: given a path $\pi \in \mathsf{paths}^{\fin}(T)$ of length $k$, there
exists a path $\pi' \in \mathsf{paths}^{\fin}(T')$, also of length $k$, with $\sem{\pi} = \sem{\pi'}$.
To prove this, we show inductively that $\pi'$ exists, and for all nodes $v$ in path $\pi'$ it holds that $\lambda(v) \subseteq \sem{\pi}$. Let $v_0$ be the
first node of $\pi'$, that is, the root node of tree $T'$.

\textbf{Base Case.} It clearly holds that the path of length $0$ consisting only $v_0$ exists in $T'$,
and that $\lambda(v_0) \subseteq \sem{\pi}$, since $\lambda(v_0) = \emptyset$.

\textbf{Inductive Step.} Assume that a path $v_0v_1\ldots v_n$, with $n < k$,
exists in $T'$ such that $\lambda(v_i) \subseteq \sem{\pi}$ holds for $0
\leqslant i \leqslant n$. Note that, by the monotonicity of grounder $G$, we
have that $G(\lambda(v_n)) \subseteq \sem{\pi}_G$. Let $\alpha = \text{\rm
Active}_{|\bar q|}^{\delta}(\bar p,\bar q)$ be the trigger that was applied to
$\lambda(v_n)$ to obtain $T'$. We claim that such a trigger exists, and will
proof this claim, denoted Claim (*), at the end of the proof. It must thus hold
that $\alpha \in \heads{\sem{\pi}}$. But since $\pi$ is a finite path, its leaf
node, by Definition~\ref{def:chase-tree}, contains no more triggers for
$G(\sem{\pi})$ on $\sem{\pi}$, and hence $\sem{\pi}$ must contain an AtR TGD
$\rho = \alpha \ra \text{\rm Return}^{\delta}_{|\bar q|}(\bar p,\bar q,o)$, for
some $o \in \ins{C}$. But since $\alpha$ is the trigger applied to $v_n$ in
$T'$, there must, by definition, exist a child node $v_{n+1}$ of $v_n$ with
$\lambda(v_{n+1}) = \lambda(v_n) \cup \{ \rho \}$. Clearly, $\lambda(v_{n+1})
\subseteq \sem{\pi}$, which establishes our desired property for the path $\pi'
= v_0v_1\ldots v_k$ in $T'$.

Finally, we need to show that for leaf node $v_k$ we have $\lambda(v_k) =
\sem{\pi}$. This is, however, straightforward: the above inductive proof shows
that overall, the exact same triggers were applied in $\pi'$ and in $\pi$, with
the same AtR TGDs added in each, resulting in the same set.

It remains to show Claim (*). Assume to the contrary that such a trigger does not exist for $v_n$. Then, $v_n$, for $n < k$, is the leaf node of path $\pi'$,
and hence $\sem{\pi'} \subset \sem{\pi}$. But then there must be an AtR TGD
$\rho_j = \text{\rm Active}_{|\bar q_j|}^{\delta_j}(\bar p_j,\bar q_j) \ra
\text{\rm Return}^{\delta_j}_{|\bar q_j|}(\bar p_j,\bar q_j,o_j) \in \sem{\pi}$,
with $\rho_j \not\in \sem{\pi'}$. Let $v_j$, for $0 \leqslant j < k$ be the node in
$\pi$ where trigger $\alpha_j = \text{\rm Active}_{|\bar q_j|}^{\delta_j}(\bar
p_j,\bar q_j)$ first appeared. Either $\lambda(v_j) \subseteq \lambda(v_n)$, a
contradiction, or in $v_{j-1}$ an AtR TGD $\rho_{j-1} = \text{\rm Active}_{|\bar
q_{j-1}|}^{\delta_{j-1}}(\bar p_{j-1},\bar q_{j-1}) \ra \text{\rm
Return}^{\delta_{j-1}}_{|\bar q_{j-1}|}(\bar p_{j-1},\bar q_{j-1},o_{j-1}) \in
\sem{\pi}$ was added such that trigger $\alpha_j$ arose. But then, let $v_{j-1}$
be the node in $\pi$ where trigger $\alpha_{j-1} = \text{\rm Active}_{|\bar
q_{j-1}|}^{\delta_{j-1}}(\bar p_{j-1},\bar q_{j-1})$ first appeared. The same
argument applies to $v_{j-1}$, and then further, inductively, until a
contradiction is reached. Such a contradiction will be reached at the root node
at the latest, since the label of the root node is the empty set. This completes the proof of Claim (*), and this, of Lemma~\ref{lem:chase-tree-uniqueness}.

\subsection*{Proof of Lemma~\ref{lem:bijection}}

First, observe that $\sem{\cdot}_G$ is a bijective function. Hence, we only need
to show that this function maps $\mathsf{paths}^{\fin}(T)$ precisely to the set
$\Omega_{\Pi,G}^{\fin}(D)$.

($\Rightarrow$) By Definition~\ref{def:chase-tree}
and~\ref{def:chase-step}, each path $\pi \in \mathsf{paths}^{\fin}(T)$ is
associated with the set $\sem{\pi}$, such that $\sem{\pi} \comp G(\sem{\pi})$.
Hence, $\sem{\pi}_G \in G$-$\ground{\Pi}$, and, since, by
Definition~\ref{def:chase-step}, it holds that $\delta\params{\bar p}(o) > 0$
for a trigger application to take place, we also have $\sem{\pi}_G \in
\Omega_{\Pi,G}^{\fin}(D)$.

($\Leftarrow$) From $\Omega_{\Pi,G}^{\fin}(D)$, we define a partial
chase tree $T'$ recursively as follows: Add a root node $v_\mathit{root}$ with
$\lambda(v_\mathit{root}) = \emptyset$. For any node $v$ in $T'$, for which a
trigger $\alpha$ exists for $G(\lambda(v))$ on $\lambda(v)$, pick any such
$\alpha$ and let 
\[
\{ \dep_1, \dep_2, \ldots \}\ =\ \left\{ \lambda(v) \cup \{ \rho = \alpha \ra \text{\rm Return}^{\delta}_{|\bar q|}(\bar p,\bar q,o) \} \mid \dep
\in \Omega_{\Pi,G}^{\fin}(D),\ \lambda(v) \subseteq \dep,\ \rho \in
\dep\ \text{and}\ \delta\params{\bar p}(o) > 0 \right\}.  
\]
For each such set of AtR
TGDs $\dep_i$ with $i > 0$, create a child node $v_i$ of $v$ in $T'$, with
$\lambda(v_i) =\dep_i$. Note that trigger $\alpha$ neccesarily exists for any
node $v$ in $T'$ where $\dep = \lambda(v) \cup G(\lambda(v)) \not\in
\Omega_{\Pi,G}^{\fin}(D)$; otherwise, $\dep$ would not be subset minimal, and thus not in $G$-$\ground{\Pi}$ or in $\Omega_{\Pi,G}^{\fin}(D)$.

Note that, by definition of $T'$, there is a path $\pi$ in $T'$ for each set
$\dep \in \Omega_{\Pi,G}^{\fin}(D)$, such that $\sem{\pi} = \dep = \lambda(v)$,
where $v$ is the leaf node of $\pi$, and $\dep$ is constructed step-by-step
along path $\pi$ in $T'$. The resulting tree $T'$ is not yet a complete chase
tree, since it only contains finite paths (i.e., the ones reflected in
$\Omega_{\Pi,G}^{\fin}(D)$). Let $T''$ be the chase tree obtained from $T'$,
where, whenever for a node $v$ in $T'$ there exists a constant $o \in \ins{C}$
with $\delta\params{\bar p}(o) > 0$, but there is no child $v'$ of $v$, with
$\lambda(v') = \lambda(v) \cup \left\{\alpha \ra \text{\rm
Return}^{\delta}_{|\bar q|}(\bar p,\bar q,o)\right\}$, such a child $v'$ is
present in $T''$. For each such node $v$ in $T''$, but not in $T'$, let the
chase subtree rooted at $v$ be defined as in Definition~\ref{def:chase-tree}.

With this construction in place, we need show that the constructed labelled tree
$T''$ is indeed a chase tree, as per Definition~\ref{def:chase-tree}. By
construction, the root node of $T''$ is labelled with $\emptyset$. For any
non-leaf node $v$ in $T''$, we need to show that $\lambda(v) \params{\alpha}
\{\lambda(v_1),\lambda(v_2),\ldots\}$ is indeed a valid trigger application,
where $v_1, v_2, \ldots$ are the children of $v$. Take any set
$\dep' \in \Omega_{\Pi,G}^{\fin}(D)$. Recall that, by definition of
$G$-$\ground{\Pi}$, $\dep = \dep' \cap \ground{\dep_{\Pi}^{\exists}}$ is a set
of AtR TGDs that is subset-minimal w.r.t.\ the compatibility relation
$\mi{AtR}_\dep \comp G(\dep)$ (i.e., $\dep$ is a subset-minimal set within
$\terminals{G}$). In particular, $\dep$ is functionally consistent. We need to
show, as per Definition~\ref{def:chase-step}, that for every node $v$ in $T''$,
with children $v_1, v_2, \ldots$, it holds that
\begin{enumerate}
  \item\label{point:a} for each integer $i >0$, $\lambda(v_i) = \lambda(v) \cup
    \left\{\alpha \ra \text{\rm Return}^{\delta}_{|\bar q|}(\bar p,\bar
    q,o)\right\}$, where $o \in \ins{C}$ and $\delta\params{\bar p}(o) > 0$, and
  \item\label{point:b} for each constant $o \in \ins{C}$ with
    $\delta\params{\bar p}(o) > 0$, there exists $i > 0$ such that $\lambda(v_i) =
    \lambda(v) \cup \left\{\alpha \ra \text{\rm Return}^{\delta}_{|\bar q|}(\bar
    p,\bar q,o)\right\}$.
\end{enumerate}

For any node in $T''$ but not in $T'$, these conditions are trivially satisfied.
Hence we only need to argue for nodes in $T'$.  For Point~\ref{point:a}, this
follows from the definition of $\Omega_{\Pi,G}^{\fin}(D)$. Note that for any
node $v$ in $T'$, $\lambda(v_i)$ is functionally consistent and compatible with
$G(\lambda(v_i))$. We have $\mi{AtR}_{\lambda(v_i)} \comp G(\lambda(v_i))$,
since $\lambda(v)$ contains trigger $\alpha$ that was picked for $v$ in the
construction of $T'$, and $\mi{AtR}_{\lambda(v_i)}$ hence is defined for
$\alpha$. Point~\ref{point:b} follows directly by construction of $T''$ from
$T'$.

We have shown that $T''$ is a chase tree, with $T'$ a subtree of finite depth
within $T''$. By Lemma~\ref{lem:chase-tree-uniqueness}, we know that any chase,
in particular the tree $T$ chosen in the current lemma, must agree with $T''$ on
the finite paths. Thus, we have that $\sem{\mathsf{paths}^{\fin}(T)}\ =\
\sem{\mathsf{paths}^{\fin}(T'')}$.

To complete the proof, we thus need to show that all finite paths in $T''$ are
already present in $T'$. Towards a contradiction, assume that there is a finite
path $\pi$ in $T''$ whose leaf node is not in $T'$. But then $\sem{\pi}$ is
clearly a finite set of AtR TGDs, and $\sem{\pi} \comp \sem{\pi}_G$. But then,
$\sem{\pi}_G \in \Omega_{\Pi,G}^{\fin}(D)$, and, by construction of $T'$, there
is a leaf node labelled $\sem{\pi}$ in $T'$.  Hence, since the labels of tree
nodes along a path are monotonically increasing, path $\pi$ must have its leaf
node in $T'$, a contradiction.

We thus have, as required, that, indeed, the function $\sem{\cdot}_G$ is a
bijection from $\mathsf{paths}^{\fin}(T'')$ to
$\Omega_{\Pi,G}^{\fin}(D)$. But, as argued above, by
Lemma~\ref{lem:chase-tree-uniqueness}, this also holds for $T$, completing the
proof.

\subsection*{Proof of Theorem~\ref{the:chase-semantics}}

{\em Item (1).} This directly follows from Lemma~\ref{lem:bijection}.

\smallskip

\noindent {\em Item (2).} For a chase tree $T$, let $f: \mathsf{paths}^{\fin}(T) \to \Omega_{\Pi,G}^{\fin}(D)$ be the bijective function provided by Lemma~\ref{lem:bijection}. Let $g$ be the extension of $f$ to sets, such that for any set $E \subseteq \sem{\mathsf{paths}^{\fin}(T)}$ we have
$g(E) = \{ f(\pi) \mid E = \sem{\pi_1, \pi_2, \ldots}, \pi = \pi_i, i > 0 \}$. Clearly, $g$ is a bijection between the $\sigma$-algebras $\sem{\mathcal{F}_T}_G$ and $\mathcal{F}_{\Pi,G}^{D}$, since $f$ is a bijection
between the base sets of these two $\sigma$-algebras, establishing item (\ref{the:chasesemantics:2}), if, additionally, we set $g(\mathsf{paths}^{\inf}(T)) = \Omega_{\Pi,G}^{\inf}(D)$.

\smallskip

\noindent {\em Item (3).} Finally, item (\ref{the:chasesemantics:3}) follows directly from the definitions of the probability measures $\sem{P_T}_G$ and $P_{\Pi,G}^{D}$, since the
probabilities for the base sets $\mathsf{paths}^{\fin}(T)$ and
$\Omega_{\Pi,G}^{\fin}(D)$ are defined in exactly the same way: as the product
of all $\delta\params{\bar p}(o)$ for each $\textrm{Result}_{|\bar
q|}^{\delta}(\bar p, \bar q, o)$ atom occurring in an AtR TGD in $\pi \in
\mathsf{paths}^{\fin}(T)$ and in $\dep \in \Omega_{\Pi,G}^{\fin}(D)$,
respectively. Since $f(\pi) = \dep$ and $f^{-1}(\dep) = \pi$, clearly
$\sem{P_T}_G(\pi) = P_{\Pi,G}^{D}(g(\pi))$. The equivalence of the probabilities
for the infinity and composite events then follows trivially from the
definitions.

\section{Section~\ref{sec:stratified-negation}: Stratified Negation}\label{sec:apx-stratified-negation}

\subsection*{Example of Perfect Grounding}

Let us consider a simple scenario in which a given set of dimes are tossed, and if none of them show tail then also a given set of quarters are tossed.
Let $\Delta = \{\flip\}$, and assume we have the following schema: 
\begin{equation*}
	\begin{array}{lllll}
		\text{\rm Dime}(\mi{dime\_id}) & 
		\text{\rm Quarter}(\mi{quarter\_id}) &
		\text{\rm DimeTail}(\mi{dime\_id,\mi{boolean}}) &
		\text{\rm QuarterTail}(\mi{quarter\_id},\mi{boolean}) &
		\text{\rm SomeDimeTail}.
	\end{array}
\end{equation*}
For simplicity, let us assume global identifiers, that is, the sets of identifiers of dimes and quarters have empty intersection.
The GDatalog$^{\naf s}[\Delta]$ program $\Pi$ over the above schema that encodes our simple scenario is as follows:
\begin{align*}
	\text{\rm Dime}(x) &\ra \text{\rm DimeTail}(x,\flip\params{0.5}[x])\\
	\text{\rm DimeTail}(x,1) &\ra \text{\rm SomeDimeTail}\\
	\text{\rm Quarter}(x), \naf\text{\rm SomeDimeTail} &\ra \text{\rm QuarterTail}(x,\flip\params{0.5}[x]).
\end{align*}
Hence, the TGD$^\naf$ program $\dep_\Pi$ is the following:
\begin{align*}
	\text{\rm Dime}(x) &\ra \text{\rm Active}_{1}^{\flip}(0.5,x)\\
	\text{\rm DimeTail}(x,1) &\ra \text{\rm SomeDimeTail}\\
	\text{\rm Quarter}(x), \naf \text{\rm SomeDimeTail} &\ra \text{\rm Active}_{1}^{\flip}(0.5,x)\\
	\text{\rm Active}_{1}^{\flip}(0.5,x) &\ra \exists z \, \text{\rm Result}_{1}^{\flip}(0.5,x,z)\\
	\text{\rm Result}_{1}^{\flip}(0.5,x,z), \text{\rm Dime}(x) &\ra \text{\rm DimeTail}(x,z) \\
	\text{\rm Result}_{1}^{\flip}(0.5,x,z), \text{\rm Quarter}(x), \naf \text{\rm SomeDimeTail} &\ra \text{\rm QuarterTail}(x,z).
\end{align*}

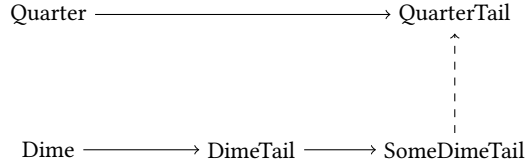
\begin{figure}
	\tikzstyle{node} = [text centered]
	\tikzstyle{line} = [draw, solid]
	\tikzstyle{arc} = [draw, solid, ->]
	\tikzstyle{darc} = [draw, dashed, ->]
	
	\begin{tikzpicture}[scale=.9]
		\node at (0,0) [node](d) {\text{\rm Dime}};
		\node at (0,2) [node](q) {\text{\rm Quarter}};
		\node at (3,0) [node](dt) {\text{\rm DimeTail}};
		\node at (6,0) [node](sdt) {\text{\rm SomeDimeTail}};
		\node at (6,2) [node](qt) {\text{\rm QuarterTail}};
		
		\draw (d) edge[arc] (dt);
		\draw (dt) edge[arc] (sdt);
		\draw (q) edge[arc] (qt);
		\draw (sdt) edge[darc] (qt);
		
		% \draw (c) edge[arc, out=45, in=-45] (bc);
		% \draw (a) edge[arc, out=135, in=-135] (ab);
	\end{tikzpicture}
	\caption{Dependency graph of the example program with stratified negation (negative arcs are dashed)}\label{fig:dg}
\end{figure}
The dependency graph $\dg{\Pi}$ is shown in Figure~\ref{fig:dg}.
In the following, let us consider the topological ordering
\[
C_1 = \{\text{\rm Dime}\} \quad C_2 = \{\text{\rm Quarter}\} \quad  C_3 = \{\text{\rm DimeTail}\} \quad C_4 = \{\text{\rm SomeDimeTail}\} \quad C_5 = \{\text{QuarterTail}\}
\]
and a database
\[
D\ =\ \{\text{\rm Dime(1)},\text{\rm Dime(2)}, \text{\rm Quarter(3)}\}.
\]
Note that $\Pi_{|C_1} \cup \Pi_{|C_2} = \{ \ra \alpha \mid \alpha \in D\}$,
$\Pi_{|C_3} = \{\text{\rm Dime}(x) \ra \text{\rm DimeTail}(x,\flip\params{0.5}[x])\}$,
$\Pi_{|C_4} = \{\text{\rm DimeTail}(x,1) \ra \text{\rm SomeDimeTail}\}$, and
$\Pi_{|C_5} = \{\text{\rm Quarter}(x), \naf\text{\rm SomeDimeTail} \ra \text{\rm QuarterTail}(x,\flip\params{0.5}[x])\}$.
Hence, the perfect grounding for $\Pi[D]$ works with $D$ and 
%the following TGD$^\naf$ programs:
\begin{align*}
	\dep_{\Pi_{|C_3}} = &
	\left\{
	\begin{array}{rl}
		\text{\rm Dime}(x) &\ra \text{\rm Active}_{1}^{\flip}(0.5,x)\\
		\text{\rm Active}_{1}^{\flip}(0.5,x) &\ra \exists z \, \text{\rm Result}_{1}^{\flip}(0.5,x,z)\\
		\text{\rm Result}_{1}^{\flip}(0.5,x,z), \text{\rm Dime}(x) &\ra \text{\rm DimeTail}(x,z)
	\end{array}
	\right\}\\
	\dep_{\Pi_{|C_4}} = &
	\left\{
	\begin{array}{rl}
		\text{\rm DimeTail}(x,1) &\ra \text{\rm SomeDimeTail}
	\end{array}
	\right\}\\
	\dep_{\Pi_{|C_5}} = &
	\left\{
	\begin{array}{rl}
		\text{\rm Quarter}(x), \naf \text{\rm SomeDimeTail} &\ra \text{\rm Active}_{1}^{\flip}(0.5,x)\\
		\text{\rm Active}_{1}^{\flip}(0.5,x) &\ra \exists z \, \text{\rm Result}_{1}^{\flip}(0.5,x,z)\\
		\text{\rm Result}_{1}^{\flip}(0.5,x,z), \text{\rm Quarter}(x), \naf \text{\rm SomeDimeTail} &\ra \text{\rm QuarterTail}(x,z)
	\end{array}
	\right\}.
\end{align*}

Let us first consider the case in which the first dime shows tail, and the second dime shows head, which is encoded by the following $\dep \in [2^{\ground{\dep_{\Pi}^{\exists}}}]_{=}$:
\begin{align*}
	\text{\rm Active}_{1}^{\flip}(0.5,1) &\ra \text{\rm Result}_{1}^{\flip}(0.5,1,1)\\
	\text{\rm Active}_{1}^{\flip}(0.5,2) &\ra \text{\rm Result}_{1}^{\flip}(0.5,2,0).
\end{align*}
The perfect grounding is obtained as follows:
\begin{align*}
	\dep \uparrow C_{2} = &\ \left\{
	\begin{array}{rl}
		&\ra \text{\rm Dime}(1)\\
		&\ra \text{\rm Dime}(2)\\
		&\ra \text{\rm Quarter}(3)
	\end{array}
	\right\}\\
	\dep \uparrow C_{3} = &\ \dep \uparrow C_{2} \cup \left\{
	\begin{array}{rl}
		\text{\rm Dime}(1) &\ra \text{\rm Active}_{1}^{\flip}(0.5,1)\\
		\text{\rm Result}_{1}^{\flip}(0.5,1,1), \text{\rm Dime}(1) &\ra \text{\rm DimeTail}(1,1) \\
		\text{\rm Dime}(2) &\ra \text{\rm Active}_{1}^{\flip}(0.5,2) \\
		\text{\rm Result}_{1}^{\flip}(0.5,2,0), \text{\rm Dime}(2) &\ra \text{\rm DimeTail}(2,0)
	\end{array}
	\right\}\\
	\dep \uparrow C_{4} = &\ \dep \uparrow C_{3} \cup \left\{
	\begin{array}{rl}
		\text{\rm DimeTail}(1,1) &\ra \text{\rm SomeDimeTail}
	\end{array}
	\right\}\\
	\dep \uparrow C_{5} = &\ \dep \uparrow C_{4} \cup \emptyset = \mi{GPerfect}_{\Pi[D]}(\dep).
\end{align*}
Note that $\mi{AtR}_\dep \comp \mi{GPerfect}_{\Pi[D]}(\dep)$, that is, $\dep \in \terminals{\mi{GPerfect}_{\Pi[D]}}$.
Moreover, $\dep \cup \mi{GPerfect}_{\Pi[D]}(\dep)$ belongs to $\mi{GPerfect}_{\Pi[D]}$-$\ground{\Pi[D]}$, and
$\sms{\dep \cup \mi{GPerfect}_{\Pi[D]}(\dep \cup \mi{GPerfect}_{\Pi[D]}(\dep))} = \{\{\text{\rm Dime}(1),$ $\text{\rm Dime}(2),$ $\text{\rm Quarter}(3),$ $\text{\rm Active}_{1}^{\flip}(0.5,1),$ $\text{\rm DimeTail}(1,1),$ $\text{\rm Active}_{1}^{\flip}(0.5,2),$ $\text{\rm DimeTail}(2,0),$ $\text{\rm SomeDimeTail}\}\} = \{\heads{\dep \cup \mi{GPerfect}_{\Pi[D]}(\dep)}\}$.

Let us now consider the case in which no dime shows tail, which is encoded by the following $\dep \in [2^{\ground{\dep_{\Pi}^{\exists}}}]_{=}$:
\begin{align*}
	\text{\rm Active}_{1}^{\flip}(0.5,1) &\ra \text{\rm Result}_{1}^{\flip}(0.5,1,0)\\
	\text{\rm Active}_{1}^{\flip}(0.5,2) &\ra \text{\rm Result}_{1}^{\flip}(0.5,2,0).
\end{align*}
The perfect grounding is obtained as follows:
\begin{align*}
	\dep \uparrow C_{2} = &\ \left\{
	\begin{array}{rl}
		&\ra \text{\rm Dime}(1)\\
		&\ra \text{\rm Dime}(2)\\
		&\ra \text{\rm Quarter}(3)
	\end{array}
	\right\}\\
	\dep \uparrow C_{3} = &\ \dep \uparrow C_{2} \cup \left\{
	\begin{array}{rl}
		\text{\rm Dime}(1) &\ra \text{\rm Active}_{1}^{\flip}(0.5,1)\\
		\text{\rm Result}_{1}^{\flip}(0.5,1,0), \text{\rm Dime}(1) &\ra \text{\rm DimeTail}(1,0) \\
		\text{\rm Dime}(2) &\ra \text{\rm Active}_{1}^{\flip}(0.5,2) \\
		\text{\rm Result}_{1}^{\flip}(0.5,2,0), \text{\rm Dime}(2) &\ra \text{\rm DimeTail}(2,0)
	\end{array}
	\right\}\\
	\dep \uparrow C_{4} = &\ \dep \uparrow C_{3} \cup \emptyset\\
	\dep \uparrow C_{5} = &\ \dep \uparrow C_{4} \cup
	\left\{
	\begin{array}{rl}
		\text{\rm Quarter}(3), \naf \text{\rm SomeDimeTail} &\ra \text{\rm Active}_{1}^{\flip}(0.5,3)
	\end{array}
	\right\}
	= \mi{GPerfect}_{\Pi[D]}(\dep).
\end{align*}
Note that $\mi{AtR}_\dep \comp \mi{GPerfect}_{\Pi[D]}(\dep)$ does not hold due to $\text{\rm Active}_{1}^{\flip}(0.5,3) \in \heads{\dep \cup \mi{GPerfect}_{\Pi[D]}(\dep)}\}$.

\subsection*{Proof of Proposition~\ref{pro:gperfect-grounder}}

Before giving the proof of Proposition~\ref{pro:gperfect-grounder}, we first need to establish following auxiliary lemma:

\begin{lemma}\label{lem:unique-stable-model}
	%Consider a GDatalog$^{\naf s}[\Delta]$ program $\Pi$. For every $\dep \in \mi{GPerfect}_{\Pi}\text{-}\ground{\dep}$, $\sms{\dep} = \{\heads{\dep}\}$.
	%
 	For every $\dep \in \terminals{\mi{GPerfect}_{\Pi}}$ such that there is no $\dep'\in \terminals{\mi{GPerfect}_{\Pi}}$ with $\dep' \subsetneq \dep$, $\sms{\dep \cup \mi{GPerfect}_{\Pi}(\dep)} = \{\heads{\dep \cup \mi{GPerfect}_{\Pi}(\dep)}\}$.
 	%($\mi{GPerfect}_{\Pi}\text{-}\ground{\dep}$ is defined only for grounders).}
\end{lemma}

\begin{proof}
For brevity, let $\dep^\star = \dep \cup \mi{GPerfect}_{\Pi}(\dep)$. Observe that $\dep^\star$ is a stratified ground program according to \cite{GeLi88}.
From Corollary~1 in \cite{GeLi88}, we know that $\dep^\star$ has a unique stable model.
We shall show that such a stable model is $\heads{\dep^\star}$, i.e., $\heads{\dep^\star}$ is a model of $\SM[\dep^\star] = \bigwedge_{\sigma \in \dep^\star} \sigma\ \wedge\ \neg \exists \ins{X}_{\dep^\star} \left( (\ins{X}_{\dep^\star} < \ins{R}) \wedge \repl_{\dep^\star}(\dep^\star)\right)$.
Actually, $\heads{\dep^\star} \models \bigwedge_{\sigma \in \dep^\star} \sigma$ is trivial, so it remains to show that there is no smaller model for the program obtained from $\dep^\star$ by fixing the interpretation of negative literals according to $\heads{\dep^\star}$ (such a program is called \emph{reduct} in the literature).

By Definition~\ref{def:grounder-gdatalog}, $\dep \in \terminals{\mi{GPerfect}_{\Pi}}$ implies $\mi{AtR}_\dep \comp \mi{GPerfect}_{\Pi}(\dep)$.
Moreover, by construction of $\dep^\star$ (Definition~\ref{def:perfect-grounder} and in particular of $\mi{Perfect}_{\dep}$), (i) all $\sigma \in \dep^\star$ satisfy $B^-(\sigma) \cap \heads{\dep^\star} = \emptyset$, and (ii) $\dep^\star$ is ordered so that all atoms in $B^+(\sigma)$ occur in some head of previously derived TGD$^\naf$.
Moreover, by assumption, (iii) there is no $\dep'\in \terminals{\mi{GPerfect}_{\Pi}}$ with $\dep' \subsetneq \dep$.
From (i) we have that the program reduct is equivalent to the program obtained from $\dep^\star$ by removing all negative literals (because such negative literals are true w.r.t.~$\heads{\dep^\star}$).
From (ii) and (iii), we have that $\heads{\dep^\star}$ is the subset-minimal model of the program reduct. This completes our proof.
\end{proof}

To prove that $\mi{GPerfect}_{\Pi}$ is a grounder of $\Pi$, we have to show that if $\dep \in [2^{\ground{\dep_{\Pi}^{\exists}}}]_{=}$ is such that $\mi{AtR}_\dep \comp G(\dep)$, and $\dep'$ is a totalizer of $\mi{AtR}_{\dep}$, then $\sms{\mi{GPerfect}_{\Pi}(\dep) \cup \dep} = \sms{\dep_{\Pi}^{\not\exists} \cup \dep'}$.
From Lemma~\ref{lem:unique-stable-model} we know that $\sms{\mi{GPerfect}_{\Pi}(\dep) \cup \dep} = \{\heads{\mi{GPerfect}_{\Pi}(\dep) \cup \dep}\}$.
From Corollary~1 in \cite{GeLi88} we know that $\sms{\dep_{\Pi}^{\not\exists} \cup \dep'} = \{I\}$ (because $\dep_{\Pi}^{\not\exists} \cup \dep'$ is stratified).
Hence, we have to show that $I = \heads{\mi{GPerfect}_{\Pi}(\dep) \cup \dep}$.
In fact, $\heads{\mi{GPerfect}_{\Pi}(\dep) \cup \dep}$ is a stable model of $\dep_{\Pi}^{\not\exists} \cup \dep'$ because 
$\mi{GPerfect}_{\Pi}(\dep) \cup \dep \subseteq \ground{\dep_{\Pi}^{\not\exists}} \cup \dep'$, and
all $\sigma \in (\ground{\dep_{\Pi}^{\not\exists}} \cup \dep') \setminus (\mi{GPerfect}_{\Pi}(\dep) \cup \dep)$ are such that $B^+(\sigma) \not\subseteq \heads{\mi{GPerfect}_{\Pi}(\dep) \cup \dep}$ or $B^-(\sigma) \cap \heads{\mi{GPerfect}_{\Pi}(\dep) \cup \dep} \neq \emptyset$.

\subsection*{Proof of Theorem~\ref{the:gperfect-best}}

According to Definition~\ref{def:compare-semantics}, $\Pi_{\mi{GPerfect}_{\Pi[D]}}(D)$ is as good as $\Pi_{G}(D)$ if
\begin{equation}\label{eq:prob-as-good-as}
P_{\Pi,\mi{GPerfect}_{\Pi[D]}}^{D}\left(\{\dep \in \Omega_{\Pi,\mi{GPerfect}_{\Pi[D]}}^{\fin}(D) \mid \sms{\dep} = \mathcal{I}\}\right)\ \geq 
P_{\Pi,G}^{D}\left(\{\dep \in \Omega_{\Pi,G}^{\fin}(D) \mid \sms{\dep} = \mathcal{I}\}\right)
\end{equation}
for every $\mathcal{I} \subseteq \sms{\dep_{\Pi[D]}}$.
Expanding \eqref{eq:prob-as-good-as} with Definition~\ref{def:possible-outcome} and Definition~\ref{def:grounder-gdatalog}, we have
\begin{equation}\label{eq:prob-as-good-as:2}
	\begin{split}
	P_{\Pi,\mi{GPerfect}_{\Pi[D]}}^{D}\left(\{\dep \cup \mi{GPerfect}_{\Pi[D]}(\dep) \in \Omega_{\Pi,\mi{GPerfect}_{\Pi[D]}}^{\fin}(D) \mid \sms{\dep \cup \mi{GPerfect}_{\Pi[D]}(\dep)} = \mathcal{I}\}\right)\\ 
	{}\geq
	P_{\Pi,G}^{D}\left(\{\dep \cup G(\dep) \in \Omega_{\Pi,G}^{\fin}(D) \mid \sms{\dep \cup G(\dep)} = \mathcal{I}\}\right)
	\end{split}
\end{equation}
We proceed to provide a construction showing that, for every $\dep \cup G(\dep) \in \Omega_{\Pi,G}^{\fin}(D)$, there exists $\dep' \subseteq \dep$ such that $\dep' \cup \mi{GPerfect}_{\Pi[D]}(\dep') \in \Omega_{\Pi,\mi{GPerfect}_{\Pi[D]}}^{\fin}(D)$.
Let $\dep_0 = \emptyset$, and for $i \geq 1$,
\[
	%\text{\rm Active}_{|\bar q|}^{\delta}(\bar p,\bar q) \ra \exists y\ \textrm{Result}_{|\bar q|}^{\delta}(\bar p,\bar q,y)
	\dep_i\ =\ \{\sigma \in \dep \mid B^+(\sigma) \subseteq \mi{GPerfect}_{\Pi[D]}(\dep_{i-1})\}.
\]
The least fixed point of the above sequence is reached in at most $|\dep|$ steps (because the sequence in increasing, and $\dep_i \subseteq \dep$ for all $i \geq 0$).
Let $\dep'$ be such a least fixed point.
By applying Proposition~\ref{pro:gperfect-grounder} and Definition~\ref{def:grounder-gdatalog}, we derive $\sms{\dep' \cup \mi{GPerfect}_{\Pi[D]}(\dep')} = \sms{\dep \cup G(\dep)} = \mathcal{I}$.
Finally, from $\dep' \subseteq \dep$ we immediately get that $\mi{Pr}(\dep' \cup \mi{GPerfect}_{\Pi[D]}(\dep')) \geq \mi{Pr}(\dep \cup G(\dep))$, and \eqref{eq:prob-as-good-as:2} follows.

%\subsection*{Proof of Theorem~\ref{the:stable-invariance-perfect}}
%\fi

\end{document}